\documentclass[twocolumn]{svjour3}
\smartqed

\usepackage{graphicx}
\usepackage{amsmath,amssymb,amsfonts,mathtools}
\usepackage{bm}
\usepackage{booktabs}
\usepackage{subfig}
\usepackage{algorithm}
\usepackage[noend]{algpseudocode}
\usepackage[hidelinks]{hyperref}

\makeatletter
\@ifundefined{cl@chapter}{}{%
  \let\cl@chapter\@empty
}
\makeatother

\usepackage[nameinlink,noabbrev]{cleveref}

\crefname{equation}{Equation}{Equations}
\Crefname{equation}{Equation}{Equations}

\crefname{figure}{Figure}{Figures}
\Crefname{figure}{Figure}{Figures}

\crefname{table}{Table}{Tables}
\Crefname{table}{Table}{Tables}

\crefname{algorithm}{Algorithm}{Algorithms}
\Crefname{algorithm}{Algorithm}{Algorithms}

\crefname{section}{Section}{Sections}
\Crefname{section}{Section}{Sections}

\crefname{appendix}{Appendix}{Appendices}
\Crefname{appendix}{Appendix}{Appendices}

\crefname{theorem}{Theorem}{Theorems}
\Crefname{theorem}{Theorem}{Theorems}
\crefname{lemma}{Lemma}{Lemmas}
\Crefname{lemma}{Lemma}{Lemmas}
\crefname{definition}{Definition}{Definitions}
\Crefname{definition}{Definition}{Definitions}
\crefname{proposition}{Proposition}{Propositions}
\Crefname{proposition}{Proposition}{Propositions}
\crefname{corollary}{Corollary}{Corollaries}
\Crefname{corollary}{Corollary}{Corollaries}

\usepackage{lmodern}

\usepackage[most]{tcolorbox}

\newif\ifshowchanges
\showchangesfalse

\definecolor{revone}{HTML}{E67E22}
\definecolor{revtwo}{HTML}{1F77B4}
\definecolor{revthree}{HTML}{2CA02C}

\ifshowchanges
    \DeclareRobustCommand{\Rone}[1]{{\color{revone}#1}}
    \DeclareRobustCommand{\Rtwo}[1]{{\color{revtwo}#1}}
    \DeclareRobustCommand{\Rthree}[1]{{\color{revthree}#1}}
\else
    \DeclareRobustCommand{\Rone}[1]{#1}
    \DeclareRobustCommand{\Rtwo}[1]{#1}
    \DeclareRobustCommand{\Rthree}[1]{#1}
\fi

\journalname{The VLDB Journal}

\graphicspath{{.}}

\renewcommand{\makeheadbox}{}
\begin{document}

\title{Fast Construction of Partitioned Learned Bloom Filter with Theoretical Guarantees}

\titlerunning{Fast Construction of PLBF with Theoretical Guarantees}

\author{Atsuki Sato \and Yusuke Matsui}

\authorrunning{Sato and Matsui}

\institute{A. Sato \and Y. Matsui \at Department of Information and Communication Engineering, Graduate School of Information Science and Technology, The University of Tokyo, Japan. \\
\email{\{a\_sato, matsui\}@hal.t.u-tokyo.ac.jp}}

\date{Accepted for publication in The VLDB Journal.}

\maketitle

\begin{abstract}
Bloom filter is a widely used classic data structure for approximate membership queries.
Learned Bloom filters improve memory efficiency by leveraging machine learning, with the partitioned learned Bloom filter (PLBF) being among the most memory-efficient variants.
However, PLBF suffers from high construction cost, with time complexity $\mathcal{O}(N^3k)$, where $N$ and $k$ are hyperparameters. 
In this paper, we propose three methods, \textit{fast PLBF}, \textit{fast PLBF++}, and \textit{fast PLBF\#}, that reduce the construction complexity to $\mathcal{O}(N^2k)$, $\mathcal{O}(Nk \log N)$, and $\mathcal{O}(Nk \log k)$, respectively. 
Fast PLBF preserves the original PLBF structure and memory efficiency. 
Although fast PLBF++ and fast PLBF\# may have different structures, we theoretically prove they are equivalent to PLBF under ideal data distribution.
\Rone{Furthermore, for non-ideal scenarios, we introduce an instance-specific measure of deviation from the ideal distribution and use it to bound the resulting difference in memory efficiency between PLBF and fast PLBF++.}
\Rtwo{Experiments on real-world datasets demonstrate that, under the default setting ($N=1{,}000$ and $k=5$), fast PLBF, fast PLBF++, and fast PLBF\# are up to 58.7, 74.5, and 75.7 times faster to construct than original PLBF, respectively.
In an ablation with larger $k$, which enables the construction of more accurate PLBFs, the speedups reach 235, 865, and 871 times at $k=50$.}
Additionally, fast PLBF maintains the same data structure as PLBF, and fast PLBF++ and fast PLBF\# achieve nearly identical memory efficiency.
\end{abstract}

\section{Introduction}
Bloom filter is a memory-efficient probabilistic data structure for approximate membership queries.
While Bloom filters produce false positives due to hash collisions, they do not produce false negatives. Because of this favorable property, they are widely used in various memory-constrained applications such as networks, databases, etc. 
There is a trade-off between the size of the Bloom filter and the false positive rate (FPR); the lower the FPR, the larger the filter needed. 
According to Pagh et al.~\cite{pagh2005optimal}, a data structure that stores a set of $n$ elements with an FPR of $F$ needs at least $n \log_{2}(1/F)$ bits of memory. 
However, in practical applications, the dataset and/or query has a certain structure, and thus we may be able to achieve better efficiency than this theoretical limit.
Innovatively, Kraska et al. integrate a machine learning model with a Bloom filter and use this model as a prefilter, further reducing the space requirement. Such a Bloom filter is called a learned Bloom filter (LBF) and has been actively studied.
In particular, partitioned learned Bloom filter (PLBF)~\cite{vaidya2021partitioned} solves an optimization problem using dynamic programming to find the optimal configuration of the prefilter and the backup filters, which allows almost full utilization of the machine learning performance.

\Rthree{However, the construction time of PLBF can become a bottleneck in real-world applications.
In practical environments, the set of stored keys and the query distribution of non-keys often shift over time.
This data and model drift naturally degrades the FPR, making it necessary to periodically reconstruct the PLBF and, in some cases, retrain the underlying machine learning model.
Even when retraining is not required, reconstruction of the PLBF itself can be expensive:
the original PLBF requires $\mathcal{O}(N^3 k)$ time to solve its optimization problem, where $N$ and $k$ are hyperparameters of PLBF.
In our experiments, this construction can take several minutes under practical parameter settings, which may be problematic when filters must be refreshed frequently or when many filters must be maintained simultaneously.
Moreover, increasing $N$ and $k$ can improve the memory--FPR trade-off by allowing finer threshold choices and more score regions, but it further increases the construction cost and can make the original PLBF impractical.}

We propose three methods to speed up the construction of PLBF while maintaining their excellent memory efficiency: \textit{fast PLBF}, \textit{fast PLBF++}, and \textit{fast PLBF\#}. The time complexity for solving the optimization problem of fast PLBF, fast PLBF++, and fast PLBF\# is $\mathcal{O}(N^2 k)$, $\mathcal{O}(Nk \log N)$, and $\mathcal{O}(Nk\log k)$, respectively. 
Fast PLBF has the same data structure as PLBF, that is, it achieves the same memory efficiency as PLBF. Although fast PLBF++ and fast PLBF\# do not necessarily have the same data structure as PLBF, we theoretically guarantee that they have the same data structures as PLBF under an intuitive and ideal condition on the distribution. Furthermore, we theoretically bound the difference of memory efficiency between (fast) PLBF and fast PLBF++.

We evaluate our methods using real-world datasets and artificial datasets.
The results show that (i) \Rtwo{under the default setting ($N=1{,}000$ and $k=5$), fast PLBF, fast PLBF++, and fast PLBF\# can be constructed up to 58.7, 74.5, and 75.7 times faster than PLBF, respectively; in an ablation with larger $k$, which enables the construction of more accurate PLBFs, the speedups reach 235, 865, and 871 times at $k=50$}, (ii) fast PLBF achieves exactly the same memory efficiency as PLBF, and (iii) fast PLBF++ and fast PLBF\# achieve almost the same memory efficiency as PLBF, and the difference is within limits guaranteed by our theorem.

\begin{table*}[t]
    \caption{Comparison between PLBF~\cite{vaidya2021partitioned} and our proposed methods. Here, $N$ denotes the number of segments and $k$ the number of regions. The second and third rows correspond to our previously published conference results~\cite{sato2023fast}, while the fourth and fifth rows summarize the new contributions of this journal version, including improved time complexity and stronger guarantees for non-ideal score distributions.}
    \centering
    \begin{tabular}{@{}lll@{}}
        \toprule
        Method & Construction Time & Memory-Efficiency Comparison with PLBF\\
        \midrule
        PLBF & $\mathcal{O}(N^3 k)$ & - \\
        Fast PLBF & $\mathcal{O}(N^2 k)$ & Always identical \\
        Fast PLBF++ & $\mathcal{O}(Nk \log N + Nk^2)$ & Identical under ideal conditions \\
        & $\rightarrow \mathcal{O}(Nk \log N)$ & + Theoretical bound under non-ideal conditions \\
        Fast PLBF\# & $\mathcal{O}(Nk\log k)$ & Identical under ideal conditions \\
        \bottomrule
    \end{tabular}
    \label{tab:comparison}
\end{table*}

This study builds on our previous shorter conference paper~\cite{sato2023fast}, in which we proposed fast PLBF and fast PLBF++ and provided theoretical guarantees for fast PLBF++ only under ideal conditions.
We significantly extend this work and make several further contributions, as summarized in \cref{tab:comparison}.
First, we improve the worst-case time complexity of fast PLBF++ from $\mathcal{O}(Nk\log N + Nk^2)$ to $\mathcal{O}(Nk\log N)$ by modifying parts of the algorithm.
\Rone{Second, we extend the theoretical guarantee for the memory efficiency of fast PLBF++, providing a data- and model-dependent bound that remains valid even under non-ideal conditions.}
Third, we introduce fast PLBF\# for the first time in this paper, along with its theoretical guarantee. Fast PLBF\# achieves even lower time complexity than fast PLBF++.
A technical report version is available on arXiv~\cite{sato2024fast}.

The remainder of this paper is organized as follows.
In \cref{sec: related work}, we present the related work. 
In \cref{sec: preliminaries}, we give preliminaries consisting of two parts: (1) an overview of PLBF, including its architecture, the optimization problem it solves, and (2) the \textit{matrix problem} and the related algorithms used for our acceleration.
In \cref{sec: methods}, we propose our three methods to speed up the construction of PLBF.
In \cref{sec: experiments}, we present extensive experiments that demonstrate these methods' effectiveness and our theorems' validity.

\section{Related Work}
\label{sec: related work}
Our work is in the context of LBF, which uses machine learning to improve Bloom filter performance.
We first review Bloom filters and their variants in \cref{sec: related work: Bloom filter}, and then present related work on LBF in \cref{sec: related work: lbf}.

\subsection{Bloom Filter and Its Variants}
\label{sec: related work: Bloom filter}
The Bloom filter~\cite{bloom1970space} is a traditional data structure that answers approximate membership queries. An approximate membership query is a query that determines whether a set $\mathcal{S}$ contains an element $q$, while allowing false positives with some probability. Bloom filters are often used in contexts such as networks~\cite{broder2004network,tarkoma2012theory,geravand2013bloom} and databases~\cite{chang2008bigtable,goodrich2011invertible,lu2012bloomstore}, where memory constraints are tight and a small rate of false positives is tolerable.

The Bloom filter is a data structure consisting of an $m$-bit bit array, $\bm{b}$, and $k$ independent hash functions, $h_1, h_2, \dots, h_k$, each of which maps an input to an integer value between $1$ and $m$.
To set up a Bloom filter, all bits of $\bm{b}$ are initially set to $0$, and then for each element $x \in \mathcal{S}$ and for each $i \in \{1, 2, \dots, k\}$, set the bit $b_{h_i(x)}$ to 1.
To check whether an element $q$ is in the set $\mathcal{S}$, the Bloom filter checks the bits at positions $h_1(q), h_2(q), \dots, h_k(q)$ in the bit array $\bm{b}$. If any of these bits is $0$, then $q$ is definitely not in the set $\mathcal{S}$ (no false negatives occur). If all of these bits are $1$, the Bloom filter indicates that $q$ is in the set $\mathcal{S}$. However, due to the possibility of hash collisions, this may not be correct, leading to a false positive. The number of bits the Bloom filter requires to achieve an FPR of $F \in (0, 1)$ is 
\begin{equation}
    (\log_2 e) \cdot |\mathcal{S}| \log_2\frac{1}{F}.
\end{equation}
It is known that the lower bound for the number of bits required for a given FPR of $F$ is $|\mathcal{S}| \log_2 \left(1/F\right)$, so this means that a Bloom filter uses $\log_2 e \approx 1.44$ times the number of bits of the theoretical lower bound.

Several variants of Bloom filters have been proposed.
One of the most representative variants is the Cuckoo filter~\cite{fan2014cuckoo}.
The Cuckoo filter is more memory-efficient than the original Bloom filter and supports dynamic addition and removal of elements.
In addition, various derivatives such as Vacuum filter~\cite{wang2019vacuum}, Xor filter~\cite{graf2020xor}, and Ribbon filter~\cite{peter2021ribbon} have been proposed with the aim of achieving memory efficiency close to the theoretical lower bound and high query speeds.
A variant of the Bloom filter with this lower bound has also been proposed~\cite{pagh2005optimal}, but it is so complex that it is difficult to implement in practice.

\subsection{Learned Bloom Filter}
\label{sec: related work: lbf}

The concept of learned indexes, introduced by Kraska et al.~\cite{kraska2018case}, suggests that classical data structures can be enhanced by learned models that exploit data distribution.
The original learned Bloom filter (LBF)~\cite{kraska2018case} combines a learned classifier with a backup Bloom filter.
Let $\mathcal{S}$ be the set of keys to store and let $\mathcal{Q}$ be a set of non-keys used to approximate the distribution of negative queries (typically obtained from past queries), where $\mathcal{S}\cap\mathcal{Q}=\varnothing$.
The learned model is trained as a binary classifier using positives from $\mathcal{S}$ and negatives from $\mathcal{Q}$.
At query time, the LBF first evaluates the model. If the model predicts that the query $q$ is in $\mathcal{S}$, the LBF reports that $q$ is present.
Otherwise, it consults the backup Bloom filter and returns the Bloom filter's output.
This hybrid design maintains Bloom filters' no-false-negative guarantee while reducing total memory usage.

Several works improve this design by more effectively utilizing the model output.
Mitzenmacher~\cite{mitzenmacher2018model} proposed the Sandwiched LBF, which inserts a prefilter before model inference and optimizes memory allocation between the prefilter and the backup filter.
Ada-BF~\cite{dai2020adaptive} and PLBF~\cite{vaidya2021partitioned} further exploit the model score by partitioning the score space into multiple regions and assigning different backup Bloom filters to different regions.
Among these variants, PLBF is one of the most memory-efficient designs, as it optimally allocates the false positive rates across score regions via dynamic programming.
We provide a detailed description of PLBF in \cref{sec: preliminaries PLBF}.

Another line of research studies the choice or composition of the machine learning model itself~\cite{fumagalli2022choice,dai2022optimizing,malchiodi2024role,lin2025ensemble}.
These works evaluate different model classes and configurations, and some consider ensemble-based designs, providing empirical insights into how model capacity and data characteristics affect the memory--FPR trade-offs.
While such approaches are important, our work focuses on a complementary setting in which the learned model is treated as fixed in advance (e.g., due to deployment constraints or reuse of an existing classifier) and we optimize the filter construction on top of the given model scores.

\section{Preliminaries}
\label{sec: preliminaries}

In this section, we first give a detailed explanation of PLBF, one of the state-of-the-art LBFs (\cref{sec: preliminaries PLBF}). 
Next, we introduce the \textit{matrix problem} and the classes of matrices related to the problem (\cref{sec: preliminaries MP}).

\subsection{Partitioned Learned Bloom Filter}
\label{sec: preliminaries PLBF}

\begin{figure}[t]
    \centering
    \includegraphics[width=\columnwidth]{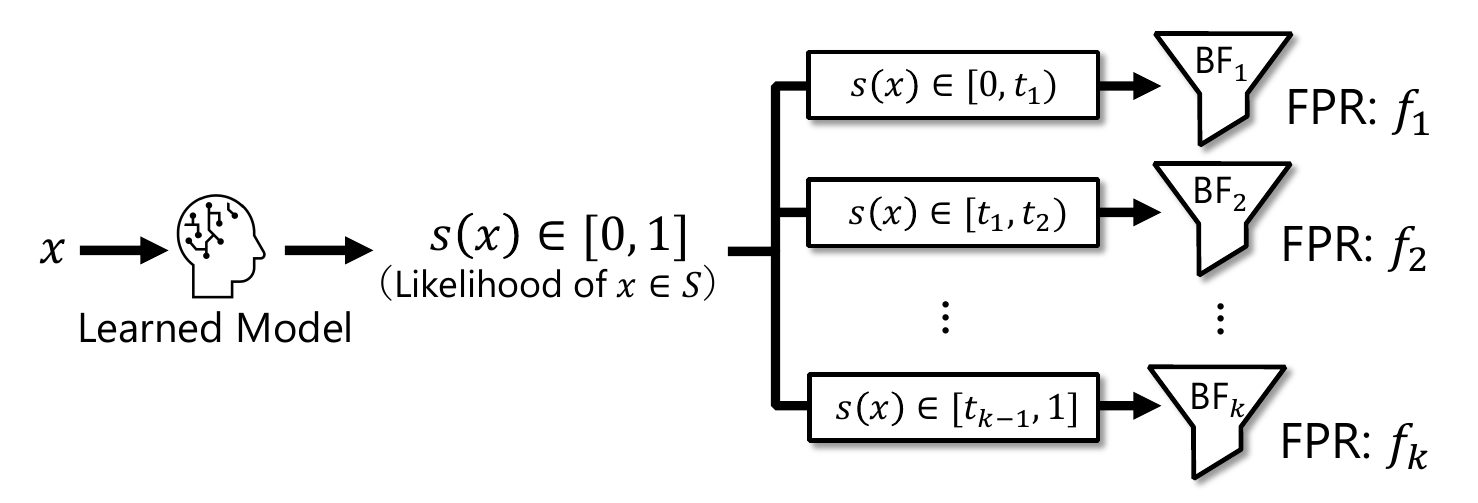}
    \caption{PLBF partitions the score space into $k$ regions and assigns backup Bloom filters with different FPRs to each region.}
    \label{fig: PLBF_idea}
\end{figure}

\begin{figure}[t]
    \centering
    \includegraphics[width=\columnwidth]{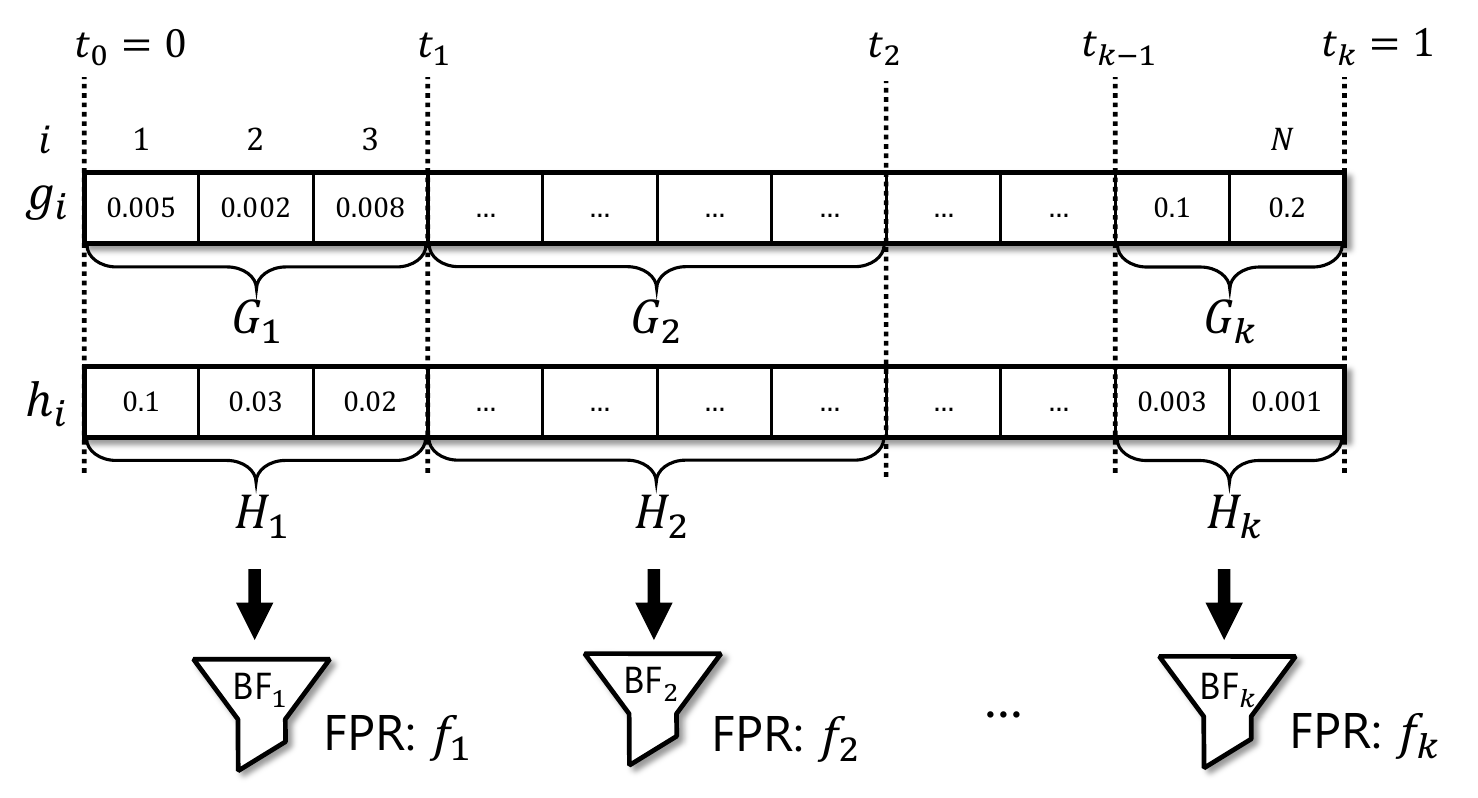}
    \caption{PLBF divides the score space into $N$ \Rone{equal-width} segments and then clusters the $N$ segments into $k$ regions. PLBF uses dynamic programming to find the optimal way to cluster segments into regions.}
    \label{fig: PLBF_division_score_space}
\end{figure}

First, we define terms to describe PLBF~\cite{vaidya2021partitioned} and explain its design.
Let $\mathcal{S}$ be a set of elements for which the Bloom filter is to be built, and let $\mathcal{Q}$ be the set of elements not included in $\mathcal{S}$ that is used when constructing PLBF ($\mathcal{S} \cap \mathcal{Q} = \varnothing$).
The elements included in $\mathcal{S}$ are called keys, and those not included are called non-keys.
$\mathcal{Q}$ is employed to approximate the distribution of non-key queries.
To build PLBF, a machine learning model is trained to predict whether a given element $x$ is included in the set $\mathcal{S}$ or $\mathcal{Q}$.
For a given element $x$, the machine learning model outputs a score $s(x) \in [0,1]$.
The score $s(x)$ indicates ``how likely is $x$ to be included in the set $\mathcal{S}$.''
PLBF partitions the score space $[0,1]$ into $k$ regions and assigns backup Bloom filters with different FPRs to each region (\cref{fig: PLBF_idea}).
Given a target overall FPR, $F \in (0,1)$, PLBF optimizes \Rone{$\bm{t} \in [0,1]^{k+1}$ and $\bm{f} \in (0,1]^{k}$} to minimize the total memory usage.
Here, $\bm{t}$ is a vector of thresholds for partitioning the score space into $k$ regions, and $\bm{f}$ is a vector of FPRs for each region, satisfying $t_0 = 0 < t_1 < t_2 < \dots < t_k = 1$ and $0 < f_i \leq 1 ~ (i=1,2, \dots, k)$.

Next, we explain the optimization problem and its solution designed by PLBF.
PLBF finds values of $\bm{t}$ and $\bm{f}$ that minimize the total memory usage of the backup Bloom filter, under the condition that the expected value of the overall FPR is $F$.
\Rone{Here, we define $G_i$ and $H_i$ $(i=1,2,\dots,k)$ as the fractions of keys and non-keys, respectively, whose scores fall into the $i$-th region.
Specifically,
\begin{equation}
    G_i = \Pr[t_i \leq s(x) < t_{i+1} | x \in \mathcal{S}],
\end{equation}
\begin{equation}
    H_i = \Pr[t_i \leq s(x) < t_{i+1} | x \in \mathcal{Q}].
\end{equation}}
Once $\bm{t}$ is determined, $\bm{G}$ and $\bm{H}$ can be obtained using $\mathcal{S}$, $\mathcal{Q}$, and the learned model.
The overall expected FPR is $\sum_{i=1}^{k} H_i f_i$, and the backup Bloom filter uses $\sum_{i=1}^{k} c|\mathcal{S}|G_i \log_{2}(1/f_i)$ bits, where $c$ is a constant depending on the backup filter type.
Therefore, PLBF formulates the following optimization problem:
\begin{equation}
\label{equ:prob}
    \begin{aligned}
    & \underset{\bm{f}, \bm{t}} {\text{minimize}} && 
    \sum_{i=1}^{k} c |\mathcal{S}| G_i \log_{2}\left(\frac{1}{f_i}\right) \\
    &\text{subject to} && \left\{ \,
        \begin{aligned}
            & \sum_{i=1}^{k} H_i f_i \leq F \\
            & t_0 = 0 < t_1 < t_2 < \dots < t_k = 1 \\
            & f_i \leq 1 ~~~~~~ (i=1,2,\dots, k).
        \end{aligned}
        \right.
    \end{aligned}
\end{equation}
The original PLBF paper solves the problem by ignoring the constraint $f_i \leq 1$.
By introducing the index set $\mathcal{I}_{f=1} \coloneq \{\, i \mid f_i = 1 \,\}$, we obtain a more general expression for $\bm{f}$ (see \cref{app:SolutionOptProb} for details):
\begin{equation}
    f_i = 
    \begin{dcases}
    1 & (i \in \mathcal{I}_{f=1}) \\
    \frac{(F - H_{f=1})G_i}{(1 - G_{f=1})H_i} & (i \notin \mathcal{I}_{f=1}),
    \end{dcases}
\end{equation}
where
\begin{equation}
\label{equ:define_G_f=1_H_f=1}
    G_{f=1} = \sum_{i \in \mathcal{I}_{f=1}} G_i, ~~ H_{f=1} = \sum_{i \in \mathcal{I}_{f=1}} H_i.
\end{equation}
Then, the total backup Bloom filter memory usage, that is, the objective function of \cref{equ:prob}, is
\begin{equation}
\begin{aligned}
\label{equ:minterm}
& c |\mathcal{S}| (1-G_{f=1})\log_{2}\left(\frac{1-G_{f=1}}{F-H_{f=1}}\right) \\
&\quad -
c |\mathcal{S}| \sum_{i\notin\mathcal{I}_{f=1}} G_i \log_{2}\left(\frac{G_i}{H_i}\right).
\end{aligned}
\end{equation}

Next, we explain how PLBF finds the optimal $\bm{t}$ and $\bm{f}$.
PLBF divides the score space $[0,1]$ into $N(>k)$ \Rone{equal-width} segments and then finds the optimal $\bm{t}$ and $\bm{f}$ using dynamic programming (DP).
Deciding how to cluster $N$ segments into $k$ consecutive regions corresponds to determining the threshold $\bm{t}$ (\cref{fig: PLBF_division_score_space}).
We denote the probabilities that the key and non-key scores are contained in the $i$-th \textbf{segment} by $g_i$ and $h_i$, respectively.
Then, the probabilities that the key and non-key scores are contained in the $i$-th \textbf{region}, that is, $G_i$ and $H_i$, are then calculated by summing the $g$ and $h$ values of the segments within each region (for example, $g_1+g_2+g_3=G_1$ in \cref{fig: PLBF_division_score_space}).

PLBF makes the following assumption to find the optimal thresholds $\bm{t}$: for optimal $\bm{t}$ and $\bm{f}$, $\mathcal{I}_{f=1}=\varnothing$ or $\mathcal{I}_{f=1}=\{k\}$.
In other words, PLBF assumes that there is at most one region for which $f=1$, and if there is one, it is the region with the highest score output by the machine learning model.
Using this assumption, PLBF finds the optimal solution by trying all possible threshold values $t_{k-1}$.
That is, PLBF finds the thresholds $\bm{t}$ that minimize memory usage by solving the following problem for each $j=k,k+1, \dots, N$;
we find a way to cluster the $1$st to $(j-1)$-th segments into $k-1$ regions that maximizes
\begin{equation}
        \sum_{i=1}^{k-1} G_i \log_{2}\left(\frac{G_i}{H_i}\right).
\end{equation}

PLBF solves this problem by building a $j \times k$ DP table $\mathrm{DP}_\mathrm{KL}^{j}[p][q]$ ($p=0,1, \dots, j-1$ and $q=0,1, \dots, k-1$) for each $j=k,k+1, \dots, N$;
$\mathrm{DP}_\mathrm{KL}^{j}[p][q]$ denotes the maximum value of $\sum_{i=1}^{q} {G_i}\log_{2} \left(G_i / H_i \right)$ that can be obtained by clustering the $1$st to $p$-th segments into $q$ regions.
To construct PLBF, one must find a clustering method that achieves $\mathrm{DP}_\mathrm{KL}^{j}[j-1][k-1]$.
Initially, we set $\mathrm{DP}_\mathrm{KL}^{j}[0][0] = 0$, $\mathrm{DP}_\mathrm{KL}^{j}[p][0]=-\infty ~ (p = 1,2,\dots,j-1)$, and $\mathrm{DP}_\mathrm{KL}^{j}[0][q]=-\infty ~ (q = 1,2,\dots,k-1)$.
Then, for $1 \leq p \leq N-1$ and $ 1 \leq q \leq k-1$, $\mathrm{DP}_\mathrm{KL}^{j}[p][q]$ is computed recursively as
\begin{equation}
\begin{aligned}
    &~ \mathrm{DP}_\mathrm{KL}^{j}[p][q] \\
=&~ \max_{i=1,2, \dots, p} \left(
    \mathrm{DP}_\mathrm{KL}^{j}[i-1][q-1] + 
    d_\mathrm{KL}(i, p)
\right),
\end{aligned}
\end{equation}
where the function $d_\mathrm{KL}(i_l,i_r)$ is the following function defined for integers $i_l$ and $i_r$ satisfying $1 \leq i_l \leq i_r \leq N$:
\begin{equation}
    d_\mathrm{KL}(i_l,i_r) = \left(\sum_{i=i_l}^{i_r} g_i\right) \log_{2} \left(
    \frac{\sum_{i=i_l}^{i_r} g_i}{\sum_{i=i_l}^{i_r} h_i}
    \right).
\end{equation}
The time complexity to construct this DP table is $\mathcal{O}(j^2 k)$.
Then, by tracing the recorded transitions backward from $\mathrm{DP}_\mathrm{KL}^{j}[j-1][k-1]$, we obtain the best clustering with a time complexity of $\mathcal{O}(k)$.
As the DP table is constructed for each $j=k,k+1, \dots, N$, the overall complexity is $\mathcal{O}(N^3k)$.
We can divide the score space more finely with a larger $N$ and thus obtain a $\bm{t}$ closer to the optimum.
However, the time complexity increases rapidly with increasing $N$.

The worst case time complexity of \textsc{OptimalFPR} is $\mathcal{O}(k^2)$ because this algorithm repeatedly solves the relaxation problem up to $k$ times with a complexity of $\mathcal{O}(k)$ until there are no more $i$s such that $f_i > 1$.
The time complexity of \textsc{SpaceUsed} is $\mathcal{O}(k)$.
As the DP table is constructed with a time complexity of $\mathcal{O}(j^2 k)$ for each $j=k,k+1, \dots, N$, the overall complexity is $\mathcal{O}(N^3k)$.
We provide the pseudo-code for the algorithm of PLBF construction and the \textsc{OptimalFPR} in \cref{app:AlgorithmDetails}.

\subsection{Matrix Problem}
\label{sec: preliminaries MP}

We define the \textit{matrix problem}, \textit{monotone matrix}, and \textit{totally monotone matrix} following Aggarwal et al.~\cite{aggarwal1987geometric}.
\begin{definition}[\textit{matrix problem}]
\label{def: matrix problem}
Let $B$ be an $n \times m$ real matrix, and let $J: \{1, 2, \dots, n\} \to \{1, 2, \dots, m\}$ be a function defined on $B$. For each row index $i$, $J(i)$ denotes the smallest column index $j$ such that $B_{i, j}$ is the maximum value in the $i$-th row of $B$.
Finding the value of $J(i)$ for all $i \in \{1,2,\dots,n\}$ is called a \textit{matrix problem}.
\end{definition}
\begin{definition}[\textit{monotone matrix}]
\label{def: monotone matrix}
An $n \times m$ real matrix $B$ is called a \textit{monotone matrix} if $J(i) \leq J(i')$ for any $i$ and $i'$ that satisfy $1 \leq i < i' \leq n$.
\end{definition}
\begin{definition}[\textit{totally monotone matrix}]
An $n \times m$ real matrix $B$ is called a \textit{totally monotone matrix} if every submatrix of $B$ is monotone.
\end{definition}
It is known that the definition of \textit{totally monotone matrix} is equivalent to the following: all $2 \times 2$ sub-matrices are monotone, that is, for any $1\leq i < i' \leq n$ and $1\leq j < j' \leq m$, $A_{i,j} < A_{i,j'} \Rightarrow A_{i',j} < A_{i',j'}$ holds.

Solving the \textit{matrix problem} for a general $n \times m$ matrix requires $\mathcal{O}(nm)$ time complexity because all matrix values must be checked.
Meanwhile, if the matrix is known to be a \textit{monotone matrix}, the \textit{matrix problem} for this matrix can be solved with time complexity of $\mathcal{O}(n + m \log n)$ using the monotone maxima algorithm~\cite{aggarwal1987geometric}.
In addition, the SMAWK algorithm can solve \textit{matrix problem} in $\mathcal{O}(n + m)$ for a \textit{totally monotone matrix}. This algorithm efficiently solves \textit{matrix problem} by recursively reducing the number of rows and columns to be considered, using the properties of totally monotone matrices to maintain the necessary elements.

\section{Fast Construction Methods}
\label{sec: methods}

In this section, we propose three methods for accelerating PLBF construction.
First, we show that the computation of building the DP table $\mathcal{O}(N)$ times in PLBF is redundant, and propose a method called fast PLBF that builds the DP table only one time and reuses it efficiently (\cref{sec: fast plbf}).
The time complexity of building PLBF is $\mathcal{O}(N^3 k)$, but fast PLBF is $\mathcal{O}(N^2 k)$.
Next, we propose fast PLBF++ (\cref{sec: fast plbf pp}).
This method further reduces the construction time to $\mathcal{O}(Nk\log N)$ by (i) replacing OptimalFPR with a method that has a smaller worst-case time complexity, and (ii) using the monotone maxima algorithm to speed up the construction of DP table. 
Finally, we propose fast PLBF\# (\cref{sec: fast plbf sharp}).
This algorithm reduces the construction time to $\mathcal{O}(Nk\log k)$ by using the SMAWK algorithm instead of the monotone maxima.
\Rone{Although fast PLBF++ and fast PLBF\# do not necessarily achieve the same data structure as PLBF, we guarantee that, under an ideal condition, fast PLBF++ and fast PLBF\# will have the same structure as PLBF. Additionally, we provide an instance-dependent guarantee for fast PLBF++ even in non-ideal conditions.}

\subsection{Fast PLBF}
\label{sec: fast plbf}

We propose fast PLBF,  which constructs the same data structure as PLBF more quickly by eliminating redundant DP-table construction.
Fast PLBF uses the same design as PLBF and finds the optimal $\bm{t}$ and $\bm{f}$ to minimize memory usage.

PLBF constructs a DP table for each $j=k,k+1, \dots, N$.
We found that this computation is redundant and that we can also use the last DP table $\mathrm{DP}_\mathrm{KL}^{N}$ for $j=k,k+1, \dots, N-1$.
This is because the maximum value of $\sum_{i=1}^{k-1} G_i \log_{2}\left(G_i / H_i \right)$ when clustering the $1$st to $(j-1)$-th segments into $k-1$ regions is equal to $\mathrm{DP}_\mathrm{KL}^{N}[j-1][k-1]$.
We can obtain the best clustering by tracing the transitions backward from $\mathrm{DP}_\mathrm{KL}^{N}[j-1][k-1]$.
The time complexity of tracing the transitions is $\mathcal{O}(k)$, which is faster than constructing the DP table.

The time complexity of building $\mathrm{DP}_\mathrm{KL}^{N}$ is $\mathcal{O}(N^2k)$, and the worst-case complexity of subsequent computations is $\mathcal{O}(Nk^2)$.
Because $N > k$, the total complexity is $\mathcal{O}(N^2k)$, which is faster than $\mathcal{O}(N^3k)$ for PLBF, although fast PLBF constructs the same data structure as PLBF.
The pseudo-code for fast PLBF construction is provided in \cref{app:AlgorithmDetails}.

\subsection{Fast PLBF++}
\label{sec: fast plbf pp}

We propose fast PLBF++, which further reduces the construction time beyond fast PLBF.
The acceleration of fast PLBF++ consists of two components.
First, fast PLBF++ improves the worst-case time complexity of \textsc{OptimalFPR}, the algorithm for finding the optimal FPR for each backup Bloom filter.
Second, fast PLBF++ accelerates the construction of the DP table $\mathrm{DP}_\mathrm{KL}^{N}$, that is, \textsc{MaxDivDP}, by taking advantage of a characteristic that DP tables often have.
We provide theoretical guarantees on the memory efficiency of fast PLBF++ under both ideal conditions and non-ideal conditions.

\subsubsection{Speed up \textsc{OptimalFPR}}
\label{sec: speed up optimalfpr}
We improve the worst-case complexity of PLBF's \textsc{OptimalFPR} from $\mathcal{O}(k^2)$ to $\mathcal{O}(k \log k)$.
The original \textsc{OptimalFPR} algorithm repeatedly solves the relaxation problem, where the condition $f_i \leq 1$ is ignored, until there are no more $i$ satisfying $f_i > 1$.
Since each iteration takes a complexity of $\mathcal{O}(k)$ and the repetition is at most $k$ times, the worst-case complexity of the original \textsc{OptimalFPR} algorithm is $\mathcal{O}(k^2)$.
However, by carefully analyzing the optimization problem, we can improve the worst-case complexity to $\mathcal{O}(k \log k)$.

The following theorem is essential for improving the efficiency of the \textsc{OptimalFPR} algorithm.
\begin{theorem}
\label{thm: optimal solution}
    There exists an optimal solution $\bm{f}$ to the optimization problem (\cref{equ:prob}) that satisfies
    \begin{equation}
        \frac{G_i}{H_i} < \frac{G_j}{H_j}
    \end{equation}
    for all $i, ~ j$ such that $f_i < 1, ~ f_j = 1$. (Here, when $H_i=0$, $G_i/H_i=\infty$, which is greater than any finite value.)
\end{theorem}
The proof is provided in \cref{app: Proof for Theorem 1}.
\cref{thm: optimal solution} indicates that only some $i$ with the highest $G_i / H_i$ are included in the set $\mathcal{I}_{f=1}$. 
As a result, there are at most $k$ potential candidate sets $\mathcal{I}_{f=1}$ that could form the optimal solution. 

To find the optimal set, we first sort the values of $G_i / H_i$ in descending order.
This sorting operation has a worst-case time complexity of $\mathcal{O}(k \log k)$.
Once sorted, we evaluate the objective function for each candidate set $\mathcal{I}_{f=1}$ sequentially.
We can evaluate \cref{equ:minterm} in $\mathcal{O}(1)$ time for each candidate set by maintaining and updating the following values: $G_{f=1} ~ (= \sum_{i \in \mathcal{I}_{f=1}} G_i)$, $H_{f=1} ~ (= \sum_{i \in \mathcal{I}_{f=1}} H_i)$, and $\sum_{i \notin \mathcal{I}_{f=1}} G_i \log_{2}\left(G_i / H_i \right)$.
Therefore, the total worst-case complexity is $\mathcal{O}(k \log k)$.

\Rtwo{We note, however, that this speedup should be viewed primarily as an asymptotic improvement.
As we will see in \cref{sec: experiments}, the improvement in the actual wall-clock construction time is modest for practical parameter settings.
This is because, under typical parameter settings, the overall construction time is often dominated by other components, such as DP computation, training the machine learning model, and inserting keys into the backup Bloom filters.
Therefore, the main contribution of this improvement is not a large reduction in construction time under standard settings, but rather a stronger asymptotic worst-case complexity guarantee for \textsc{OptimalFPR}.}

\subsubsection{Speed up \textsc{MaxDivDP}}
\label{sec: speed up maxdivdp}

\begin{figure}[t]
    \centering
    \subfloat[]{
        \includegraphics[width=0.46\columnwidth]{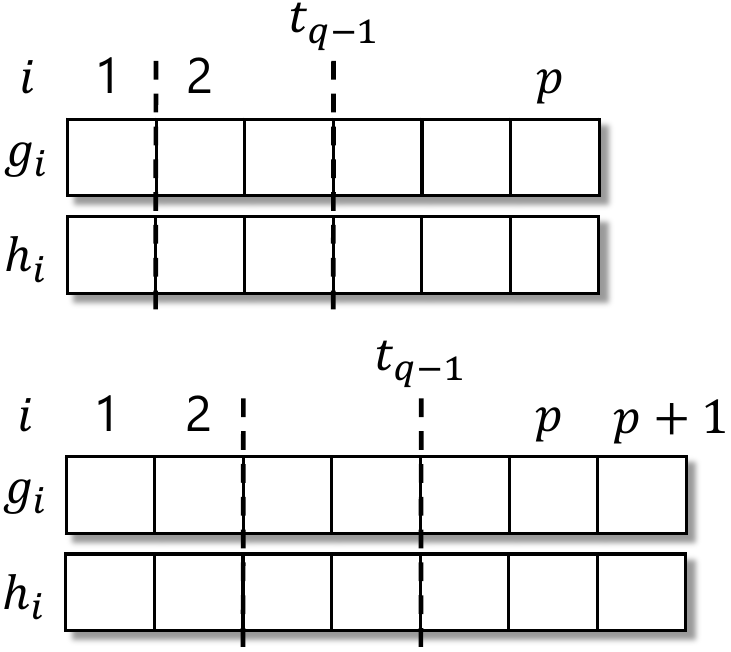}
        \label{fig: TransMonoIdea_monotone}
    }
    \subfloat[]{
        \includegraphics[width=0.46\columnwidth]{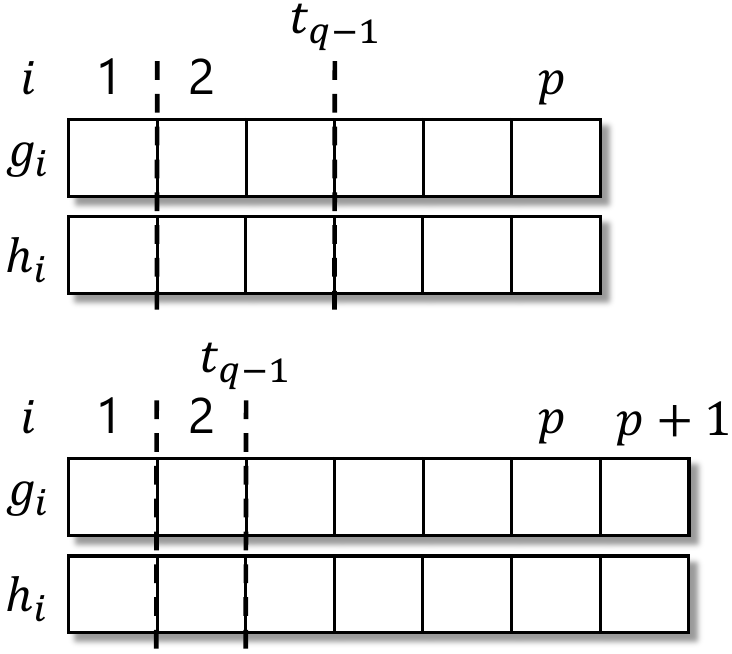}
        \label{fig: TransMonoIdea_not_monotone}
    }
    \caption{When the number of regions is fixed at $q$ and the number of segments is increased by $1$, the optimal $t_{q-1}$ remains unchanged or increases if $A$ is a \textit{monotone matrix} \protect\subref{fig: TransMonoIdea_monotone}. When $A$ is not a \textit{monotone matrix} \protect\subref{fig: TransMonoIdea_not_monotone}, the optimal $t_{q-1}$ may decrease.}
    \label{fig: TransMonoIdea}
\end{figure}

\begin{figure}[t]
    \centering
    \includegraphics[width=\columnwidth]{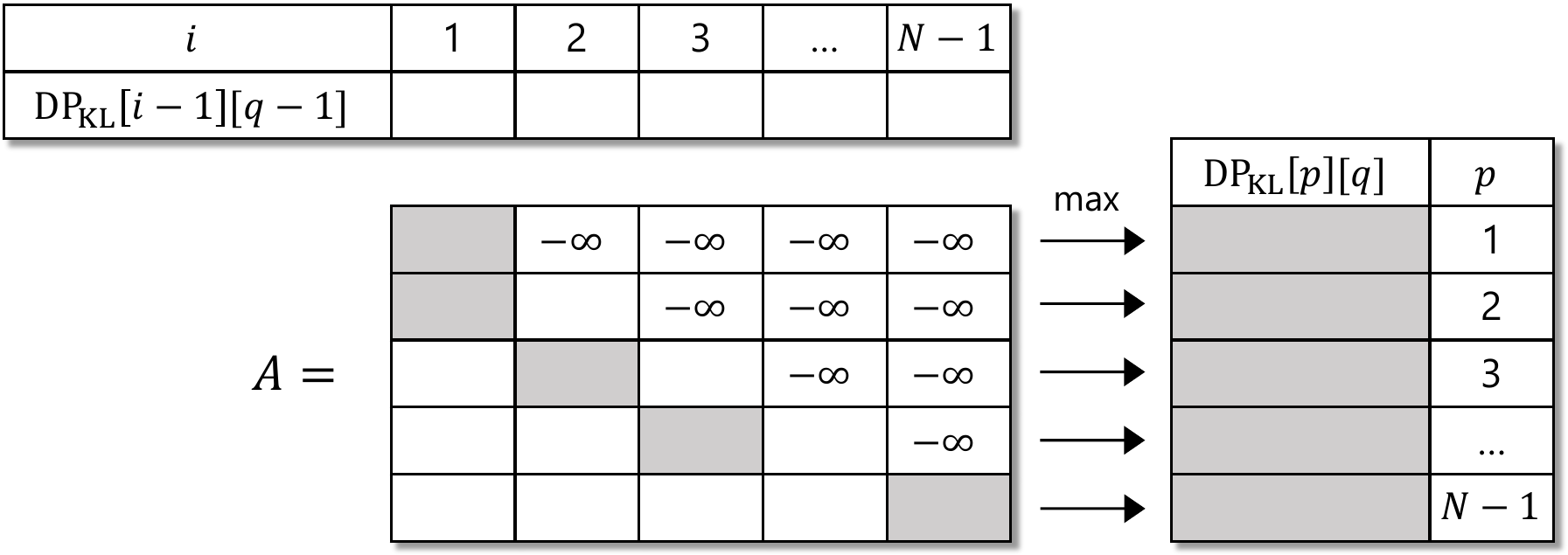}
    \caption{Computing $\mathrm{DP}_\mathrm{KL}[p][q] ~ (p=1,2, \dots, N-1)$ from $\mathrm{DP}_\mathrm{KL}[p][q-1] ~ (p=0,1, \dots, N-2)$ via the matrix $A$. The computation is the same as solving the \textit{matrix problem} for matrix $A$. When the score distribution is \textit{ideal}, $A$ is a \textit{monotone matrix}.}
    \label{fig: matr_prob_PLBF}
\end{figure}

Fast PLBF++ speed up \textsc{MaxDivDP} based on the following intuition;
it is more likely that when the number of regions is fixed at $q$ and the number of segments is increased by $1$, the optimal $t_{q-1}$ remains unchanged or increases, as in \cref{fig: TransMonoIdea_monotone}, than that $t_{q-1}$ decreases, as in \cref{fig: TransMonoIdea_not_monotone}.
Fast PLBF++ takes advantage of this characteristic to construct the DP table with less complexity $\mathcal{O}(Nk\log N)$.

First, we define the terms to describe fast PLBF++.
For simplicity, $\mathrm{DP}_\mathrm{KL}^{N}$ is denoted as $\mathrm{DP}_\mathrm{KL}$ in this section.
The $(N-1) \times (N-1)$ matrix $A$ is defined as follows:
\begin{equation}
\label{equ:ADif}
    A_{p,i} = \begin{dcases}
        \mathrm{DP}_\mathrm{KL}[i-1][q-1] + 
        d_\mathrm{KL}(i, p) & (i \leq p) \\
        - \infty & (\mathrm{else}).
    \end{dcases}
\end{equation}
Then, from the definition of $\mathrm{DP}_\mathrm{KL}$, 
\begin{equation}
    \label{equ:DP_A}
    \mathrm{DP}_\mathrm{KL}[p][q] = \max_{i=1,2, \dots, N-1} A_{p,i}.
\end{equation}
The matrix $A$ represents the intermediate calculations involved in determining $\mathrm{DP}_\mathrm{KL}[p][q]$ from $\mathrm{DP}_\mathrm{KL}[p][q-1]$ (\cref{fig: matr_prob_PLBF}). 
When the intuition that $t_{q-1}$ does not decrease, as discussed at the beginning of this section, holds, $A$ is a \textit{monotone matrix}.

While fast PLBF checks almost all elements of the matrix $A$ and thus takes $\mathcal{O}(N^2)$ computations, fast PLBF++ \Rone{reduces the number of checked elements by using the monotone maxima algorithm, thereby reducing the computational complexity to $\mathcal{O}(N \log N)$}. This is much faster than PLBF and fast PLBF, but \Rone{it may introduce errors}, so it does not necessarily achieve the same optimal memory efficiency as those methods.

\subsubsection{Theoretical Guarantees Under Ideal Conditions}
\label{sec: fast plbf pp ideal}

Here, we demonstrate that the memory efficiency of fast PLBF++ is equivalent to that of PLBF under ideal conditions. The acceleration of fast PLBF++ is attributed to two factors. The first factor, speed up \textsc{OptimalFPR}, involves no approximation, so we focus on proving that the second factor, speed up \textsc{MaxDivDP}, also preserves the solution under ideal conditions.

First, we define an \textit{ideal score distribution} as follows.
\begin{definition}
A score distribution is called \textit{ideal} if
\begin{equation}
    \frac{g_1}{h_1} \leq \frac{g_2}{h_2} \leq \dots \leq \frac{g_N}{h_N}.
\end{equation}
\end{definition}
An \textit{ideal score distribution} implies that the probability of $x \in \mathcal{S}$ and the score $s(x)$ are ideally well correlated.
In other words, an \textit{ideal score distribution} means that the machine learning model learns the distribution ideally.

Fast PLBF++ does not necessarily have the same data structure as PLBF because $A$ is not necessarily a \textit{monotone matrix}.
However, as the following theorem shows, we can prove that $A$ is a \textit{monotone matrix} and thus fast PLBF++ has the same data structure as (fast) PLBF under certain conditions.

The following lemma highlights a property of the dynamic programming formulation used in PLBF construction under an \textit{ideal} score distribution.
This property forms the foundation of the performance guarantees for fast PLBF++.
\begin{lemma}
\label{lem: dkl_monge}
If the score distribution is \textit{ideal}, then for all $1 \leq i < i'\leq p < p' \leq N$, it holds that
\begin{equation}
\label{equ:lemma target}
\left(d_\mathrm{KL}(i, p) + d_\mathrm{KL}(i', p')\right) - \left(d_\mathrm{KL}(i, p') + d_\mathrm{KL}(i', p)\right) \geq 0.
\end{equation}
\end{lemma}
The proof of this lemma is provided in the appendix.
Using \cref{lem: dkl_monge}, we can prove the following theorem, which demonstrates that the memory efficiency of fast PLBF++ is equivalent to that of PLBF when the score distribution is \textit{ideal}.
\begin{theorem}
\label{thm: DPargTotallyMonotone}
If the score distribution is \textit{ideal}, $A$ is a \textit{totally monotone matrix}.
\end{theorem}
\begin{proof}
For any $i,i',p,p'$ satisfying $1\leq i < i' \leq p < p' \leq N-1$,
\begin{align}
    &~ \left(A_{p,i} + A_{p',i'}\right) - \left(A_{p',i} + A_{p,i'}\right) \\
=&~ \bigl(\left(\mathrm{DP}_\mathrm{KL}[i-1][q-1] +  d_\mathrm{KL}(i, p)\right) \nonumber\\
&~~~~~~ + \left(\mathrm{DP}_\mathrm{KL}[i'-1][q-1] +  d_\mathrm{KL}(i', p')\right)\bigr) \nonumber\\
&~ - \bigl(\left(\mathrm{DP}_\mathrm{KL}[i-1][q-1] +  d_\mathrm{KL}(i, p')\right) \nonumber\\
&~~~~~~ + \left(\mathrm{DP}_\mathrm{KL}[i'-1][q-1] +  d_\mathrm{KL}(i', p)\right)\bigr) \\
=&~ \bigl(d_\mathrm{KL}(i, p) + d_\mathrm{KL}(i', p')\bigr) \nonumber\\
&~~~~~~ - \bigl(d_\mathrm{KL}(i, p') + d_\mathrm{KL}(i', p)\bigr) \\
\geq&~ 0.
\end{align}
The last inequality is from \cref{lem: dkl_monge}.
Thus, $A_{p,i} < A_{p,i'} \Rightarrow A_{p',i} < A_{p',i'}$ holds.
Consequently, it has been demonstrated that every $2 \times 2$ submatrix is monotone, thereby proving that $A$ is a \textit{totally monotone matrix}.
\end{proof}

\subsubsection{Theoretical Guarantees Under Non-Ideal Conditions}
\label{sec: fast plbf pp non ideal}

Here, we provide a theoretical guarantee for cases where the distribution is not \textit{ideal}.
Specifically, we offer an upper bound on the difference in $\mathrm{DP}_\mathrm{KL}$ between solving the matrix problem greedily and solving it efficiently using monotone maxima.

First, we define $\Delta: \{0, 1,\dots, N\} \times \{0, 1,\dots, N\} \rightarrow \mathbb{R}$ by, for $0 < p < p' < N$,
\begin{equation}
\begin{aligned}
\Delta(p,p') \coloneq
\max_{1 \leq i \leq i' \leq p}
\Bigl(
&- d_\mathrm{KL}(i,p) - d_\mathrm{KL}(i',p') \\
&\quad + d_\mathrm{KL}(i,p') + d_\mathrm{KL}(i',p)
\Bigr).
\end{aligned}
\end{equation}
In all other cases, $\Delta(p, p')$ is defined to be $0$.
From \cref{lem: dkl_monge}, when the distribution is \textit{ideal}, $\Delta(p, p') \leq 0$ for any $p,p'$.
Furthermore, we define $\delta: \{0, 1,\dots, N\} \rightarrow \mathbb{R}$ recursively as follows:
\begin{equation}
\delta(p) =
\begin{dcases}
    0 
    & (p = 0, N) \\[4pt]
    \begin{aligned}
        \max\bigl(0, & \delta(l) + \Delta(l, p), \\
                    &\quad \delta(r) + \Delta(p, r)\bigr)
    \end{aligned}
    & (\mathrm{else})
\end{dcases}
\end{equation}
In this definition, $l$ is the largest row index smaller than $p$ that has been processed before the $p$-th row by the monotone maxima algorithm (if no such row exists, we set $l = 0$).
Similarly, $r$ is the smallest row index greater than $p$ that has been processed before the $p$-th row by the monotone maxima algorithm (if no such row exists, we set $r = N$).

\Rtwo{To make this recursive definition more concrete, consider the special case where $N = 2^M$ for some $M \in \mathbb{N}$.
Then, for any $p \in \{1,\dots,N-1\}$, the corresponding indices $l$ and $r$ can be written as $l = p - 2^m$ and $r = p + 2^m$, where $m$ is the largest non-negative integer such that $p$ is divisible by $2^m$.
When $N$ is not a power of two, $l$ and $r$ are defined analogously according to the row-processing order of the monotone maxima algorithm.}

Then, we define $\delta_\mathrm{max}$ as follows:
\begin{equation}
    \delta_\mathrm{max} = \max_{p=1,2,\dots,N-1} \delta(p).
\end{equation}
$\delta_\mathrm{max}$ is a value that depends on the training data and $N$.
In particular, when the distribution is \textit{ideal}, $\delta(p) = 0$ for all $p$, and therefore, $\delta_\mathrm{max}=0$.

Here, the following theorem holds.
\begin{theorem}
\label{thm: diff DPKL}
    Let $\mathrm{DP}_\mathrm{true}$ be the $\mathrm{DP}_\mathrm{KL}$ in the case of (fast) PLBF, that is, when calculating greedily without using monotone maxima.
    Furthermore, let $\mathrm{DP}_\mathrm{mm}$ be the $\mathrm{DP}_\mathrm{KL}$ in the case of fast PLBF++, that is, when calculating with monotone maxima.
    Then, for all $p=1,2,\dots,N-1$ and $q=1,2,\dots,k-1$,
    \begin{equation}
        0 \leq \mathrm{DP}_\mathrm{true}[p][q] - \mathrm{DP}_\mathrm{mm}[p][q] \leq q \delta_\mathrm{max}.
    \end{equation}
\end{theorem}
The proof of \cref{thm: diff DPKL} is provided in \cref{app: Proof for Theorem 4}.
Since $\delta_\mathrm{max} = 0$ when the distribution is \textit{ideal}, the DP table generated in the fast PLBF++ construction is exactly the same as the DP table generated in the fast PLBF construction.
The conclusion from \cref{thm: DPargTotallyMonotone} covers only \textit{ideal} situations. However, \cref{thm: diff DPKL} theoretically bounds the difference even in \textit{non-ideal} situations by using $\delta_\mathrm{max}$, a value that, in a sense, indicates the deviation from the \textit{ideal} distribution.
\Rone{Since $\delta_\mathrm{max}$ must be computed from the specific training data and learned model, this should be interpreted as an instance-specific bound rather than a uniform worst-case approximation guarantee.}

\Rtwo{We also note that \cref{thm: diff DPKL} bounds the difference in the DP table values, rather than directly bounding the final memory usage or FPR.
However, if $\mathcal{I}_{f=1} = \varnothing$, the difference in the DP table values can be translated directly into the difference in memory usage between fast PLBF and fast PLBF++:}
\Rtwo{\begin{theorem}
\label{thm: diff DPKL to diff memory}
    If $\mathcal{I}_{f=1} = \varnothing$, the additional memory usage incurred by fast PLBF++ compared with fast PLBF is at most $c |\mathcal{S}| (k-1) \delta_\mathrm{max}$ bits, where $c$ is the constant determined by the backup filter type, as introduced in the definition of \cref{equ:prob}.
\end{theorem}}
\Rtwo{\begin{proof}
    When $\mathcal{I}_{f=1}=\varnothing$, we have $G_{f=1}=H_{f=1}=0$.
    Hence, \cref{equ:minterm} reduces to
    \begin{equation}
        c |\mathcal{S}| \log_2 \frac{1}{F} - c |\mathcal{S}| \sum_{i=1}^{k} G_i \log_{2}\left(\frac{G_i}{H_i}\right)
    \end{equation}
    Since both fast PLBF and fast PLBF++ enumerate all possible choices of the $k$-th region, the contribution of the $k$-th region is evaluated exactly in both methods.
    The only possible difference comes from the DP value for the remaining $k-1$ regions.
    By \cref{thm: diff DPKL}, for every fixed choice of the $k$-th region, the difference in the value of $\sum_{i=1}^{k} G_i \log_{2}\left(\frac{G_i}{H_i}\right)$ is at most $(k-1) \delta_\mathrm{max}$.
    Therefore, the additional memory usage incurred by fast PLBF++ compared with fast PLBF is at most $c |\mathcal{S}| (k-1) \delta_\mathrm{max}$ bits.
\end{proof}}
\Rtwo{A similar statement also holds for the alternative problem setting introduced in \cref{app: modification of the PLBF framework}, where the expected FPR is minimized under a fixed memory budget:}
\Rtwo{\begin{theorem}
\label{thm: diff DPKL to diff FPR}
    If $\mathcal{I}_{f=1} = \varnothing$, the expected FPR of fast PLBF++ is at most $2^{(k-1) \delta_\mathrm{max}}$ times that of fast PLBF under the same memory budget.
\end{theorem}}
\Rtwo{\begin{proof}
    From \cref{equ:expected FPR}, when $\mathcal{I}_{f=1}=\varnothing$, the expected FPR is given by $2^{-\beta}$, where $\beta$ is defined in \cref{equ:beta}.
    As in the proof of \cref{thm: diff DPKL to diff memory}, by \cref{thm: diff DPKL}, the difference in $\beta$ between fast PLBF and fast PLBF++ is at most $(k-1) \delta_\mathrm{max}$.
    Therefore, the expected FPR of fast PLBF++ is at most $2^{(k-1) \delta_\mathrm{max}}$ times that of fast PLBF.
\end{proof}}
\Rone{The URL, EMBER, DNA, and NIDS datasets used in our experiments have $\delta_\mathrm{max} = 0.028$, $0.010$, $0.002$, and $0.007$, respectively, when $N = 1024$.}
Considering $k \sim 10$, this implies that, if $\mathcal{I}_{f=1} = \varnothing$, the difference in memory usage between fast PLBF and fast PLBF++ is at most $0.25 \times \log_{2} (e) = 0.36$ bits per element.
Additionally, under the same memory budget, it implies that the difference in the FPR is at most $2^{0.25} = 1.18$ times.

\Rtwo{We note, however, that this is a conservative worst-case bound on the DP-value difference.
Indeed, our experiments show that the actual difference between fast PLBF and fast PLBF++ is often much smaller than this bound, and is sometimes zero; see \cref{sec: MemoryUsageAndFPR,ablation study for hyper-parameters}.
Therefore, the role of $\delta_\mathrm{max}$ should be understood as providing an instance-specific theoretical upper bound, rather than a tight estimate of the empirical performance gap.}

\subsection{Fast PLBF\#}
\label{sec: fast plbf sharp}

Fast PLBF\# uses the SMAWK algorithm instead of monotone maxima. 
In addition, as in fast PLBF++, it replaces \textsc{OptimalFPR} with \textsc{FastOptimalFPR}.
The DP using the SMAWK algorithm takes $\mathcal{O}(Nk)$ time, and for each of the $j = k, k+1,\dots, N$, determining the optimal FPR takes $\mathcal{O}(k\log k)$ operations, so the overall time taken is $\mathcal{O}(Nk) + \mathcal{O}(Nk \log k) = \mathcal{O}(Nk \log k)$.

The SMAWK algorithm is an algorithm that solves the problem of finding the maximum value of each row of a matrix.
While monotone maxima assumes that the matrix is a \textit{monotone matrix}, the SMAWK algorithm assumes that the matrix is a \textit{totally monotone matrix}.
The \textit{totally monotone} assumption is stronger than the \textit{monotone} assumption, that is, every \textit{totally monotone matrix} is a \textit{monotone matrix}.
Therefore, the SMAWK algorithm is an algorithm that relies on stronger assumptions than those of monotone maxima.
Although the SMAWK algorithm is more complicated than monotone maxima, it has less time complexity.

From \cref{thm: DPargTotallyMonotone}, $A$ is a \textit{totally monotone matrix} when the distribution is \textit{ideal}.
Therefore, although the SMAWK algorithm relies on stronger assumptions than monotone maxima, it satisfies the assumptions of both algorithms under the condition that the distribution is \textit{ideal}.
As we will see in the next section, fast PLBF\# achieves almost the same memory efficiency as (fast) PLBF, just as fast PLBF++ does.

\Rtwo{We emphasize that the theoretical guarantees for fast PLBF++ and fast PLBF\# are not symmetric.
Both methods are guaranteed to recover the same DP values as fast PLBF under an ideal score distribution, because the matrix $A$ is totally monotone in this case.
However, for non-ideal score distributions, the additional bound in \cref{thm: diff DPKL} applies only to fast PLBF++.
We do not provide an analogous non-ideal bound for fast PLBF\#.
The reason is that SMAWK relies on the stronger property of total monotonicity and recursively discards columns;
once total monotonicity is violated, controlling the effect of these discarded columns is substantially more difficult than in the monotone maxima procedure used by fast PLBF++.
Therefore, the non-ideal behavior of fast PLBF\# is evaluated empirically in \cref{sec: robustness of fast plbf pp}, rather than bounded theoretically.}

\section{\Rtwo{Experiments}}
\label{sec: experiments}

\begin{figure*}[t]
    \centering
    \includegraphics[width=\textwidth]{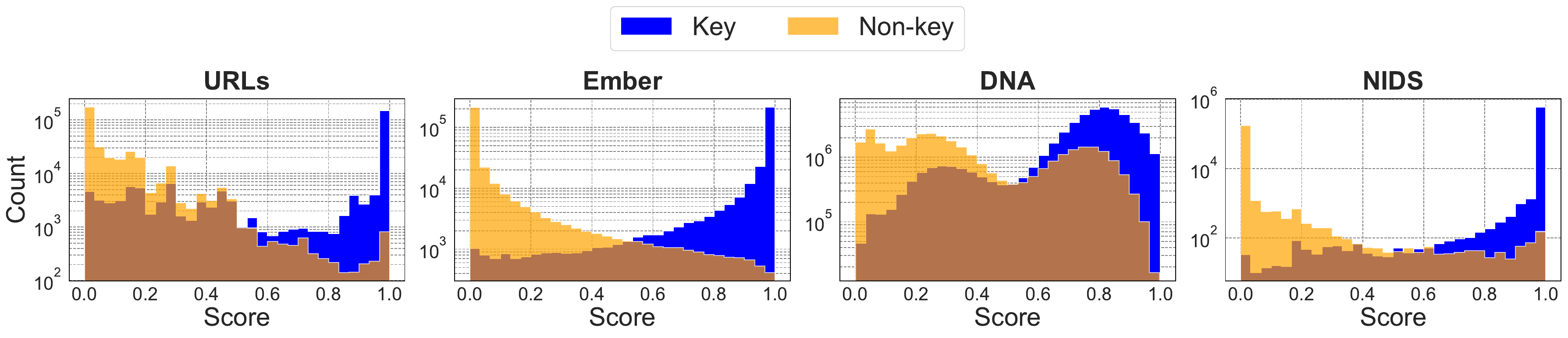}
    \caption{\Rone{Score distribution histograms of keys and non-keys produced by the classifier. Keys concentrate on higher-score regions while non-keys dominate lower-score regions.}}
    \label{fig: dist}
\end{figure*}

\begin{figure*}[t]
    \centering
    \includegraphics[width=\textwidth]{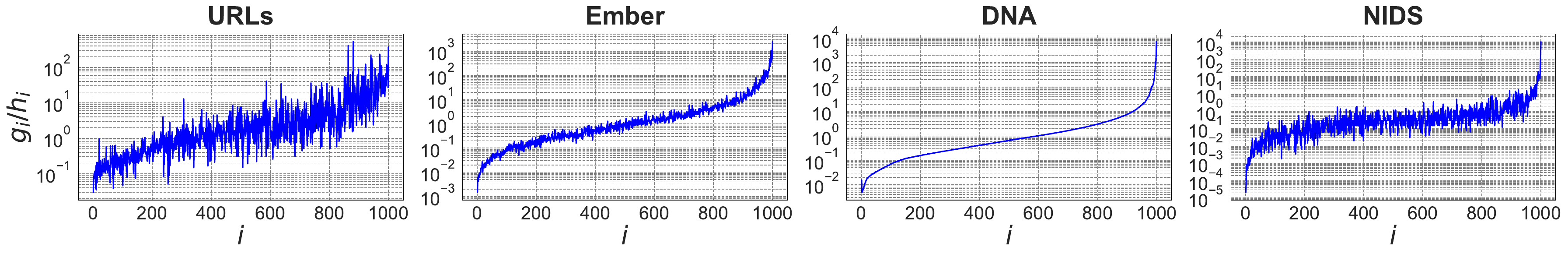}
    \caption{\Rone{Empirical ratio of key to non-key density, i.e., $g_i/h_i$, over score segments ($i=1,\dots,N$, $N=1{,}000$). While $g_i/h_i$ shows an overall increasing trend with $i$, it is not perfectly monotone and exhibits local fluctuations.}}
    \label{fig: g_to_h_ratio}
\end{figure*}

\begin{figure*}[t]
    \centering
    \includegraphics[width=\textwidth]{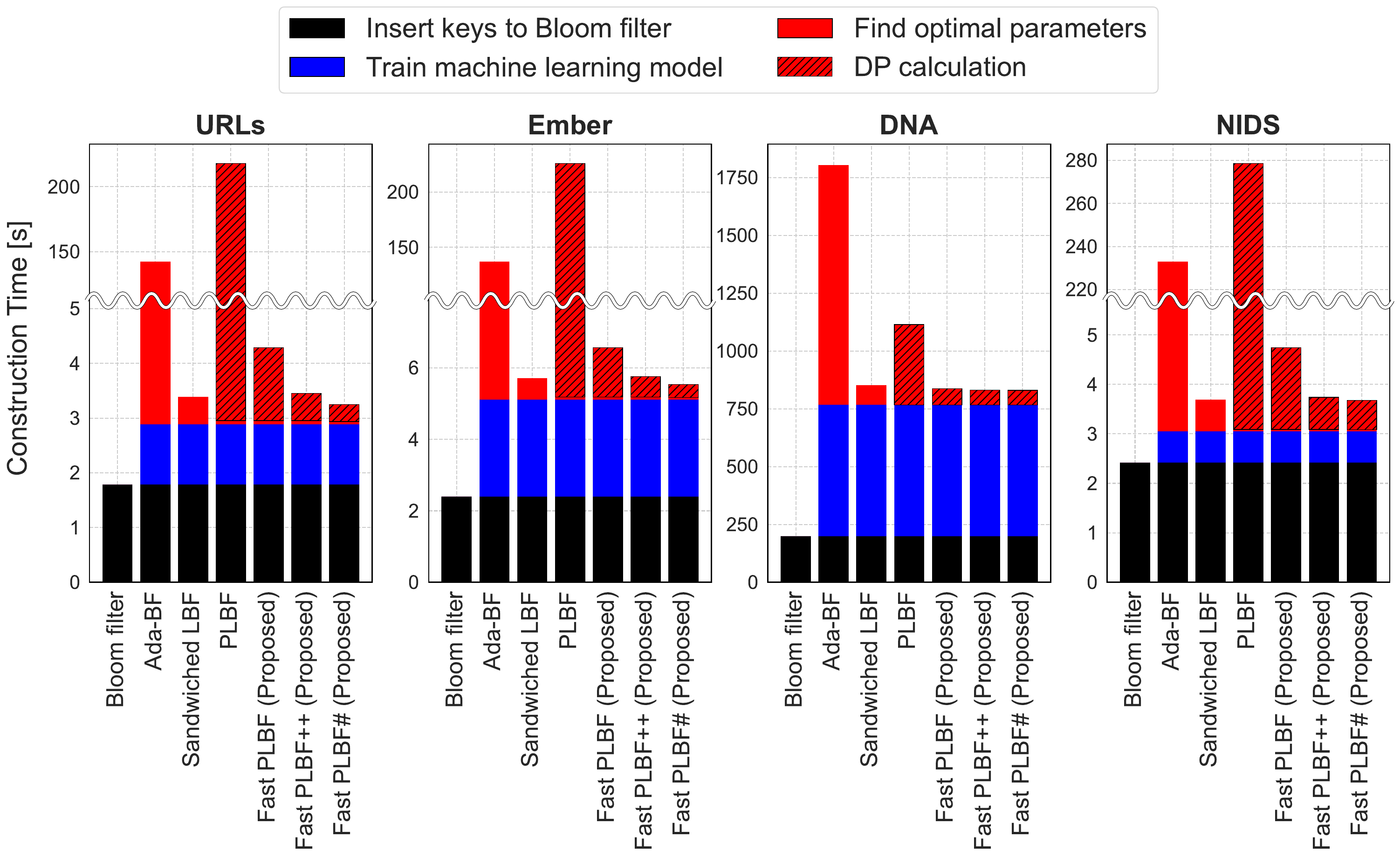}
    \caption{\Rone{End-to-end construction time. For LBF-based methods, the construction time includes model training, parameter optimization (i.e., thresholds $\bm{t}$ and FPRs $\bm{f}$), and inserting keys into the backup Bloom filters. The proposed fast constructions substantially reduce the computation time spent on PLBF parameter optimization, which is the dominant cost in PLBF construction.}}
    \label{fig: time_hist}
\end{figure*}

\begin{figure*}[t]
    \centering
    \includegraphics[width=\textwidth]{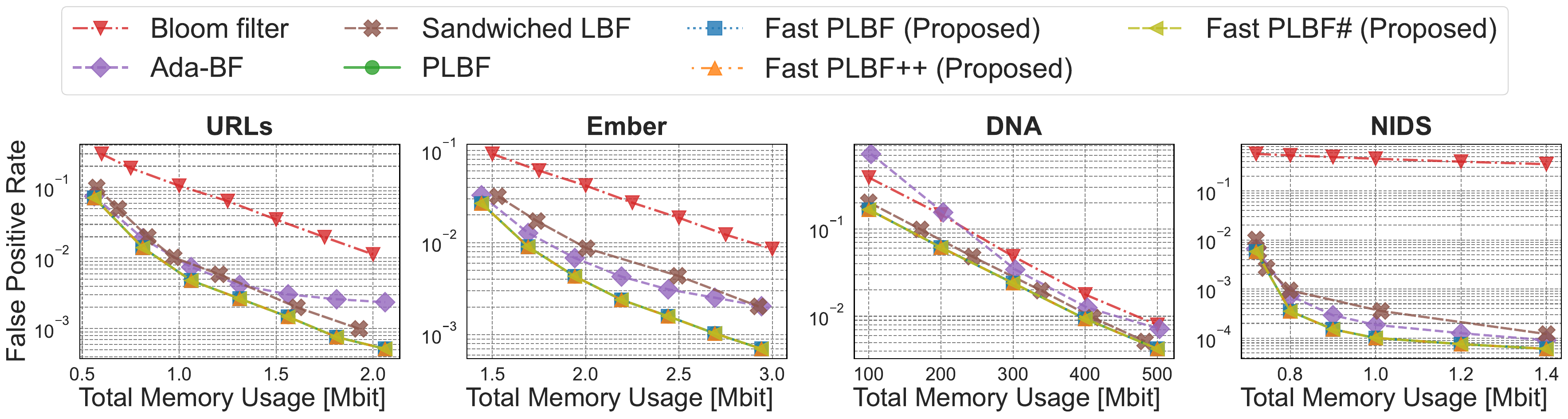}
    \caption{\Rone{Trade-off between memory usage and FPR on the test set. For LBFs, the memory usage is the sum of the backup Bloom filters and the learned model size. Fast PLBF constructs exactly the same data structure as PLBF. Fast PLBF++ and fast PLBF\# achieve nearly identical memory efficiency to (fast) PLBF across the evaluated memory budgets.}}
    \label{fig: memory_fpr}
\end{figure*}

\begin{figure*}[t]
    \centering
    \includegraphics[width=\textwidth]{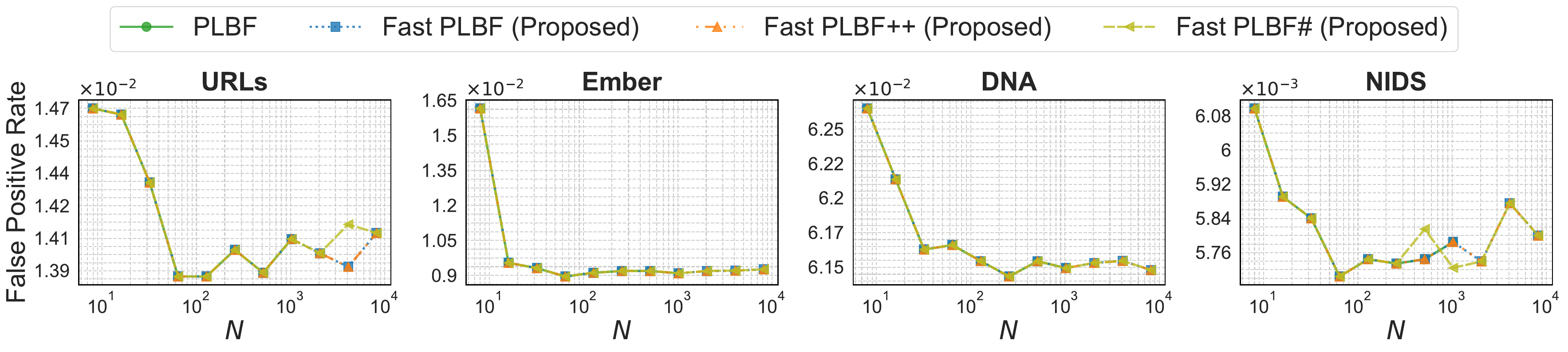}
    \caption{\Rone{Ablation study for the hyperparameter $N$ (number of segments): test-set FPR. $k$ is fixed at 5. As $N$ increases, the FPR tends to decrease for all four PLBF variants. The proposed methods retain nearly the same accuracy as PLBF.}}
    \label{fig: N_FPR}
\end{figure*}

\begin{figure*}[t]
    \centering
    \includegraphics[width=\textwidth]{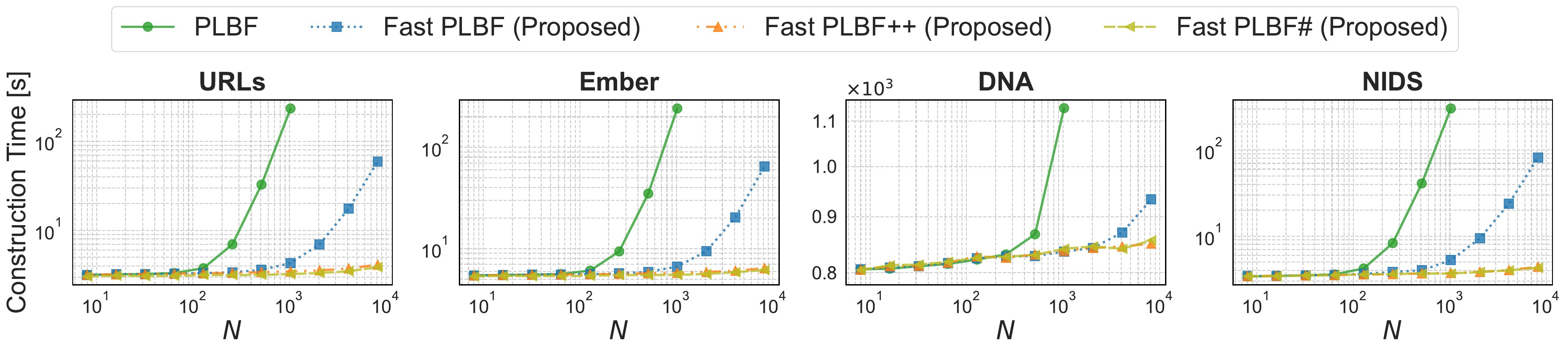}
    \caption{\Rone{Ablation study for the hyperparameter $N$ (number of segments): end-to-end construction time. $k$ is fixed at 5. As $N$ increases, the construction time increases; the proposed methods exhibit a much slower growth than PLBF.}}
    \label{fig: N_Time}
\end{figure*}

\begin{figure*}[t]
    \centering
    \includegraphics[width=\textwidth]{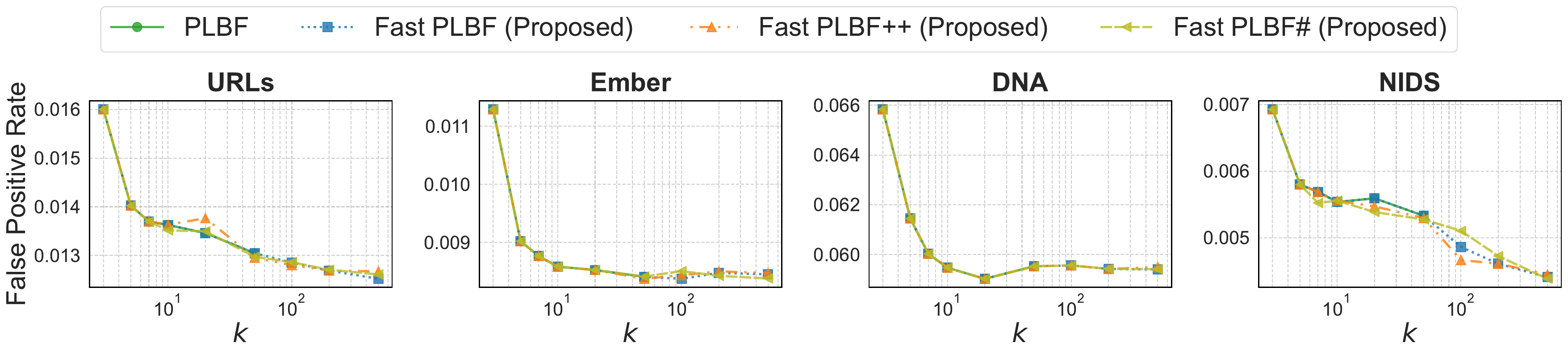}
    \caption{\Rone{Ablation study for the hyperparameter $k$ (number of regions): test-set FPR. $N$ is fixed at 1{,}000. As $k$ increases, the FPR tends to decrease for all four PLBF variants. The proposed methods retain nearly the same accuracy as PLBF.}}
    \label{fig: K_FPR}
\end{figure*}

\begin{figure*}[t]
    \centering
    \includegraphics[width=\textwidth]{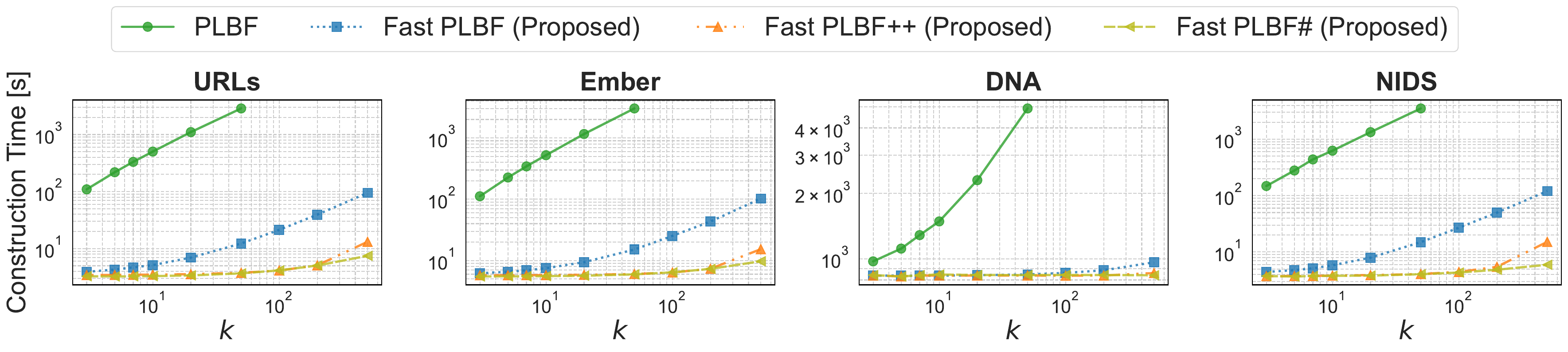}
    \caption{\Rone{Ablation study for the hyperparameter $k$ (number of regions): end-to-end construction time. $N$ is fixed at 1{,}000. As $k$ increases, the construction time increases; the proposed methods exhibit a much slower growth than PLBF.}}
    \label{fig: K_Time}
\end{figure*}

\begin{figure*}[t]
    \centering
    \subfloat[$0$ swaps]{
        \label{fig: Arti_0_0_log_hist}
        \includegraphics[width=0.65\columnwidth]{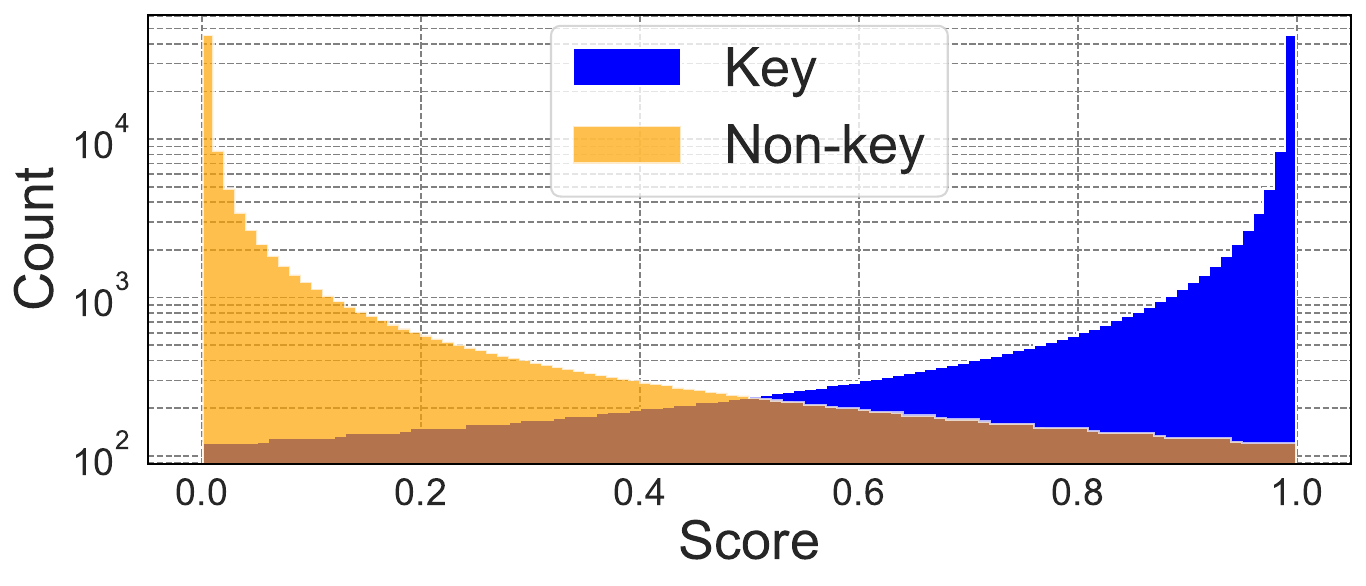}
    }
    \subfloat[$10$ swaps]{
        \label{fig: Arti_10_0_log_hist}
        \includegraphics[width=0.65\columnwidth]{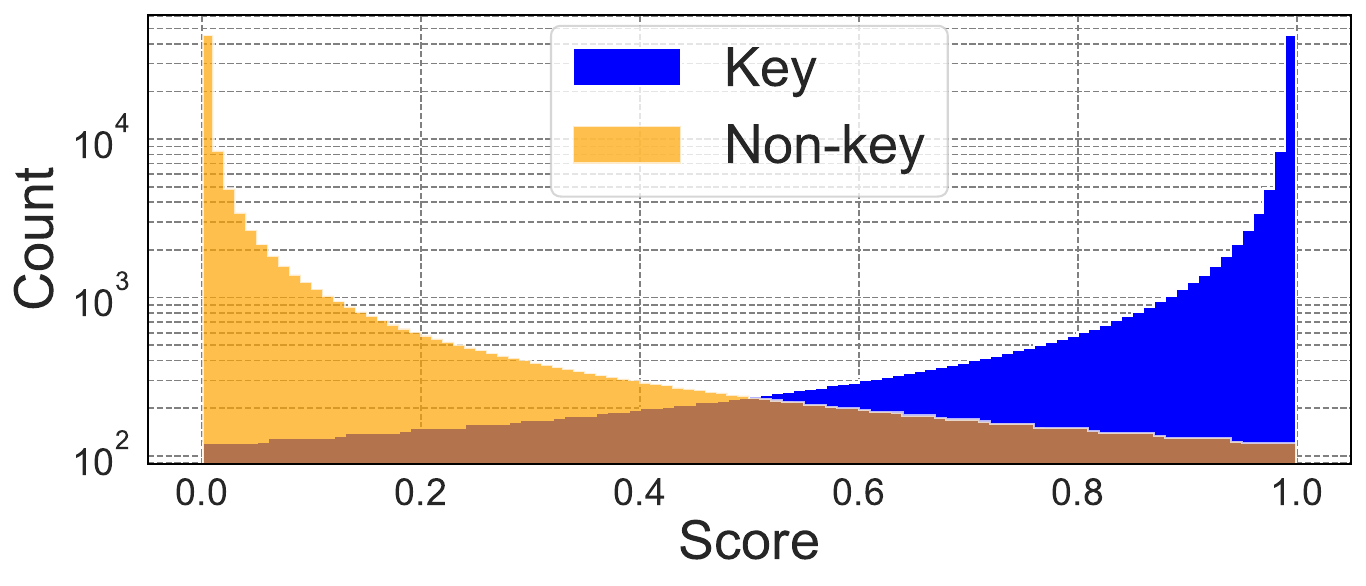}
    }
    \subfloat[$10^2$ swaps]{
        \label{fig: Arti_100_0_log_hist}
        \includegraphics[width=0.65\columnwidth]{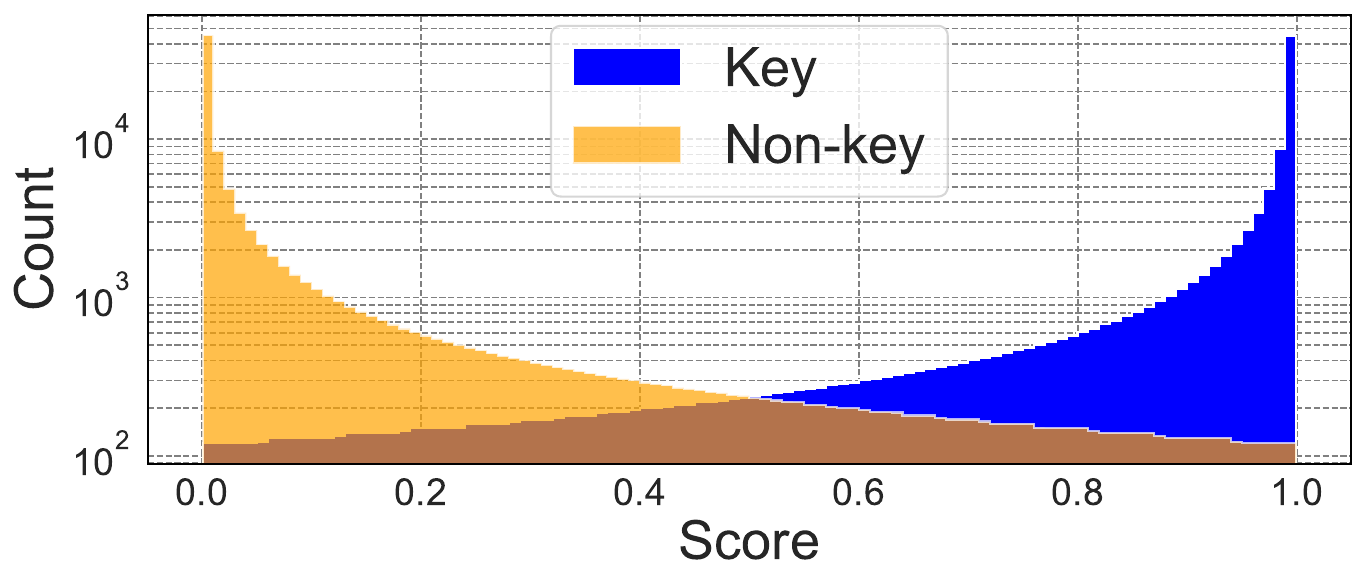}
    }
    \\
    \vspace{-0.5em}
    \subfloat[$10^3$ swaps]{
        \label{fig: Arti_1000_0_log_hist}
        \includegraphics[width=0.65\columnwidth]{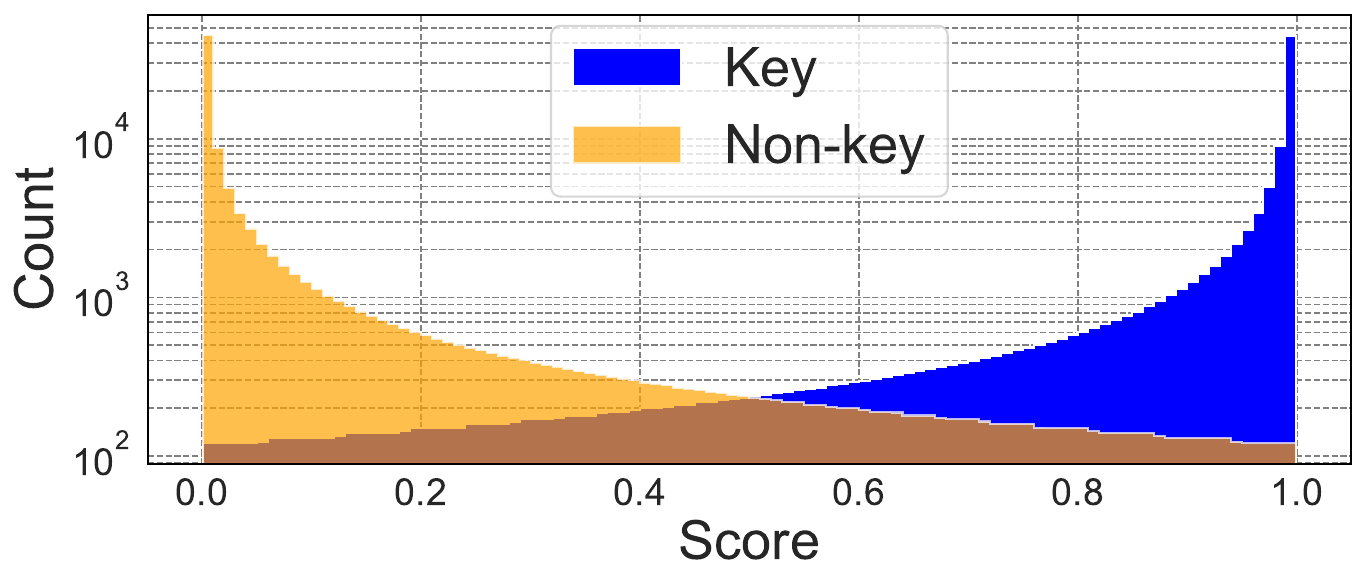}
    }
    \subfloat[$10^4$ swaps]{
        \label{fig: Arti_10000_0_log_hist}
        \includegraphics[width=0.65\columnwidth]{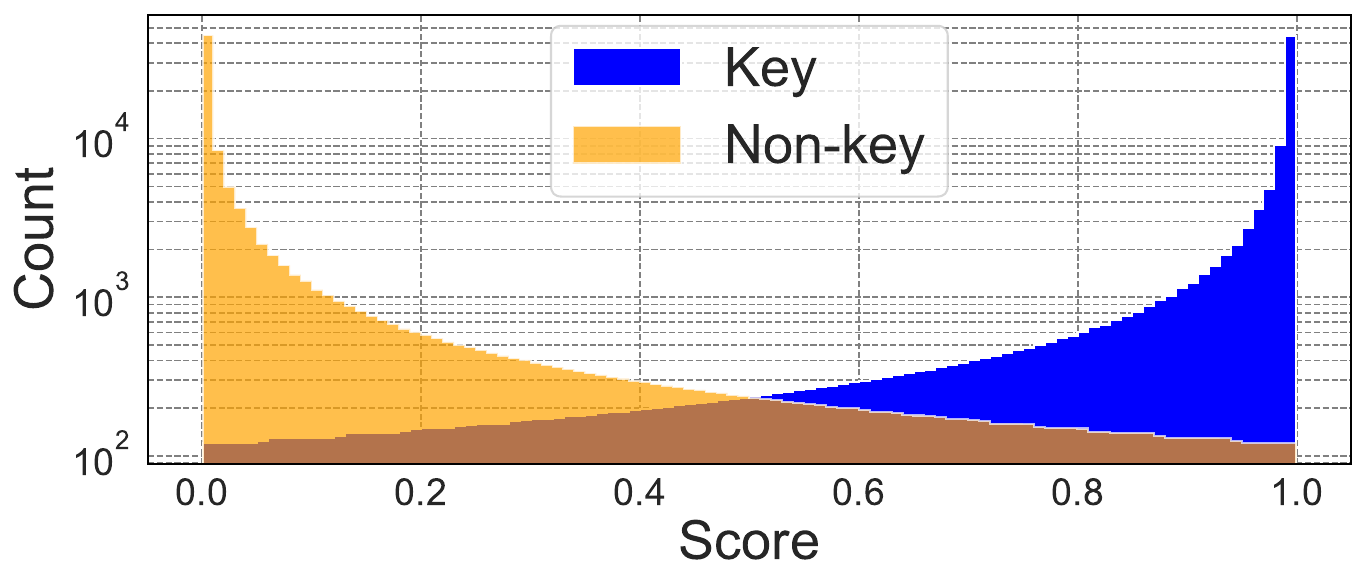}
    }
    \subfloat[$10^5$ swaps]{
        \label{fig: Arti_100000_0_log_hist}
        \includegraphics[width=0.65\columnwidth]{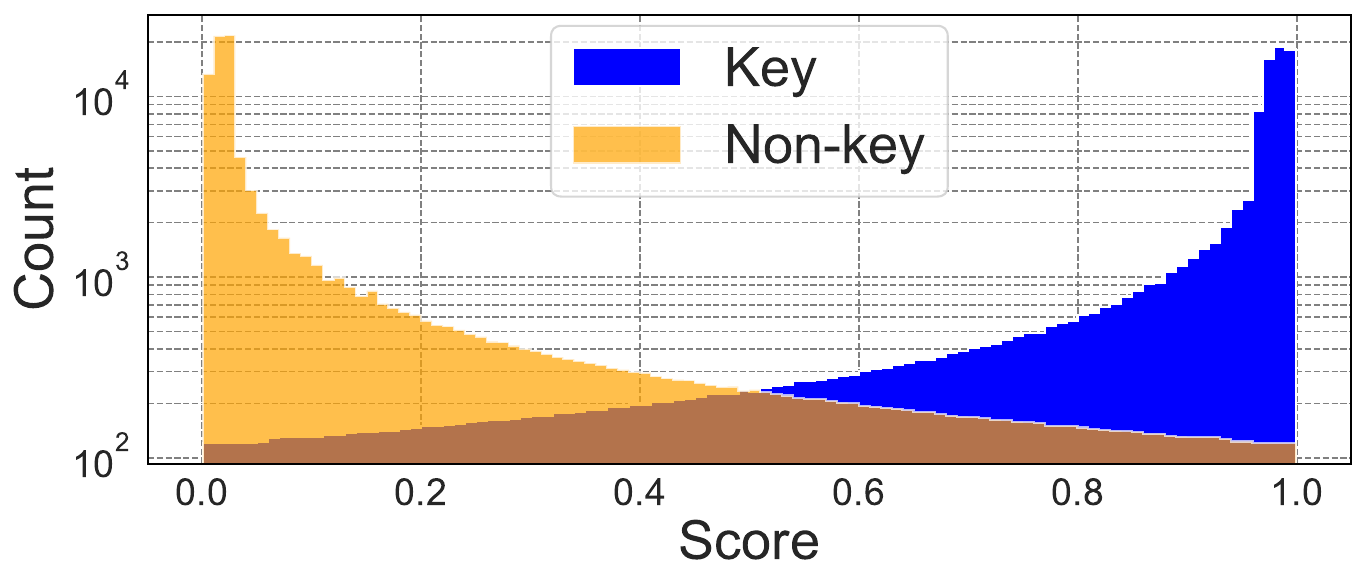}
    }
    \\
    \vspace{-0.5em}
    \subfloat[$10^6$ swaps]{
        \label{fig: Arti_1000000_0_log_hist}
        \includegraphics[width=0.65\columnwidth]{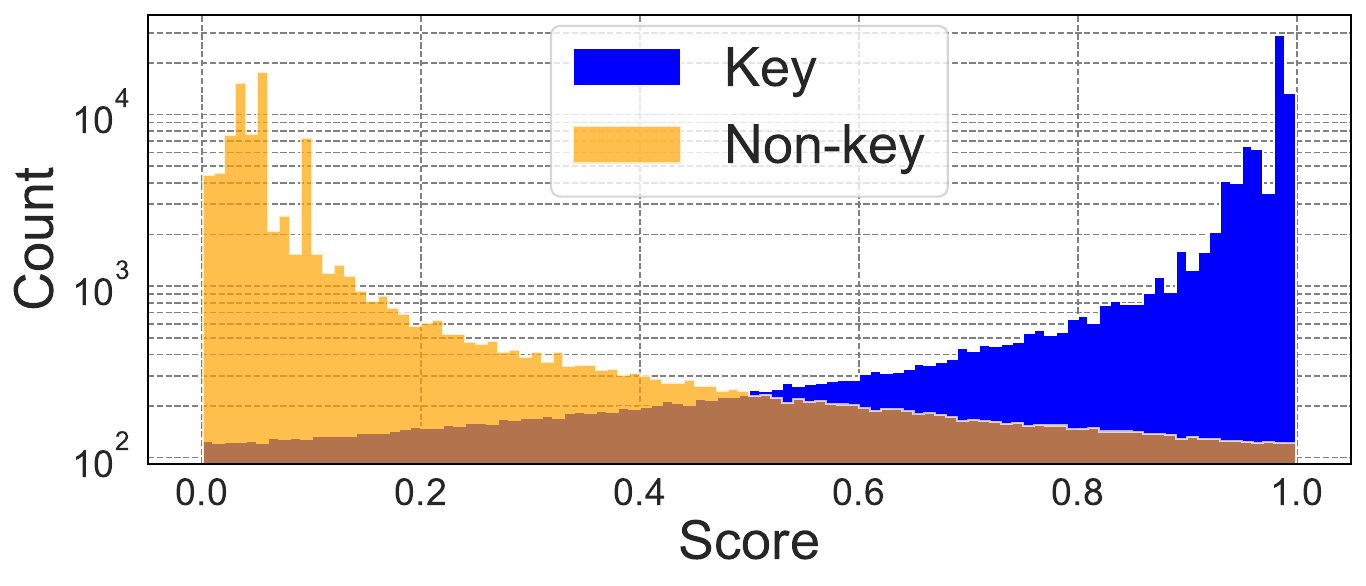}
    }
    \subfloat[$10^7$ swaps]{
        \label{fig: Arti_10000000_0_log_hist}
        \includegraphics[width=0.65\columnwidth]{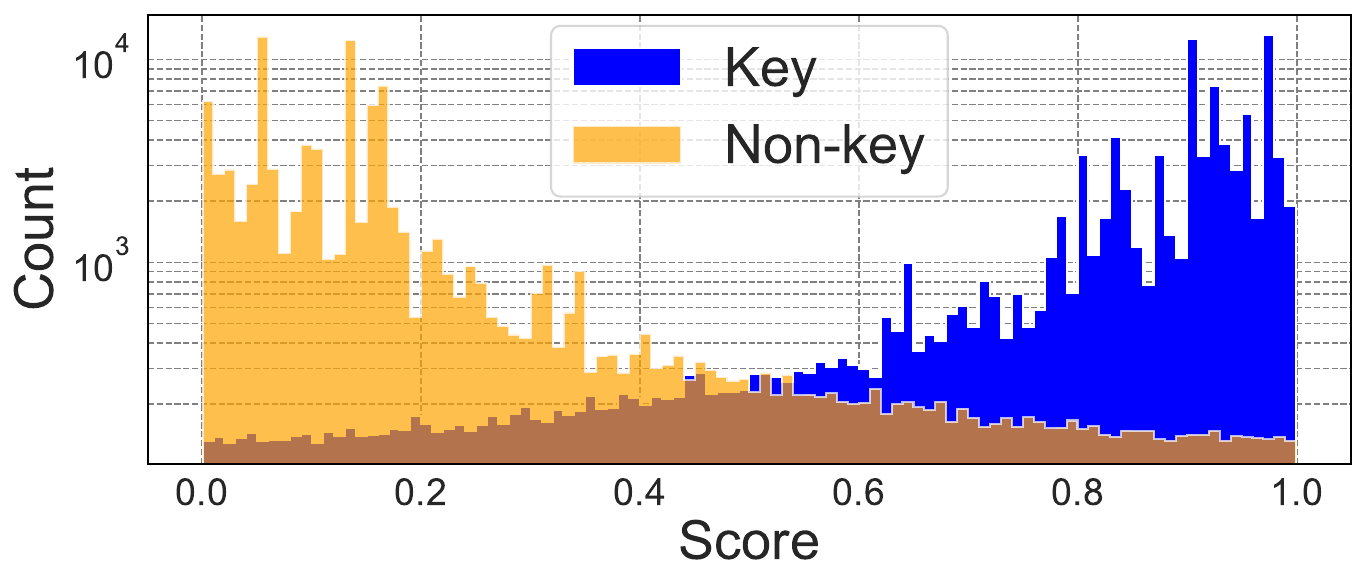}
    }
    \subfloat[$10^8$ swaps]{
        \label{fig: Arti_100000000_0_log_hist}
        \includegraphics[width=0.65\columnwidth]{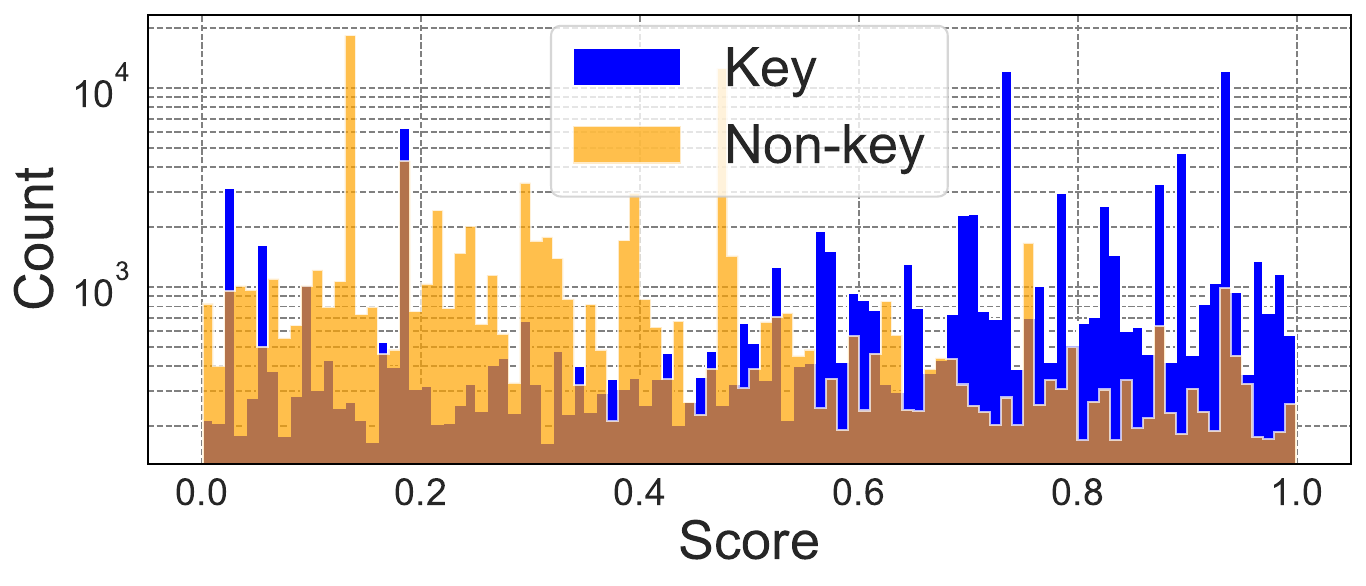}
    }
    \caption{Score distribution histograms of the artificial dataset (seed 0) with varying numbers of swaps. With $0$ swaps, keys are concentrated in higher-score regions while non-keys dominate lower-score regions; as the swap count increases, this separation progressively weakens and the two distributions overlap more.}
    \label{fig: Artificial_log_hist}
\end{figure*}

\begin{figure*}[t]
    \centering
    \subfloat[$0$ swaps]{
        \label{fig: Arti_0_0_pos_neg_ratio}
        \includegraphics[width=0.65\columnwidth]{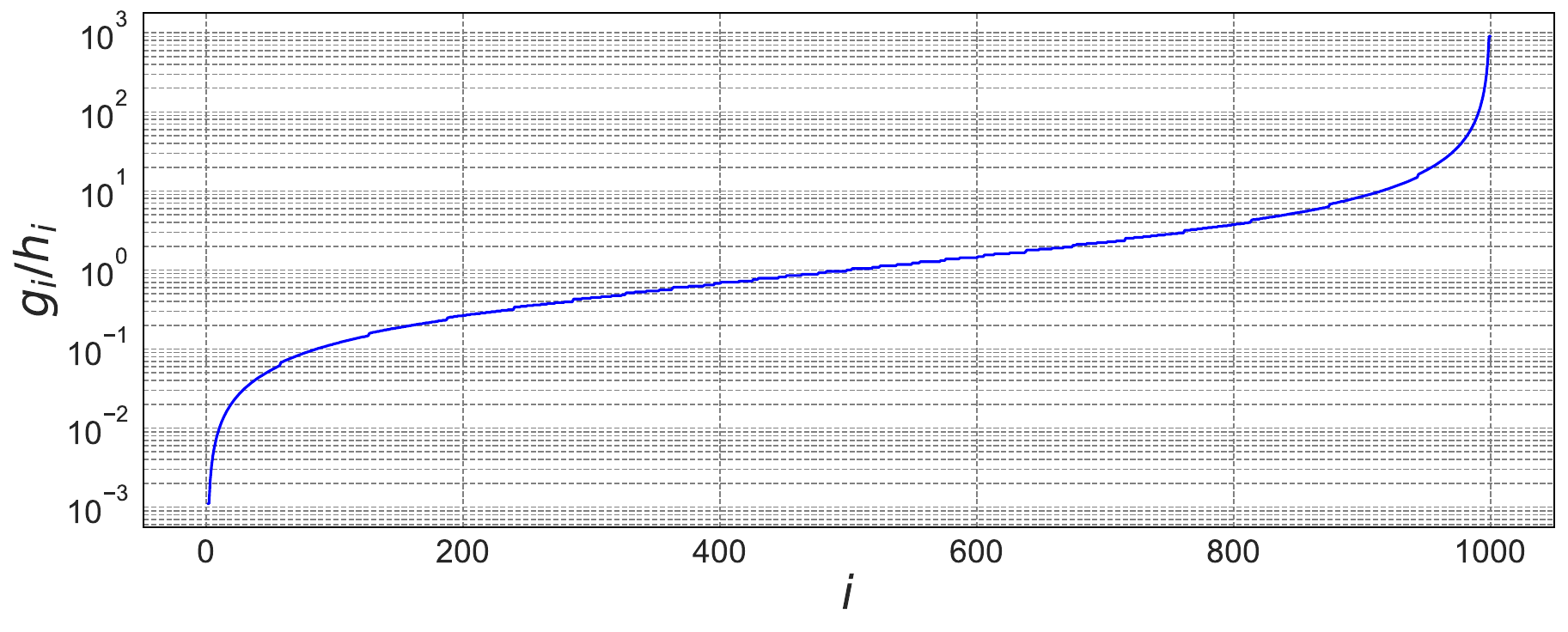}
    }
    \subfloat[$10$ swaps]{
        \label{fig: Arti_10_0_pos_neg_ratio}
        \includegraphics[width=0.65\columnwidth]{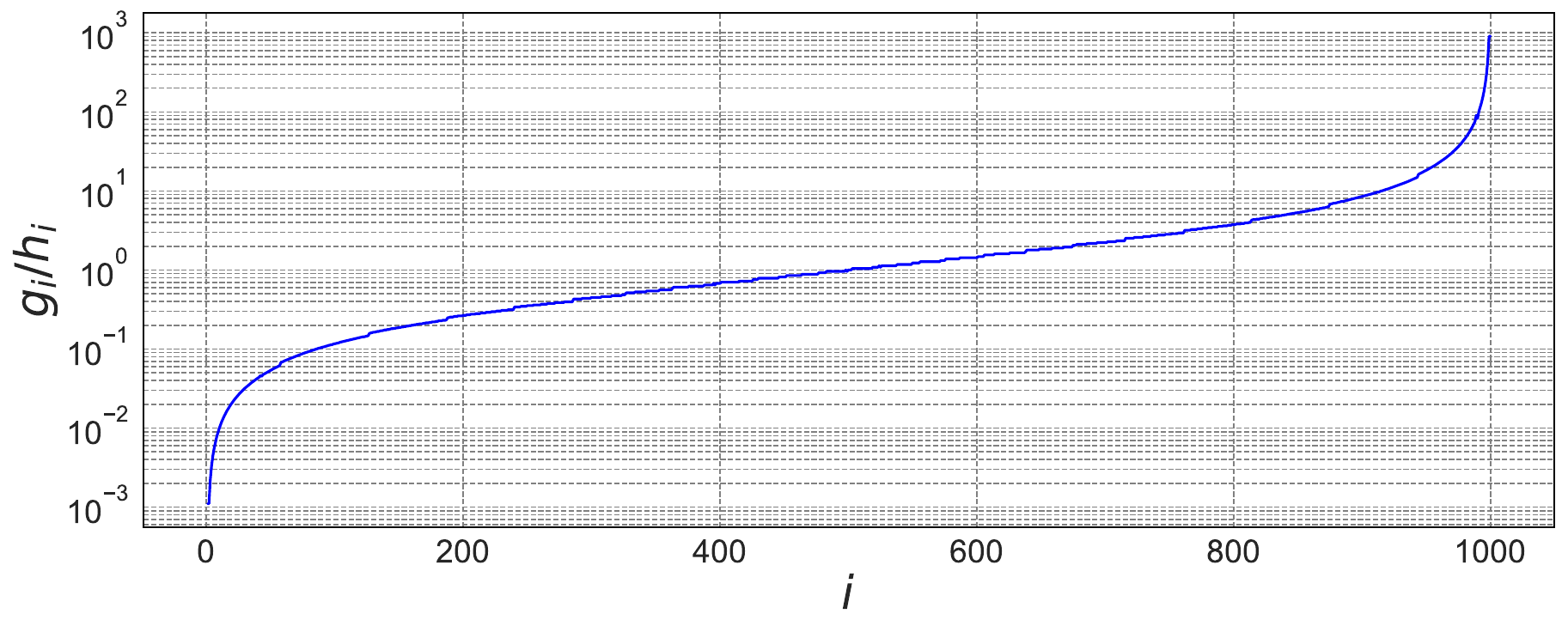}
    }
    \subfloat[$10^2$ swaps]{
        \label{fig: Arti_100_0_pos_neg_ratio}
        \includegraphics[width=0.65\columnwidth]{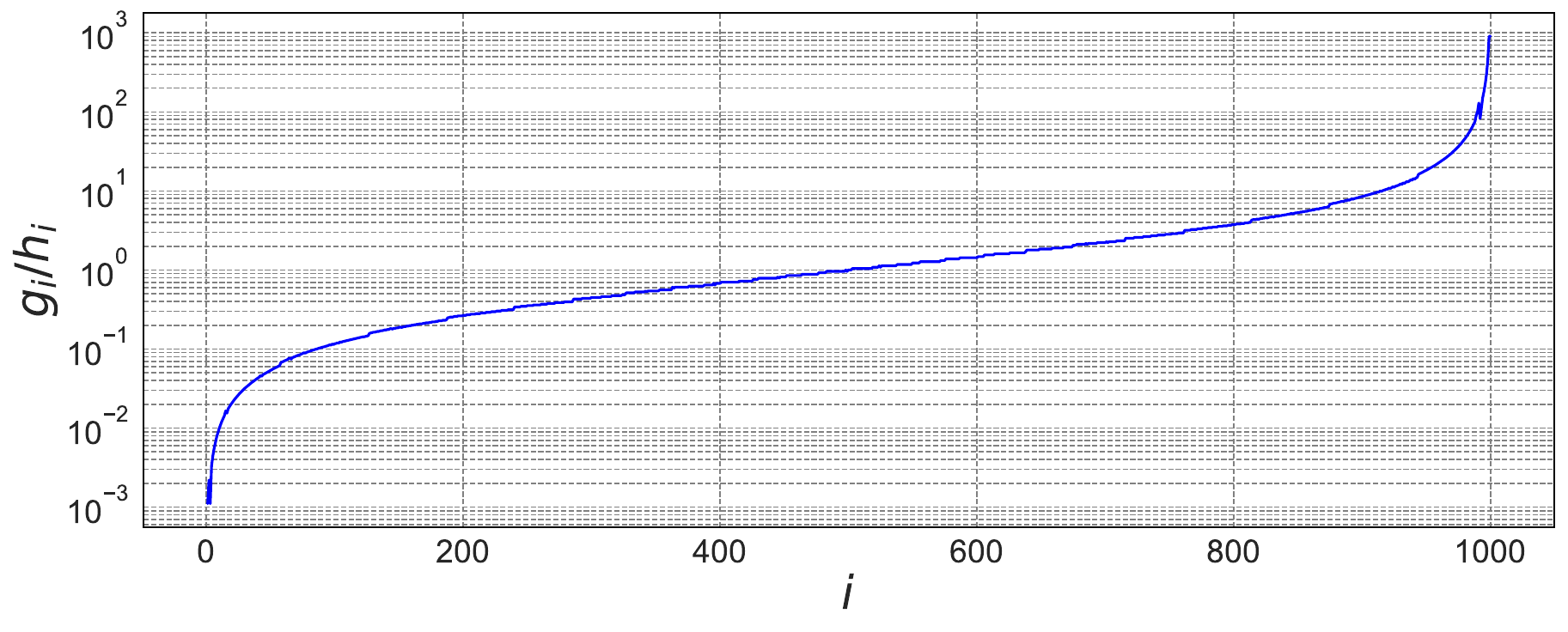}
    }
    \\
    \vspace{-0.5em}
    \subfloat[$10^3$ swaps]{
        \label{fig: Arti_1000_0_pos_neg_ratio}
        \includegraphics[width=0.65\columnwidth]{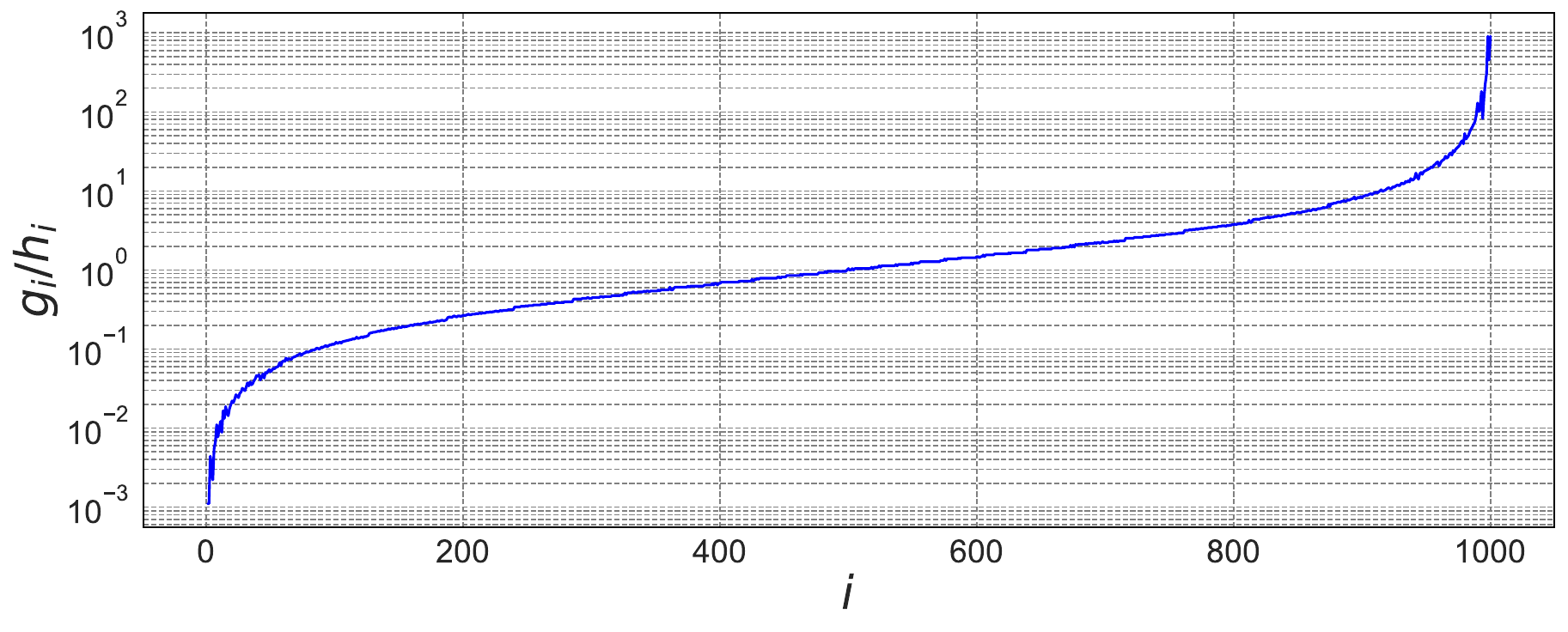}
    }
    \subfloat[$10^4$ swaps]{
        \label{fig: Arti_10000_0_pos_neg_ratio}
        \includegraphics[width=0.65\columnwidth]{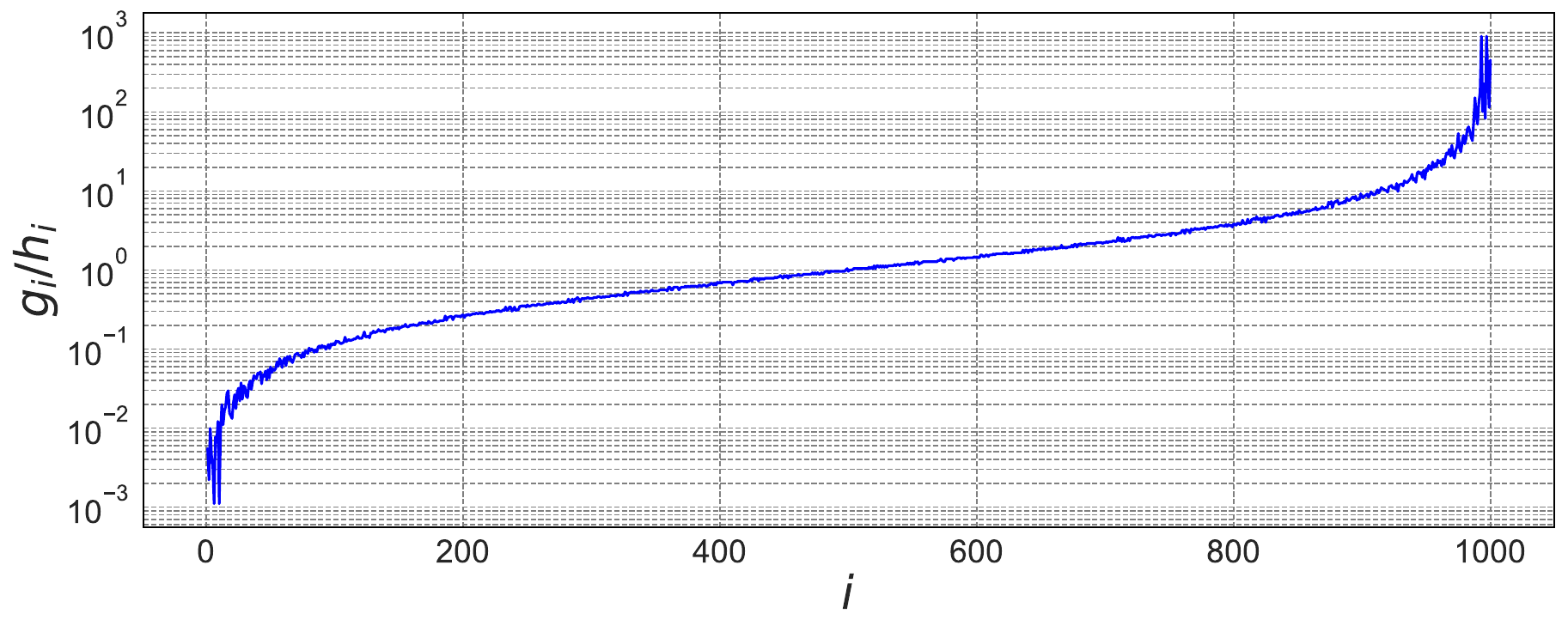}
    }
    \subfloat[$10^5$ swaps]{
        \label{fig: Arti_100000_0_pos_neg_ratio}
        \includegraphics[width=0.65\columnwidth]{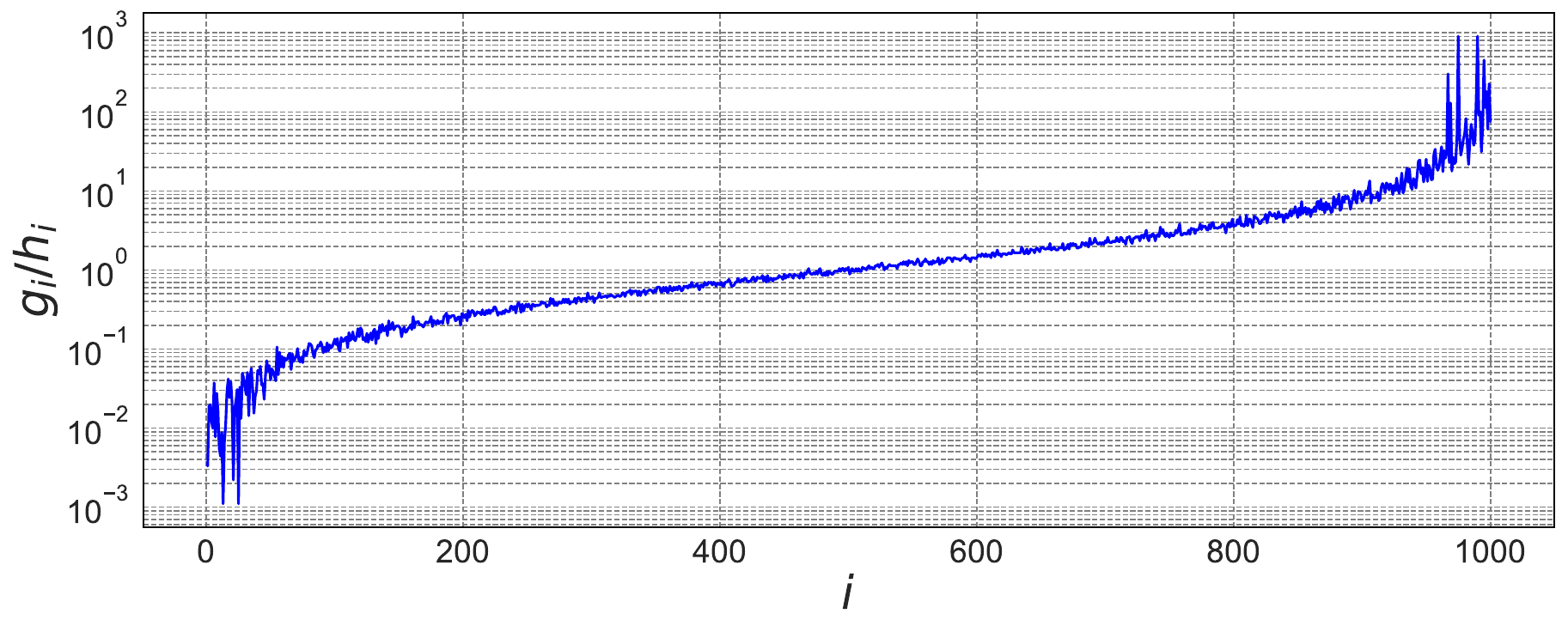}
    }
    \\
    \vspace{-0.5em}
    \subfloat[$10^6$ swaps]{
        \label{fig: Arti_1000000_0_pos_neg_ratio}
        \includegraphics[width=0.65\columnwidth]{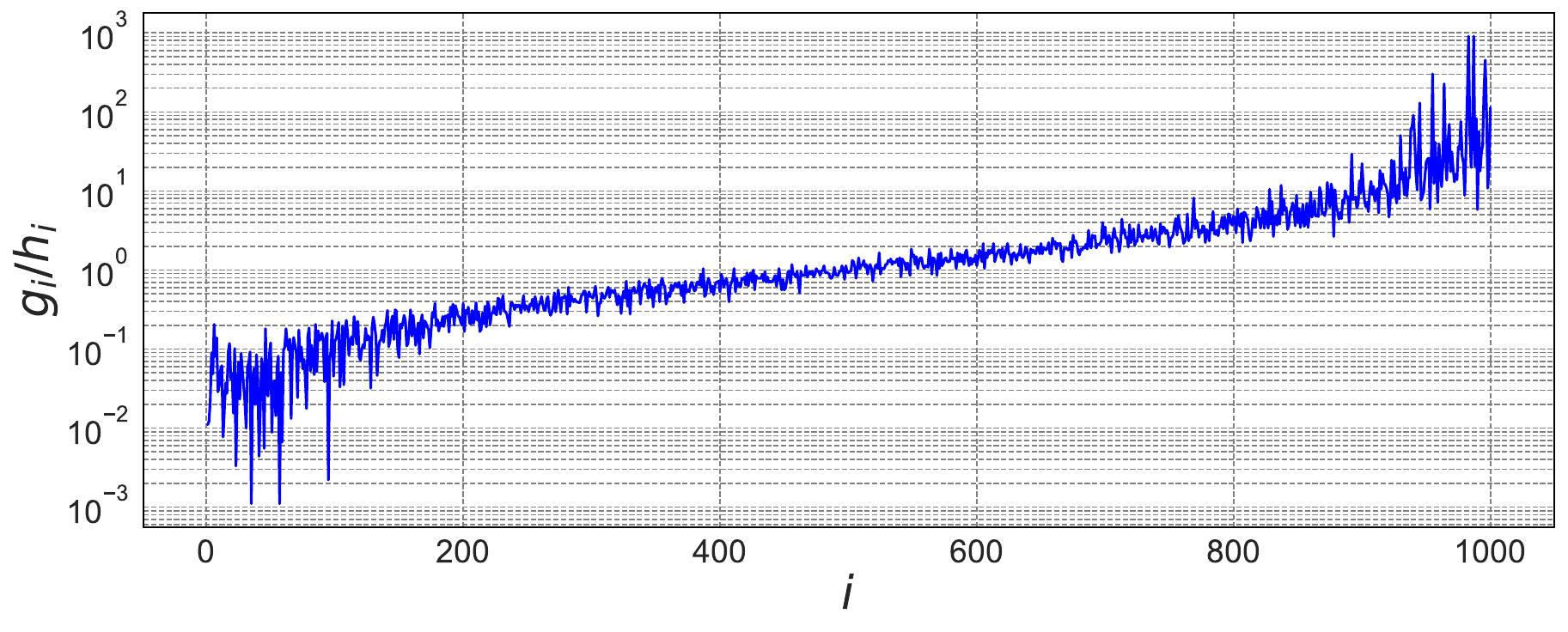}
    }
    \subfloat[$10^7$ swaps]{
        \label{fig: Arti_10000000_0_pos_neg_ratio}
        \includegraphics[width=0.65\columnwidth]{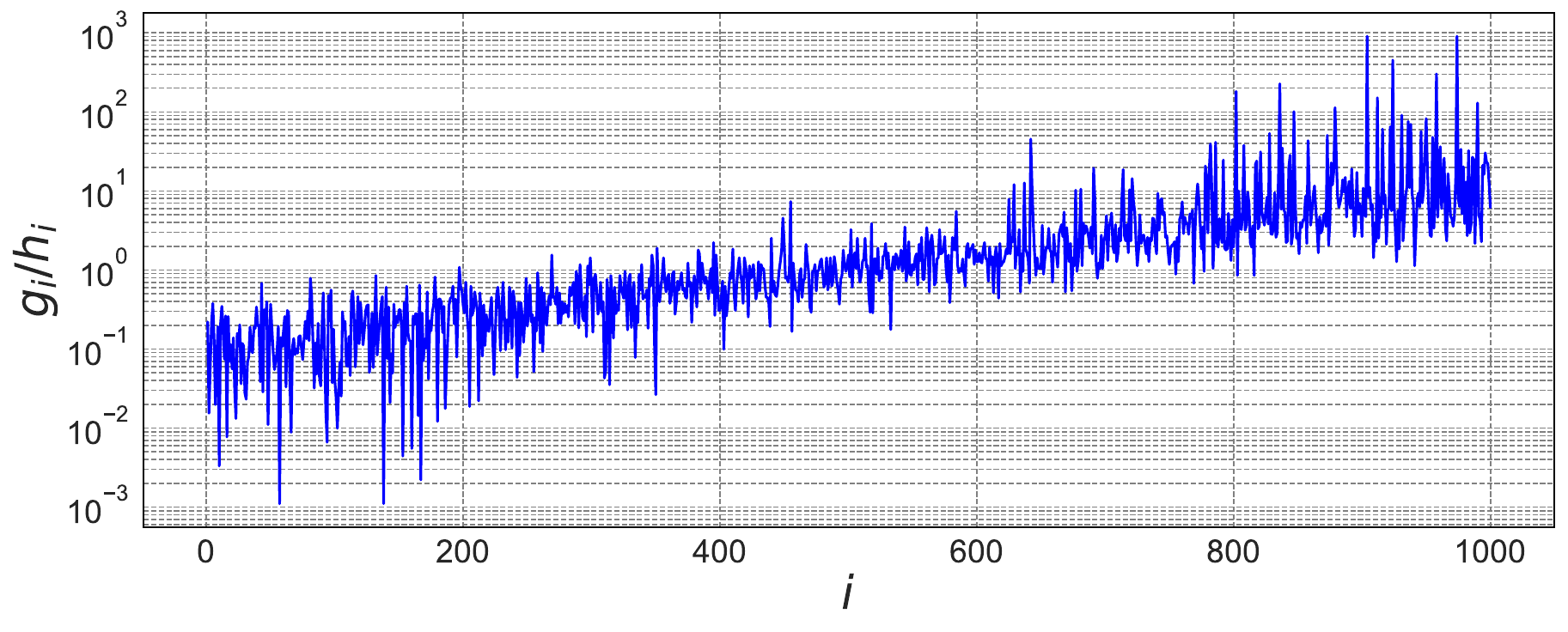}
    }
    \subfloat[$10^8$ swaps]{
        \label{fig: Arti_100000000_0_pos_neg_ratio}
        \includegraphics[width=0.65\columnwidth]{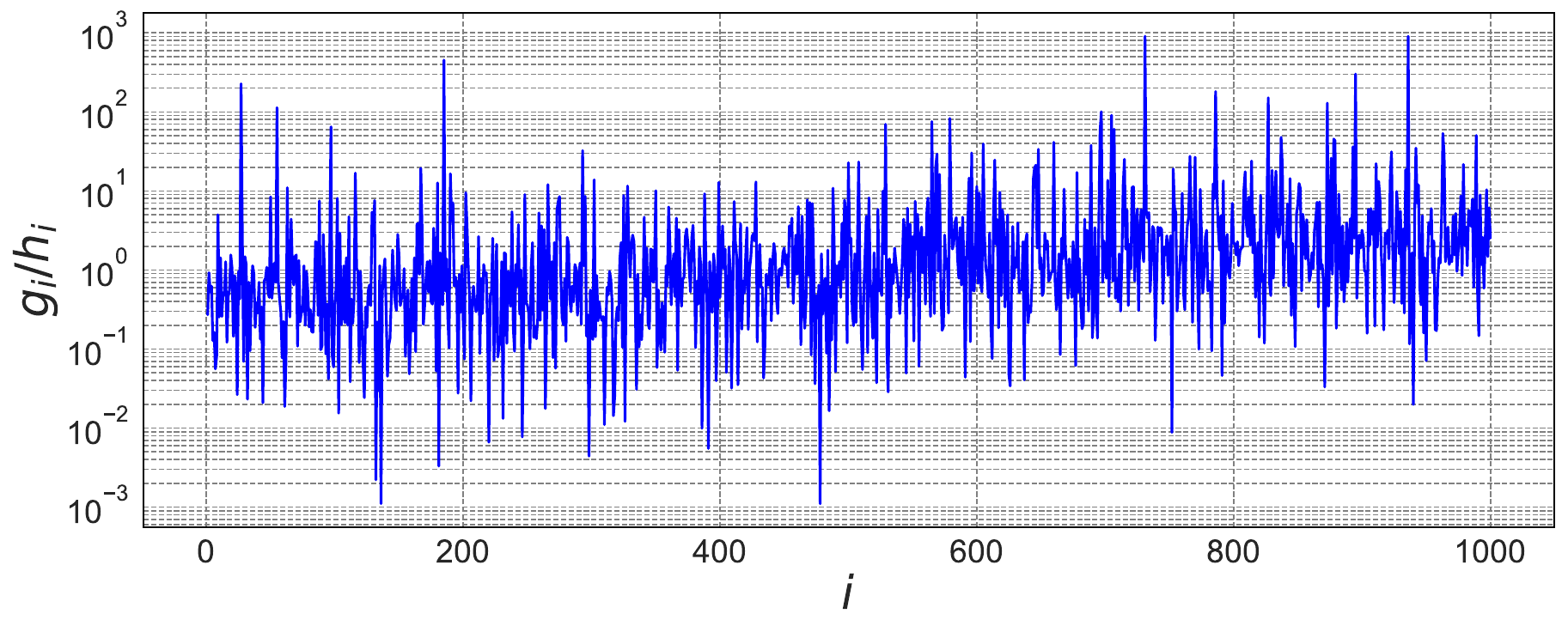}
    }
    \caption{Empirical ratio of key to non-key density, i.e., $g_i/h_i$, over score segments ($i=1,\dots,N$, $N=1{,}000$) for the artificial dataset (seed 0) with varying numbers of swaps. While $g_i/h_i$ shows an overall increasing trend with $i$ when the number of swaps is small, larger swap counts increasingly break monotonicity and induce stronger local fluctuations, indicating departures from the ideal assumption.}
    \label{fig: Artificial_pos_neg_ratio}
\end{figure*}

\begin{figure*}[t]
    \centering
    \begin{minipage}{1.7\columnwidth}
        \centering
        \includegraphics[width=\columnwidth]{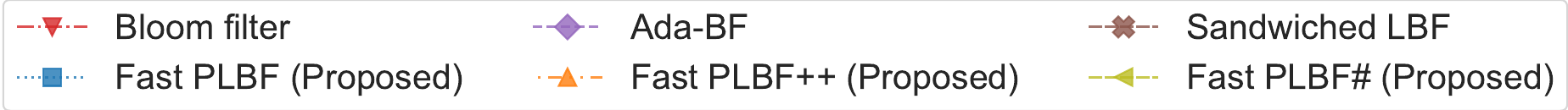}
    \end{minipage}
    
    \centering
    \subfloat[$0$ swaps]{
        \includegraphics[width=0.65\columnwidth]{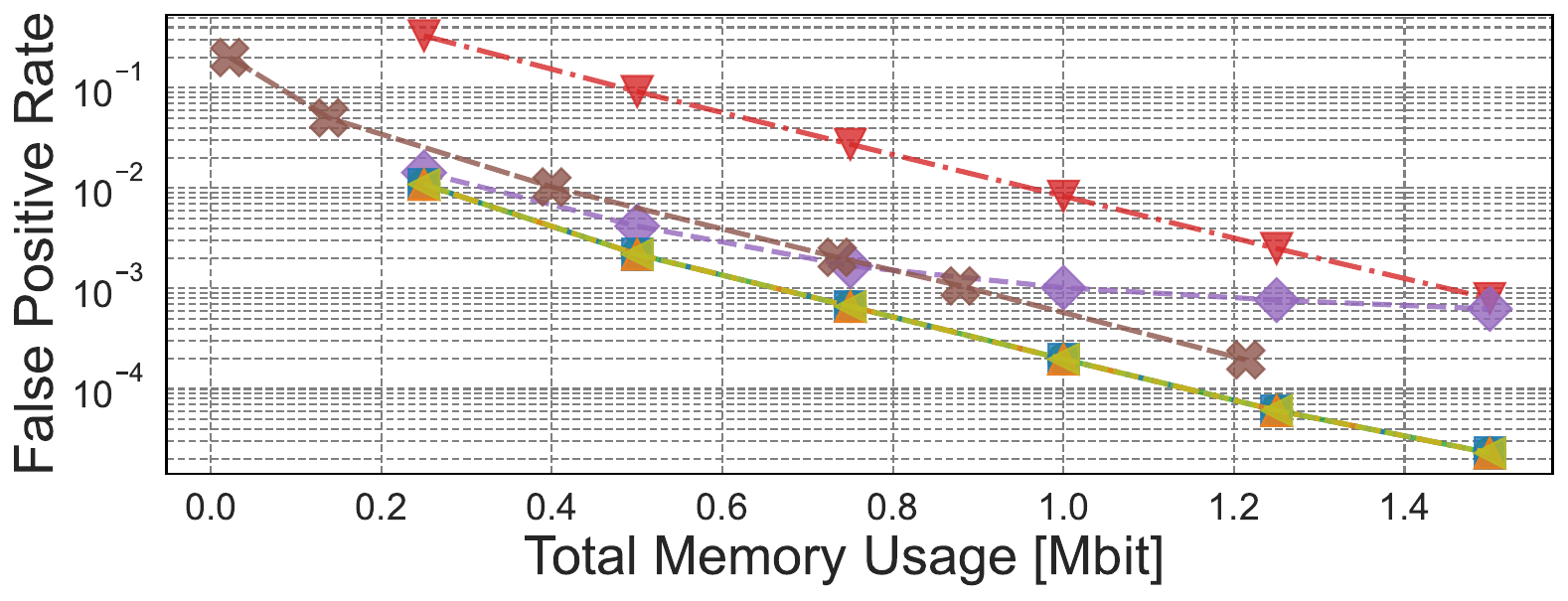}
    }
    \subfloat[$10$ swaps]{
        \includegraphics[width=0.65\columnwidth]{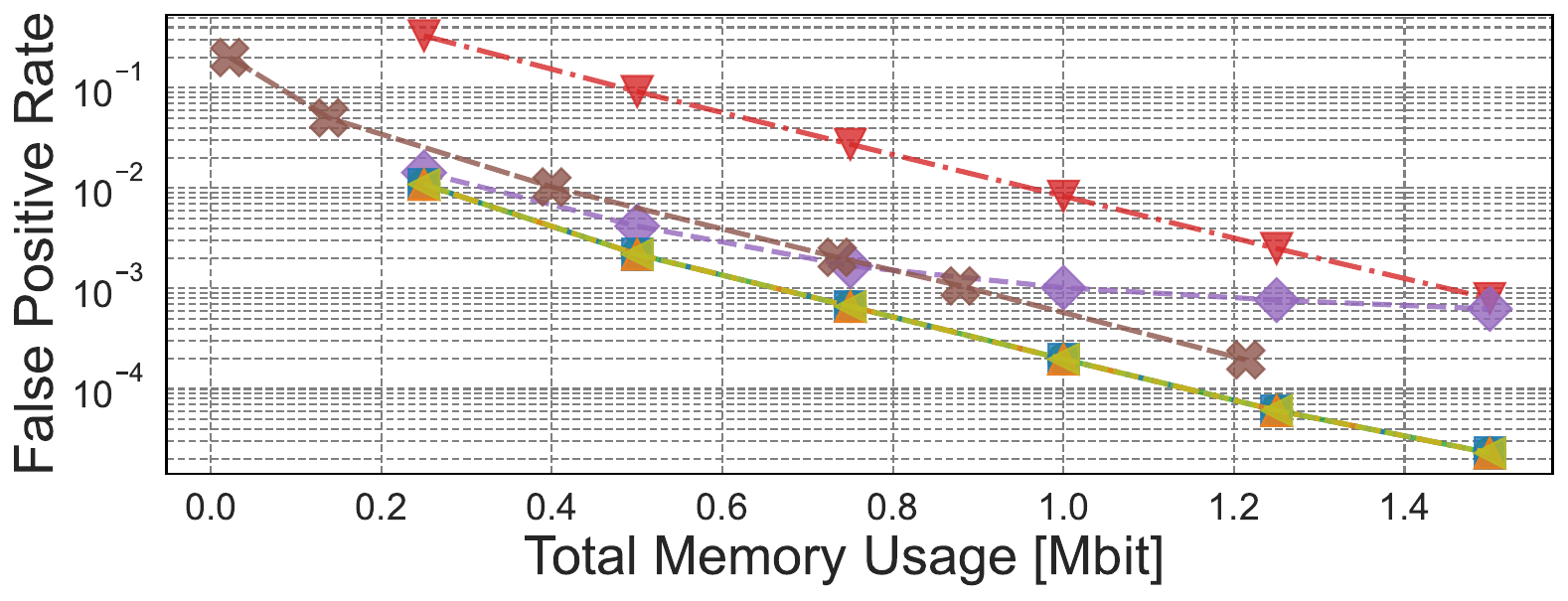}
    }
    \subfloat[$10^2$ swaps]{
        \includegraphics[width=0.65\columnwidth]{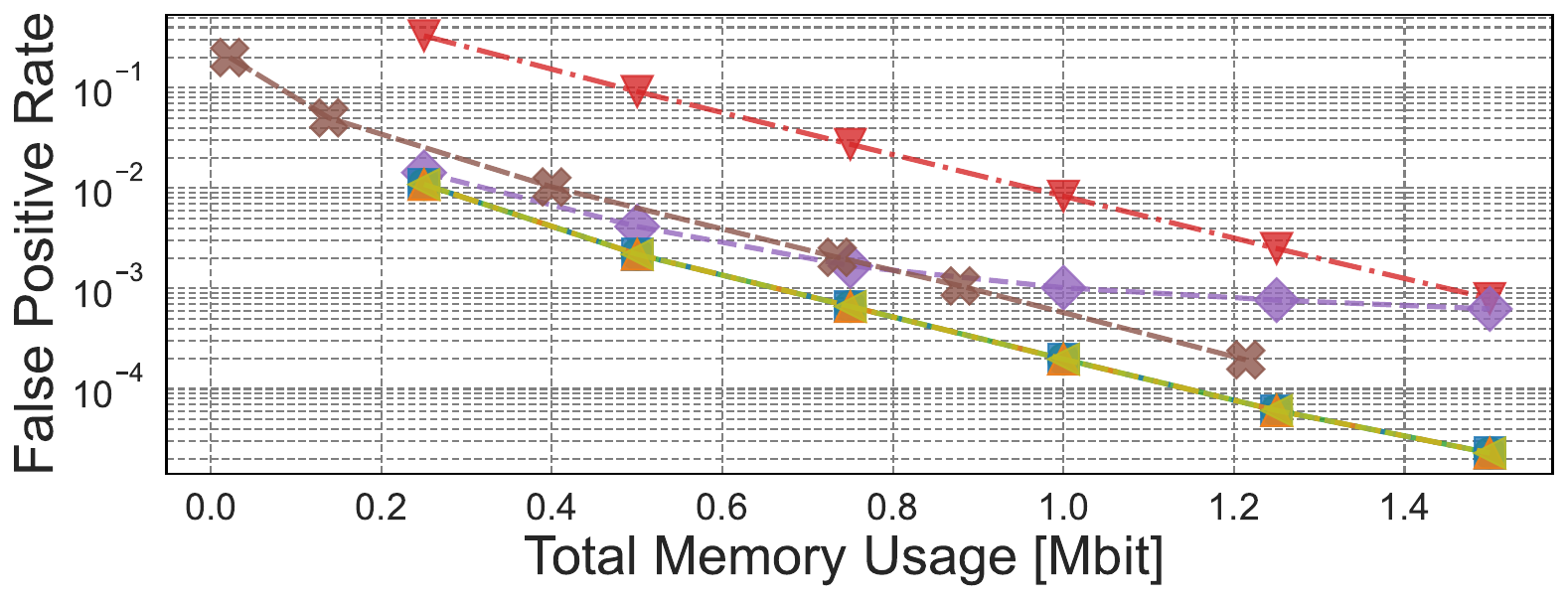}
    }
    \\
    \vspace{-0.5em}
    \subfloat[$10^3$ swaps]{
        \includegraphics[width=0.65\columnwidth]{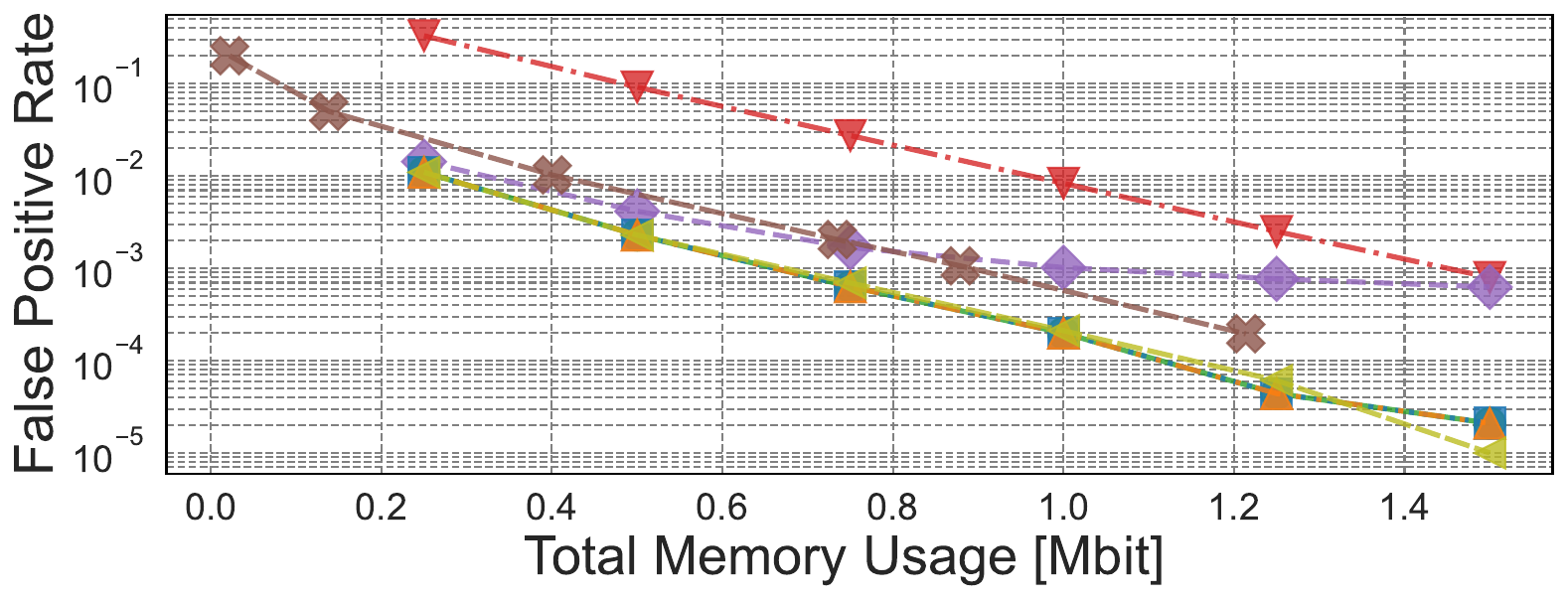}
    }
    \subfloat[$10^4$ swaps]{
        \includegraphics[width=0.65\columnwidth]{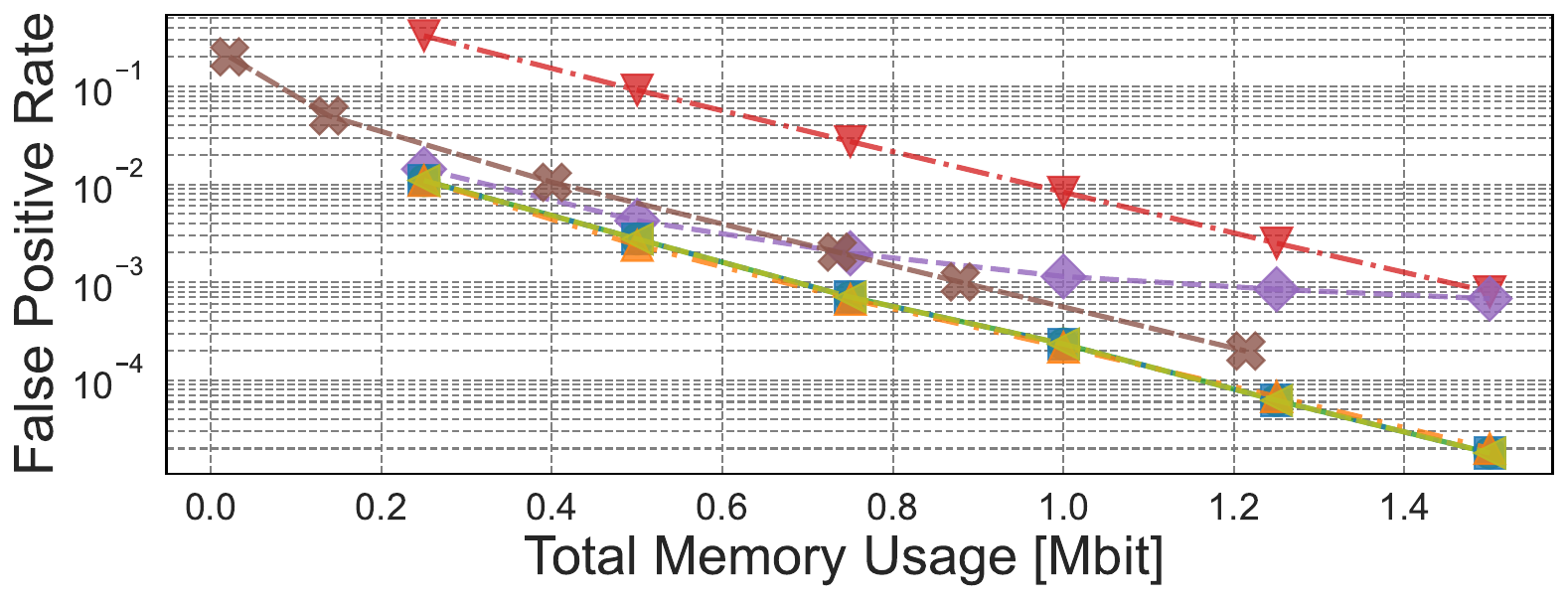}
    }
    \subfloat[$10^5$ swaps]{
        \includegraphics[width=0.65\columnwidth]{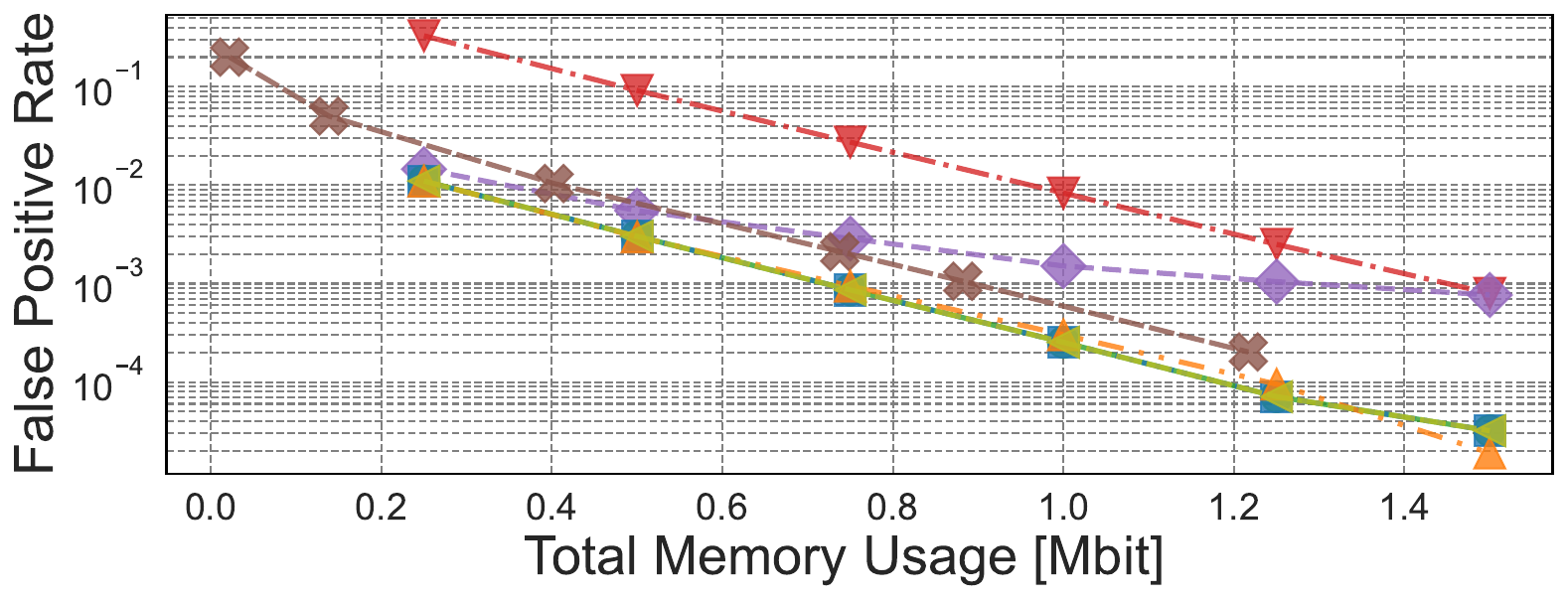}
    }
    \\
    \vspace{-0.5em}
    \subfloat[$10^6$ swaps]{
        \includegraphics[width=0.65\columnwidth]{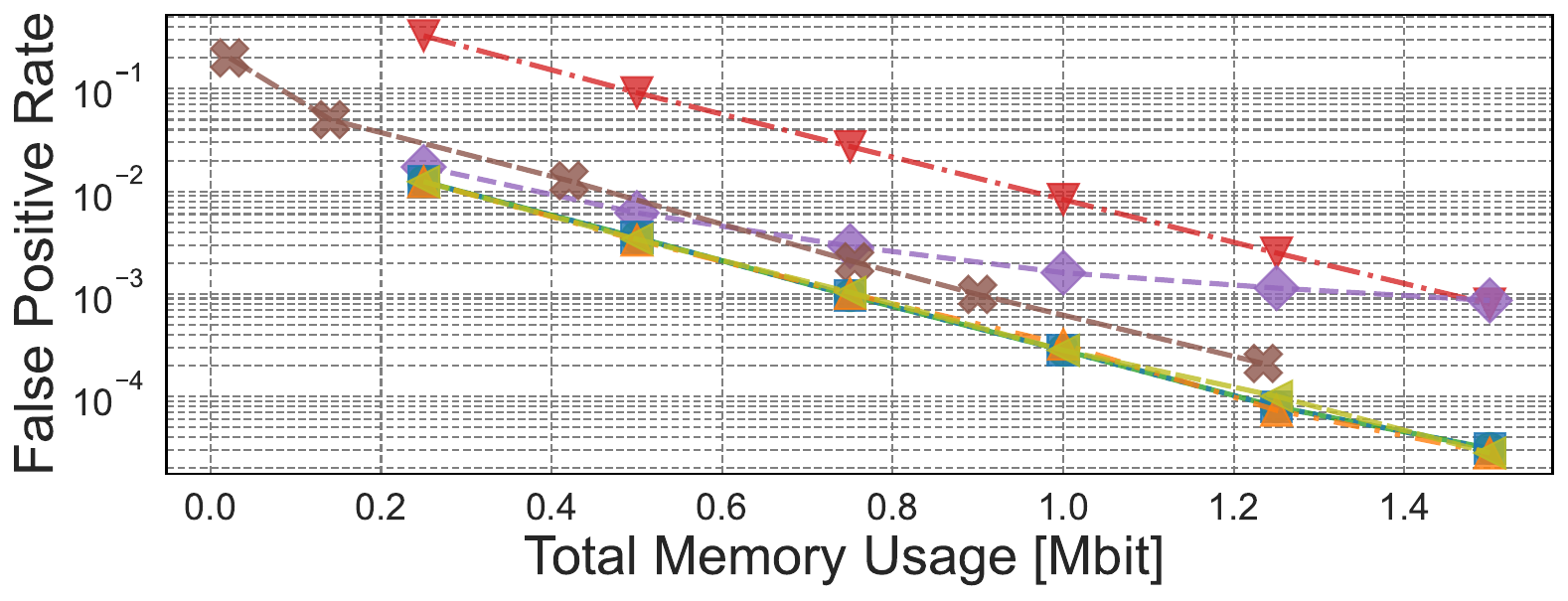}
    }
    \subfloat[$10^7$ swaps]{
        \includegraphics[width=0.65\columnwidth]{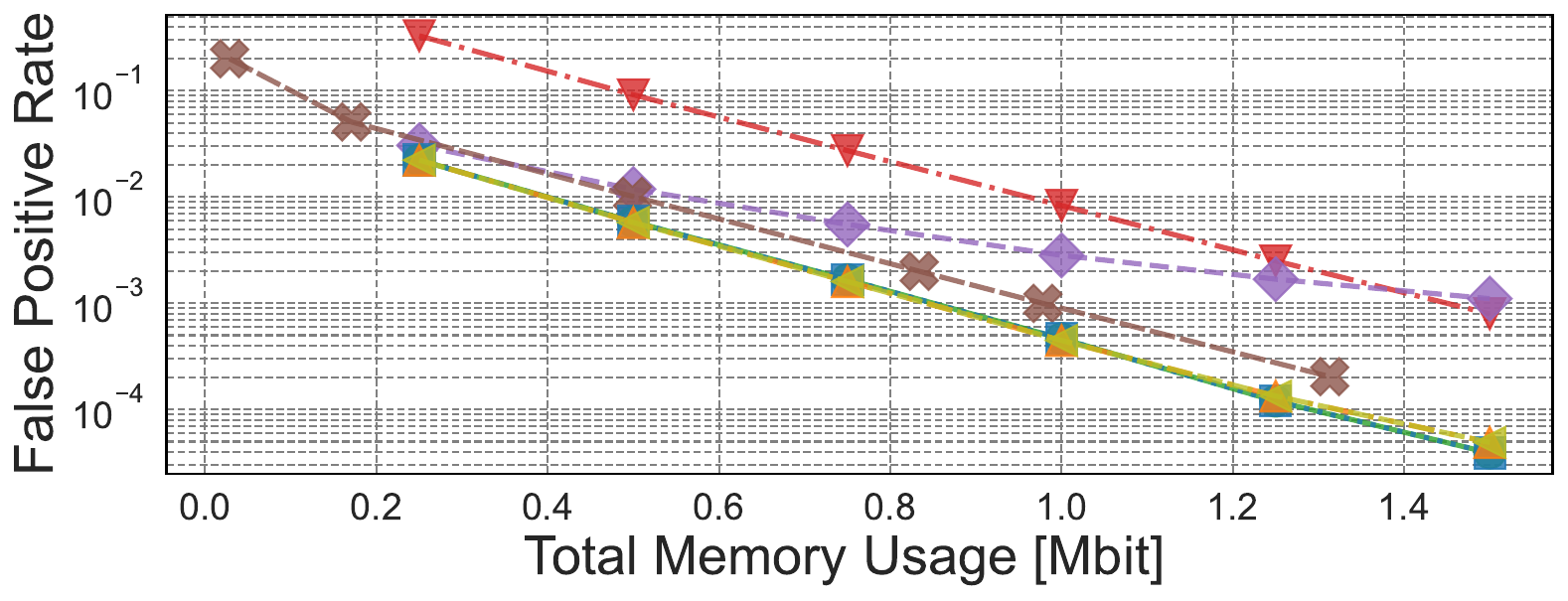}
    }
    \subfloat[$10^8$ swaps]{
        \includegraphics[width=0.65\columnwidth]{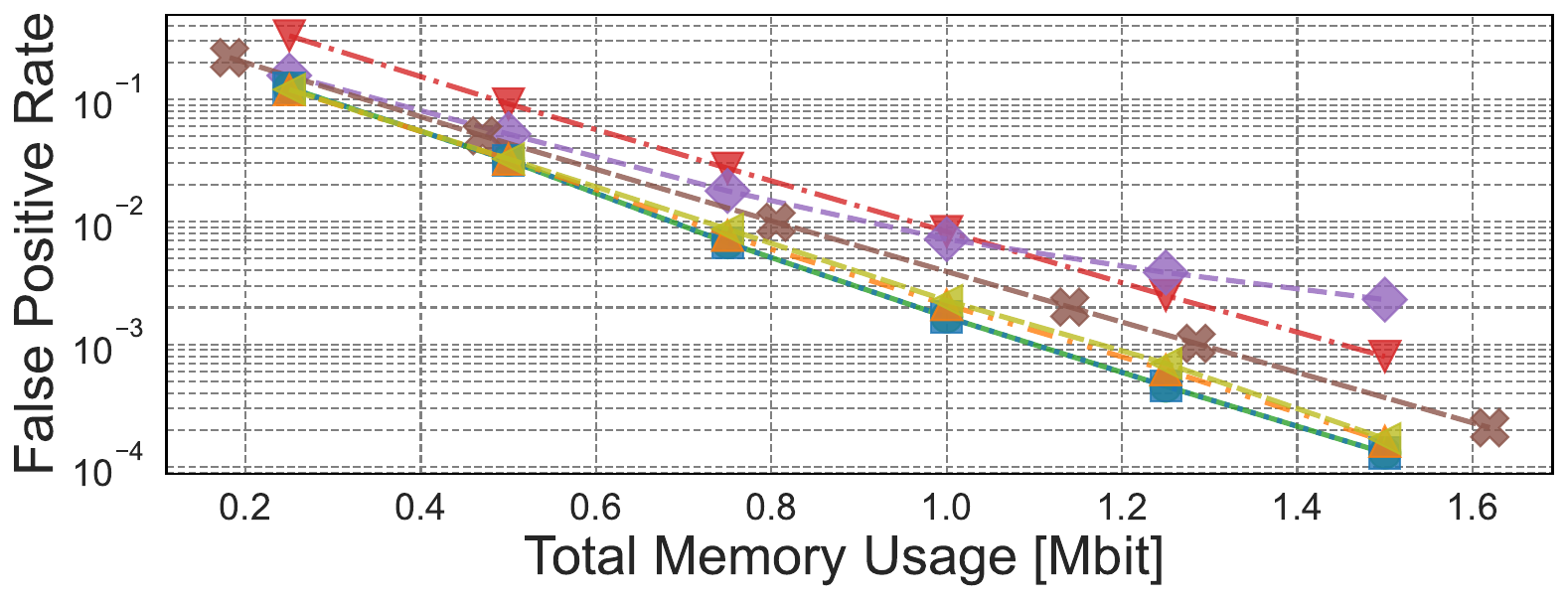}
    }
    \caption{Trade-off between memory usage and FPR on the test set for artificial datasets (seed 0) with varying numbers of swaps. A larger number of swaps introduces stronger non-monotonicity in the score distribution. When the score distribution becomes highly non-monotonic (e.g., $10^8$ swaps), fast PLBF++/\# shows a modest deviation from (fast) PLBF; otherwise, the memory--accuracy trade-offs of these methods are very similar.}
    \label{fig: Arti_Memory_and_FPR}
\end{figure*}

\begin{figure*}[t]
    \centering
    \subfloat[Fast PLBF++]{
        \label{fig: pp_diff}
        \includegraphics[width=\columnwidth]{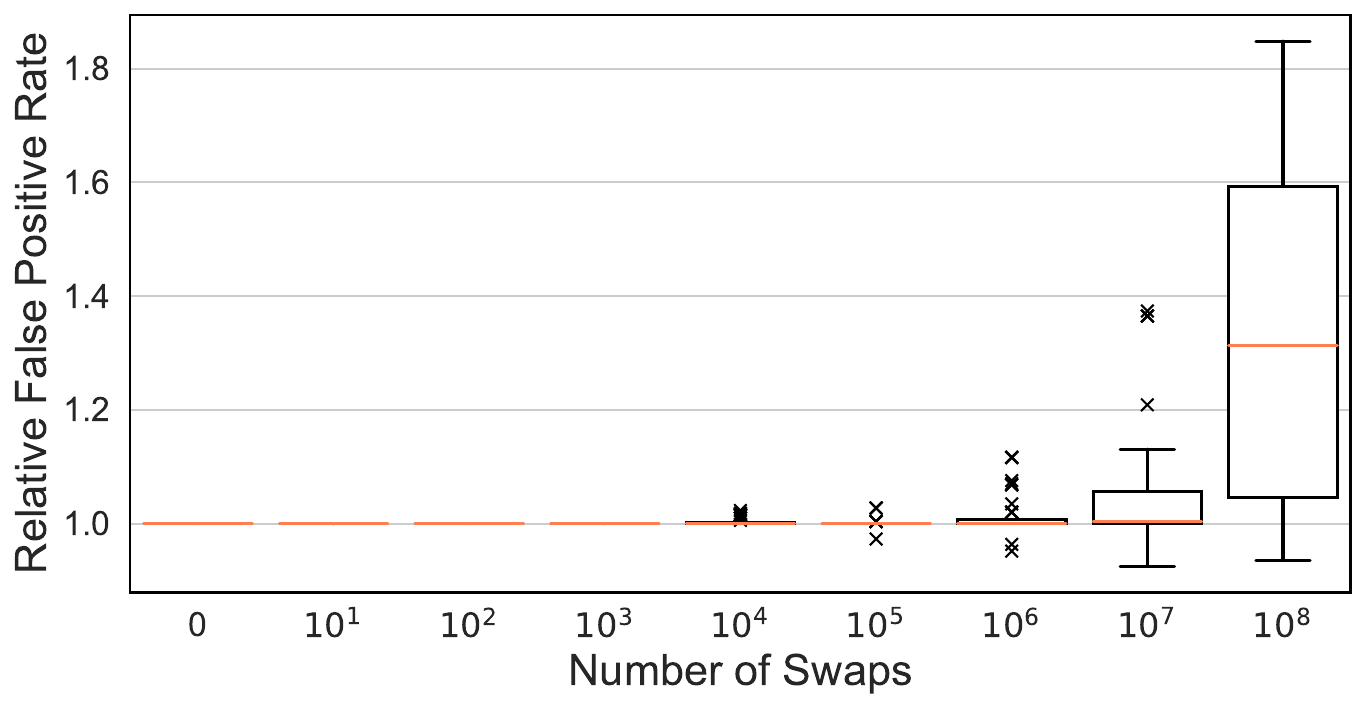}
    }
    \hspace{-1.0em}
    \subfloat[Fast PLBF\#]{
        \label{fig: monge_diff}
        \includegraphics[width=\columnwidth]{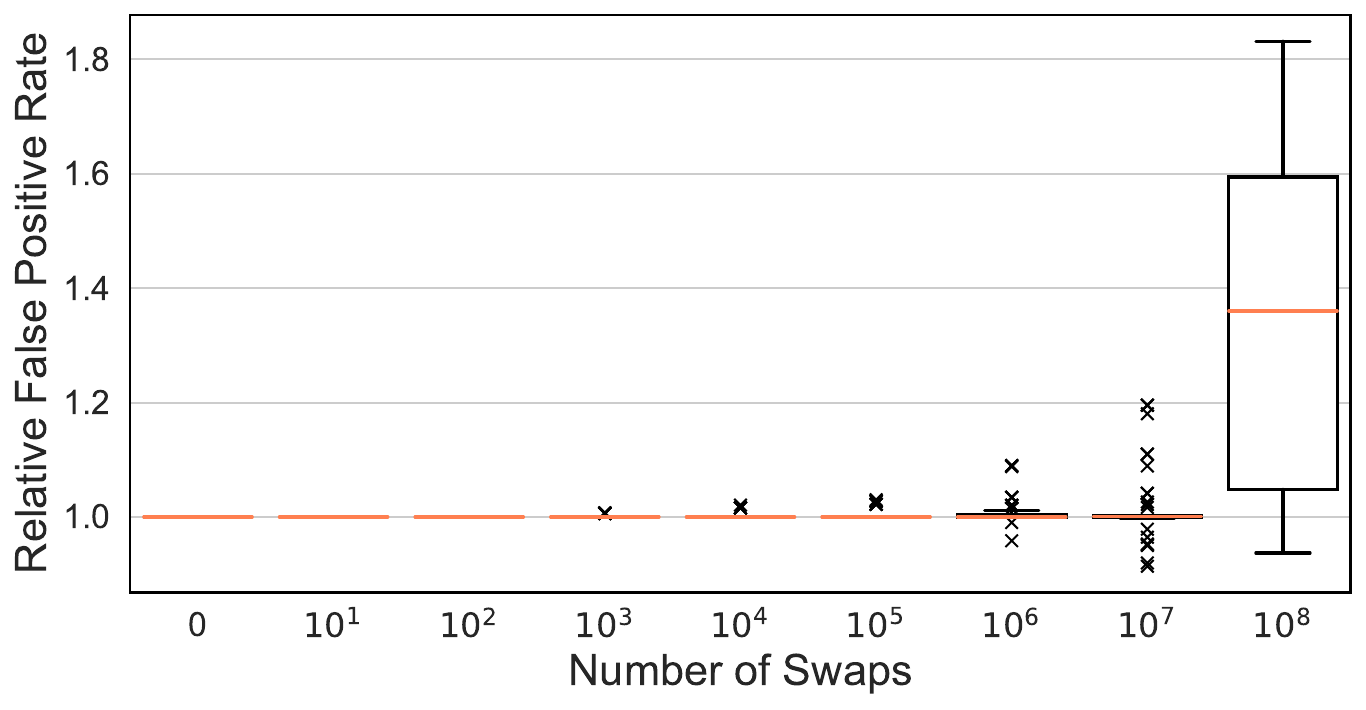}
    }
    \caption{Distribution of the relative FPR of fast PLBF++ and fast PLBF\# for each swap count, where the relative FPR is the ratio of each method’s FPR to that of the original PLBF. Values above 1 indicate degraded accuracy. With $10^8$ swaps, the relative FPR increases noticeably (up to $1.85\times$). For smaller swap counts, the degradation is minor (median $\le 1.01$ across 10 seeds).}
    \label{fig: diff}
\end{figure*}

This section evaluates the experimental performance of fast PLBF, fast PLBF++, and fast PLBF\#.
We compared the performances of our proposed methods with four baselines: Bloom filter~\cite{bloom1970space}, Ada-BF~\cite{dai2020adaptive}, sandwiched LBF~\cite{mitzenmacher2018model}, and PLBF~\cite{vaidya2021partitioned}. 
Similar to PLBF, Ada-BF is an LBF that partitions the score space into several regions and assigns different FPRs to each region.
However, Ada-BF relies heavily on heuristics for clustering and assigning FPRs.
Sandwiched LBF is a LBF that ``sandwiches'' a machine learning model with two Bloom filters.
This achieves better memory efficiency than the original LBF by optimizing the size of two Bloom filters.

To facilitate the comparison of different methods or hyperparameters results, we have slightly modified the original PLBF framework.
The original PLBF was designed to minimize memory usage under the condition of a given FPR.
However, this approach makes it difficult to compare the results of different methods or hyperparameters.
This is because both the FPR at test time and the memory usage vary depending on the method and hyperparameters, which often makes it difficult to determine the superiority of the results.
Therefore, in our experiments, we used a framework where the expected FPR is minimized under the condition of memory usage.
This makes it easy to obtain two results with the same memory usage and compare them by the FPR at test time.
See the appendix for more information on how this framework modification will change the construction method of PLBFs.

We evaluated the algorithms on \Rone{four} datasets.
The first is the \textbf{URLs} dataset~\cite{manu2021urlDataset}, which has been used in previous studies~\cite{dai2020adaptive,vaidya2021partitioned}; it contains 223,088 malicious (keys) and 428,118 benign (non-keys) URLs, from which we extracted 20 lexical features.
The second is the \textbf{EMBER} dataset~\cite{anderson2018ember}, which has also been used in the original PLBF study~\cite{vaidya2021partitioned}; it contains 300{,}000 malicious (keys) and 400{,}000 benign (non-keys) files with their features.
\Rone{The third is \textbf{DNA} dataset, which we constructed from chromosome 14 of the hg19 human reference genome provided by the UCSC Genome Browser~\cite{ucscHg19Chr14}; it contains 49{,}906{,}253 unique 14-mers observed in the chromosome sequence as keys and the same number of randomly sampled 14-mers that do not appear in the sequence as non-keys.
The fourth is \textbf{NIDS} dataset, which is based on the NF-ToN-IoT dataset~\cite{sarhan2020netflow}; after removing duplicate feature vectors, it contains 604{,}066 attack flows as keys and 198{,}449 benign flows as non-keys, with 10 flow-level features.}
For each dataset, we used all keys and a subset of non-keys for training, and used the remaining non-keys as the test set.

While any model can be used for the classifier, we used LightGBM~\cite{NIPS2017_6449f44a} because of its fast training and inference, memory efficiency, and accuracy.
\Rone{The sizes of the machine learning models for the URLs, EMBER, DNA, and NIDS datasets are 312 Kb, 1.19 Mb, 2.99 Mb, and 717 Kb, respectively.}
\Rone{The training times of the machine learning models for the URLs, EMBER, DNA, and NIDS datasets are 1.09, 2.71, 567.3, and 0.63 seconds, respectively.}
The memory usage of LBF is the memory usage of the backup Bloom filters plus the size of the machine learning model.
\cref{fig: dist} shows histograms of the score distributions for keys and non-keys in each dataset.
We can see that the frequency of keys increases and that of non-keys decreases as the score increases.
In addition, \cref{fig: g_to_h_ratio} plots $g_i/h_i ~ (i=1,2, \dots, N)$ when $N=1{,}000$.
We can see that $g_i/h_i$ tends to increase as $i$ increases, but the increase is not monotonic; that is, the score distribution is not \textit{ideal}.

\subsection{Construction Time}
\label{sec: Construction time}

We compared the construction times of our proposed methods with those of existing methods.
Following the experiments in the original PLBF paper, hyperparameters for PLBFs were set to $N=1{,}000$ and $k=5$.

\cref{fig: time_hist} shows the construction time of each method.
The construction time for LBFs includes not only the time to insert keys into the Bloom filters but also the time to train the machine learning model and the time to compute the optimal parameters ($\bm{t}$ and $\bm{f}$ in the case of PLBF).
\Rtwo{For the parameter-optimization time, we separately report the time spent on the DP computation and the remaining time, including the time for \textsc{OptimalFPR} computation.}

Ada-BF and sandwiched LBF use heuristics to find the optimal parameters, so they have shorter construction times than PLBF.
However, these methods demonstrate lower accuracy, as we will demonstrate in the following section.
PLBF has better accuracy but takes more than 3 minutes to find the optimal $\bm{t}$ and $\bm{f}$.
\Rone{In contrast, for the URLs, EMBER, and NIDS datasets, all of our proposed methods compute $\bm{t}$ and $\bm{f}$ within 2 seconds.
Specifically, fast PLBF, fast PLBF++, and fast PLBF\# construct the data structure up to 58.7, 74.5, and 75.7 times faster than PLBF, respectively.
The construction times of fast PLBF++ and fast PLBF\# are comparable to those of sandwiched LBF, which relies heavily on heuristics and, as we will see in the next section, is much less accurate than PLBF.
On the DNA dataset, the overall construction time is dominated by model-related preprocessing, including model training and the construction of the vectors $\bm{g}$ and $\bm{h}$ from the model scores.
Nevertheless, our proposed methods still reduce the total construction time by approximately 25\% on this dataset.}

\subsection{Memory Usage and FPR}
\label{sec: MemoryUsageAndFPR}

We compared the trade-off between memory usage and FPR for the proposed methods with Bloom filter, Ada-BF, sandwiched LBF, and PLBF.
Following the experiments in the original PLBF paper, hyperparameters for PLBFs were always set to $N=1{,}000$ and $k=5$.

\cref{fig: memory_fpr} illustrates the trade-off between memory usage and FPR for each method.
PLBF, fast PLBF, fast PLBF++, and fast PLBF\# consistently exhibit better Pareto curves across all datasets.
Since fast PLBF constructs the same data structure as PLBF, it always achieves the same accuracy.
Furthermore, fast PLBF++ and fast PLBF\# also construct exactly the same data structure as PLBF in all experimental settings and, consequently, always achieve the same accuracy.

\subsection{Ablation Study for Hyper-Parameters}
\label{ablation study for hyper-parameters}

The parameters of the PLBFs are total memory size, $N$, and $k$.
The total memory size is specified by the user, and $N$ and $k$ are hyperparameters that are determined by balancing construction time and accuracy.
In previous sections, we set $N$ to 1{,}000 and $k$ to 5, following the original paper on PLBF.
Here, we perform ablation studies for these hyperparameters to confirm that our proposed methods can construct accurate data structures quickly, no matter what hyperparameter settings are used.
We also confirm that the accuracy tends to be better and the construction time increases as $N$ and $k$ are increased, and that the proposed methods show a much slower increase in their construction time compared to PLBF.

\cref{fig: N_FPR} shows the test-set FPR with various $N$ while $k$ is fixed at 5.
For all four datasets, the FPR generally decreases as $N$ increases up to a moderate value and then stabilizes.
Note that this is the FPR on test data, so it does not necessarily decrease monotonically.
Fast PLBF and fast PLBF++ achieve the same FPR as PLBF for all common values of $N$ in this experiment.
Fast PLBF\# also remains very close to PLBF; its largest FPR ratio relative to PLBF over the common values of $N$ is 1.013.
\cref{fig: N_Time} shows the end-to-end construction time under the same experimental setting.
As $N$ increases, the PLBF construction time increases rapidly.
In contrast, the construction time of fast PLBF increases much more slowly, and those of fast PLBF++ and fast PLBF\# change only slightly even for large $N$.
This is because the construction time of PLBF is asymptotically proportional to $N^3$, while those of fast PLBF, fast PLBF++, and fast PLBF\# are proportional to $N^2$, $N \log N$, and $N$, respectively.
At $N=1{,}024$, fast PLBF, fast PLBF++, and fast PLBF\# construct the data structure up to 57.9, 81.5, and 82.7 times faster than PLBF, respectively, across the four datasets.
The experimental results show that the three proposed methods can achieve high accuracy without significantly changing the construction time with large $N$.

\cref{fig: K_FPR} shows the test-set FPR with various $k$ while $N$ is fixed at 1{,}000.
For all four datasets, the FPR generally decreases as $k$ increases.
The improvement is especially clear when $k$ is increased from small values, and the FPR tends to stabilize for larger $k$ on some datasets.
As in the experiments with varying $N$, this decrease is not necessarily monotonic because this is the FPR on test data.
Fast PLBF always has the same accuracy as PLBF.
Fast PLBF++ and fast PLBF\# have slightly different accuracy from PLBF, but the differences are small: over the common values of $k$, fast PLBF++ exhibits at most 1.023 times the FPR of PLBF, and fast PLBF\# exhibits at most 1.049 times the FPR of PLBF.
\cref{fig: K_Time} shows the end-to-end construction time under the same experimental setting.
The construction times of all four PLBFs increase with $k$, but fast PLBF has a much shorter construction time than PLBF, and fast PLBF++ and fast PLBF\# have even shorter construction times than fast PLBF.
When $k=50$, fast PLBF, fast PLBF++, and fast PLBF\# construct the data structure up to 235, 865, and 871 times faster than PLBF, respectively, across the four datasets.
The experimental results indicate that by increasing $k$, the three proposed methods can achieve high accuracy without significantly affecting the construction time.

\subsection{Robustness of Fast PLBF++ and Fast PLBF\# to Non-Ideal Distributions}
\label{sec: robustness of fast plbf pp}

\Rtwo{The theoretical guarantees for fast PLBF++ and fast PLBF\# rely on structural assumptions on the score distribution.
Under an ideal score distribution, both methods recover the same DP values as fast PLBF.
However, their guarantees under non-ideal score distributions are not symmetric: fast PLBF++ has the instance-specific bound in \cref{thm: diff DPKL}, whereas we do not provide an analogous non-ideal bound for fast PLBF\#.
Therefore, it is important to empirically examine whether fast PLBF\# degrades more than fast PLBF++ when the ideal monotonicity assumption is violated.}

\Rtwo{To this end, we use a controlled artificial setting in which the degree of non-monotonicity can be varied systematically.
The perturbation procedure follows the setting used in the appendix of our conference paper~\cite{sato2023fast}, where it was used to study the robustness of fast PLBF++.
In this journal version, we move this evaluation to the main text and extend it to fast PLBF\#, which is newly introduced in this paper.
As shown below, although fast PLBF\# does not have a non-ideal theoretical bound analogous to fast PLBF++, its empirical degradation from fast PLBF in terms of FPR is comparable to that of fast PLBF++ over a wide range of controlled non-ideal perturbations.}

We first describe how the artificial datasets are generated.
The construction consists of two steps.
First, as in the original PLBF paper~\cite{vaidya2021partitioned}, the key and non-key score distribution is generated using the Zipfian distribution.
\cref{fig: Arti_0_0_log_hist} shows a histogram of the distribution of key and non-key scores at this time, and \cref{fig: Arti_0_0_pos_neg_ratio} shows $g_i/h_i ~ (i=1,2, \dots, N)$ when $N=1{,}000$.
This score distribution is \textit{ideal}.
Next, we perform \textit{swaps} to add non-monotonicity to the score distribution.
\textit{Swap} refers to changing the scores of the elements in the two adjacent segments so that the number of keys and non-keys in the two segments are swapped.
Namely, an integer $i$ is randomly selected from $\{1, 2, \dots, N-1 \}$, and the scores of elements in the $i$-th segment are changed so that they are included in the $(i+1)$-th segment, and the scores of elements in the $(i+1)$-th segment are changed to include them in the $i$-th segment.
\cref{fig: Arti_10_0_log_hist,fig: Arti_100_0_log_hist,fig: Arti_1000_0_log_hist,fig: Arti_10000_0_log_hist,fig: Arti_100000_0_log_hist,fig: Arti_1000000_0_log_hist,fig: Arti_10000000_0_log_hist,fig: Arti_100000000_0_log_hist} and \cref{fig: Arti_10_0_pos_neg_ratio} and \cref{fig: Arti_100_0_pos_neg_ratio,fig: Arti_1000_0_pos_neg_ratio,fig: Arti_10000_0_pos_neg_ratio,fig: Arti_100000_0_pos_neg_ratio,fig: Arti_1000000_0_pos_neg_ratio,fig: Arti_10000000_0_pos_neg_ratio,fig: Arti_100000000_0_pos_neg_ratio} show the histograms of the score distribution and $g_i/h_i ~ (i=1,2, \dots, N)$ for $10, 10^2, \dots, 10^8$ swaps, respectively.
It can be seen that as the number of swaps increases, the score distribution becomes more non-monotonic.
For each case of the number of swaps, 10 different datasets were created using 10 different seeds.

\cref{fig: Arti_Memory_and_FPR} shows the accuracy of each method for each number of swaps, with the seed set to 0.
Hyperparameters for PLBFs are set to $N=1{,}000$ and $k=5$.
It can be seen that fast PLBF and fast PLBF++ achieve better Pareto curves than the other methods for all datasets.
It can also be seen that the higher the number of swaps, the more often there is a difference between the accuracy of fast PLBF++/\# and fast PLBF.

\cref{fig: diff} shows the difference in accuracy between fast PLBF++/\# and fast PLBF for each swap count.
Here, the ``relative false positive rate'' is the false positive rate of fast PLBF++/\# divided by that of PLBF constructed with the same hyperparameters and conditions.
We created 10 different datasets for each swap count and conducted 6 experiments for each dataset and method with a memory usage of 0.25Mb, 0.5Mb, 0.75Mb, 1.0Mb, 1.25Mb, and 1.5Mb.
Namely, we compared the false positive rates of fast PLBF and fast PLBF++/\# under 60 conditions for each swap count.

The result shows that fast PLBF++ consistently achieves the same accuracy as PLBF in a total of $240$ experiments where the number of swaps is $10^3$ or less.
Similarly, fast PLBF\# consistently achieves the same accuracy as PLBF in a total of $180$ experiments where the number of swaps is $10^2$ or less.
Furthermore, in the experiments where the number of swaps is $10^7$ or less, the relative false positive rate is generally less than 1.1, with the exception of 14 cases for fast PLBF++ and 7 cases for fast PLBF\#. 
However, in the cases of $10^8$ swap counts, where there is almost no monotonicity in the score distribution, the false positive rate for fast PLBF++ and fast PLBF\# is up to 1.85 times and 1.83 times that for PLBF, respectively.

\section{Discussion}
\label{sec: discussion}

We discuss several practical and theoretical implications of our results.

\Rone{\paragraph{On stronger approximation guarantees.}
We have provided a non-ideal guarantee for fast PLBF++, but two important open questions remain.
The first is whether one can obtain an instance-independent guarantee. Our current guarantee is stated in terms of $\delta_{\max}$, a quantity that depends on the empirical score distributions induced by the dataset and the learned model.
Establishing a data-independent approximation bound remains an important open problem.}
\Rtwo{The second is whether an analogous guarantee can be obtained for fast PLBF\#.
This appears technically challenging because, once total monotonicity is violated, controlling the effect of the columns eliminated by SMAWK becomes difficult; this seems substantially harder than in the monotone-maxima procedure used by fast PLBF++.}
\Rone{Nevertheless, our empirical results suggest that both fast PLBF++ and fast PLBF\# are robust to violations of the monotonicity condition.
Indeed, in the artificial swap experiments, both methods often remain close to fast PLBF in terms of FPR, even when the score distribution is highly non-monotone.
These observations suggest that stronger approximation guarantees may be possible, although obtaining tight data-independent bounds would likely require new analytical techniques.}

\Rone{\paragraph{On the choice of $N$ and $k$.}
Our proposed acceleration broadens the range of practical choices for the PLBF hyperparameters $N$ and $k$.
A larger $N$ enables a finer-grained choice of thresholds, while a larger $k$ allows more score regions; hence, increasing these values can potentially improve the optimal value of the PLBF objective function in \cref{equ:prob}.
On the other hand, increasing $N$ and $k$ has two drawbacks.
The first is the increase in construction time, and the second is the potential overfitting to the empirical non-key distribution.
Since our methods mitigate the first drawback, it becomes important to carefully consider the second drawback.}

\Rone{We provide theoretical guidance for choosing $N$ and $k$; see \cref{sec: theoretical guidance for choosing n and k}.
This analysis follows the general approach used by Mitzenmacher et al.~\cite{mitzenmacher2018model} and extends it to the PLBF setting.
In particular, we derive a conservative sufficient condition under which overfitting to the construction-time non-keys is unlikely.
The resulting bound suggests that, for datasets of typical size, overfitting is unlikely when $N \approx 1000$ and $k \approx 5$, which are the values used in the original PLBF paper and are also adopted as the default setting in our experiments.
Thus, we regard $N \approx 1000$ and $k \approx 5$ as reasonable starting points.
We emphasize, however, that this bound should not be interpreted as a strict upper limit on admissible values of $N$ and $k$; rather, it should be viewed as a safe theoretical guideline based on a conservative argument.
Indeed, our experiments suggest that larger values of $N$ and $k$ can also be used in practice and may even further improve the test-time FPR in some cases.}

\Rone{In practice, it is also important to prepare a validation set of non-keys and use it to select the hyperparameters $N$ and $k$ from multiple candidate pairs.
Such validation-based selection can help reduce overfitting to the construction-time non-keys.
Our acceleration is particularly useful in this strategy, as it makes such hyperparameter search substantially faster.}

\Rtwo{\paragraph{On the choice among fast PLBF variants.}
Fast PLBF is the safest choice when its construction time is acceptable, because it constructs exactly the same data structure as PLBF.
Fast PLBF++ and fast PLBF\# are useful when faster construction is required.
Although these methods do not necessarily produce the same result as PLBF in general, they are exact under ideal score distributions and empirically achieve almost the same memory-FPR trade-off as PLBF on the evaluated datasets.
Between them, fast PLBF\# has the best asymptotic construction time and is slightly faster in our experiments.
On the other hand, fast PLBF++ has a theoretical bound for non-ideal cases in terms of $\delta_{\max}$, whereas we do not provide an analogous bound for fast PLBF\#.
Therefore, if a practitioner requires a theoretical justification beyond the ideal case, or if $\delta_{\max}$ is found to be sufficiently small, fast PLBF++ is preferable.
If construction time is a particularly important factor and validation performance is satisfactory, fast PLBF\# is also a practical choice.}

\section{Conclusion}
\label{sec: Conclusion}

PLBF is an outstanding LBF that can effectively leverage the distribution of the set and queries captured by a machine learning model.
However, PLBF is computationally expensive to construct.
We proposed fast PLBF, fast PLBF++, and fast PLBF\# to solve this problem.
Fast PLBF is superior to PLBF because fast PLBF constructs exactly the same data structure as PLBF but does so faster.
Fast PLBF++ and fast PLBF\# are even faster than fast PLBF and achieve almost the same accuracy as PLBF and fast PLBF.
These proposed methods have greatly expanded the range of applications of PLBF.

\clearpage

\appendix

\section{Solution of the Optimization Problems}
\label{app:SolutionOptProb}

In the original PLBF paper~\cite{vaidya2021partitioned}, the analysis with mathematical expressions was conducted only for the relaxed problem, while no analysis utilizing mathematical expressions was performed for the general problem.
Consequently, it was unclear which calculations were redundant, leading to the repetition of constructing similar DP tables.
Hence, this appendix provides a comprehensive analysis of the general problem.
It presents the optimal solution and the corresponding value of the objective function using mathematical expressions.
This analysis uncovers redundant calculations, enabling the derivation of fast PLBF.

The optimization problem designed by PLBF~\cite{vaidya2021partitioned} can be written as follows:
\begin{equation}
\label{equ:prob_appendix}
\begin{aligned}
& \underset{\bm{f}, \bm{t}} {\text{minimize}} && 
\sum_{i=1}^{k} c |\mathcal{S}| G_i \log_{2}\left(\frac{1}{f_i}\right) \\
&\text{subject to} && \sum_{i=1}^{k} H_i f_i \leq F \\
& && t_0 = 0 < t_1 < t_2 < \dots < t_k = 1 \\
& && f_i \leq 1 ~~~~~~ (i=1,2,\dots, k).
\end{aligned}
\end{equation}
The objective function represents the total memory usage of the backup Bloom filters.
The first constraint equation represents the condition that the expected overall FPR is below $F$.
$c \geq 1$ is a constant determined by the type of backup Bloom filters.
$\bm{G}$ and $\bm{H}$ are determined by $\bm{t}$.
Once $\bm{t}$ is fixed, we can solve this optimization problem to find the optimal $\bm{f}$ and the minimum value of the objective function.
In the original PLBF paper~\cite{vaidya2021partitioned}, the condition $f_i \leq 1 ~ (i=1,2,\dots, k)$ was relaxed; however, here we conduct the analysis without relaxation.

First, in the case of $\sum_{i \in \{j \mid G_j > 0\}} H_i \leq F$, the minimum value of the objective function is clearly $0$. This is because if we set 
\begin{equation}
    f_i = 
    \begin{dcases}
    0 & (G_i = 0) \\
    1 & (G_i > 0),
    \end{dcases}
\end{equation}
the conditions of the optimization problem are satisfied, and the value of the objective function (which is always non-negative) becomes 0. 
In addition, even assuming that $H_i = 0 \Rightarrow f_i = 1$, it has no effect on the solution to the optimization problem. This is because even if there is an $i$ that satisfies $H_i=0 ~ \land ~ f_i<1$, the $\bm{f}$ that can be obtained by changing $f_i$ to 1 satisfies the conditions of the optimization problem and has the same or a smaller value for the objective function.
Therefore, in the following, we will consider the case of $\sum_{i \in \{j \mid G_j > 0\}} H_i > F$, and assume that $H_i = 0 \Rightarrow f_i = 1$.

The Lagrange function is defined using the Lagrange multipliers $\bm{\mu}\in\mathbb{R}^k$ and $\nu\in\mathbb{R}$, as follows:
\begin{equation}
\begin{aligned}
L(\bm{f}, \bm{\mu}, \nu) &= 
\sum_{i=1}^{k} c |\mathcal{S}| G_i  \log_{2}\left(\frac{1}{f_i}\right)
+
\sum_{i=1}^{k} \mu_{i}(f_i - 1) \\
&\quad + 
\nu \left( \sum_{i=1}^{k} H_i f_i - F \right).
\end{aligned}
\end{equation}

From the KKT condition, there exist $\Bar{\bm{\mu}}$ and $\Bar{\nu}$ in the local optimal solution $\Bar{\bm{f}}$ of this optimization problem, and the following holds:
\begin{equation}
\label{equ:L/f}
\frac{\partial L}{\partial f_i} (\Bar{\bm{f}}, \Bar{\bm{\mu}}, \Bar{\nu}) = 0 ~~~~~~ (i=1,2, \dots, k)
\end{equation}
\begin{equation}
\label{equ:L/mu}
\Bar{f}_{i} - 1 \leq 0, ~~~ \Bar{\mu}_{i} \geq 0, ~~~ \Bar{\mu}_{i}(\Bar{f}_{i}-1) = 0 ~~~~~~ (i=1,2, \dots, k)
\end{equation}
\begin{equation}
\label{equ:L/nu}
\sum_{i=1}^{k} H_i \Bar{f}_i - F \leq 0, ~~~ \Bar{\nu} \geq 0, ~~~ \Bar{\nu}\left(\sum_{i=1}^{k} H_i \Bar{f}_i - F\right) = 0.
\end{equation}
By introducing $\mathcal{I}_{f=1}=\{i|\Bar{f}_{i}=1\}$, \cref{equ:L/f,equ:L/mu,equ:L/nu} can be organized as follows:
\begin{equation}
\label{equ:f=1}
c |\mathcal{S}| G_i = \Bar{\mu}_{i} + \Bar{\nu} H_i, ~~~ \Bar{\mu}_{i}\geq0, ~~~ \Bar{f}_i=1 ~~~~~~ (i\in\mathcal{I}_{f=1}),
\end{equation}
\begin{equation}
\label{equ:f<1}
c |\mathcal{S}| G_i = \Bar{f}_i \Bar{\nu} H_i, ~~~ \Bar{\mu}_{i}=0, ~~~ \Bar{f}_i<1 ~~~~~~ (i\notin\mathcal{I}_{f=1}),
\end{equation}
\begin{equation}
\begin{aligned}
\label{equ:F}
& \sum_{i\in\mathcal{I}_{f=1}} H_i + \sum_{i\notin\mathcal{I}_{f=1}} H_i \Bar{f}_i - F \leq 0, ~~~
\Bar{\nu} \geq 0, \\
& \Bar{\nu}\left(\sum_{i\in\mathcal{I}_{f=1}} H_i + \sum_{i\notin\mathcal{I}_{f=1}} H_i \Bar{f}_i - F\right) = 0.
\end{aligned}
\end{equation}

Here, we can prove $\Bar{\nu}>0$ by contradiction.
Assuming $\Bar{\nu}=0$, \cref{equ:f<1} implies that $\Bar{f}_i < 1 \Rightarrow G_i=0$; that is, $G_i > 0 \Rightarrow \Bar{f}_i = 1$.
However, this does not satisfy the conditions of the optimization problem, as $\sum_{i=1}^{k} H_i \Bar{f}_i \geq \sum_{i \in \{j \mid G_j > 0\}} H_i \Bar{f}_i = \sum_{i \in \{j \mid G_j > 0\}} H_i > F$.
Therefore, we have $\Bar{\nu}>0$.

Noting that we are assuming $H_i = 0 \Rightarrow f_i = 1$, we obtain the following from equations $\Bar{\nu}>0$ and \cref{equ:f<1}:
\begin{equation}
\label{equ:fi}
\Bar{f}_i = \frac{c|\mathcal{S}|}{\Bar{\nu}} \frac{G_i}{H_i} ~~~~~~ (i\notin\mathcal{I}_{f=1}).
\end{equation}
From \cref{equ:F,equ:fi} and $\Bar{\nu}>0$, we obtain
\begin{equation}
\label{equ:nu}
\frac{c|\mathcal{S}|}{\Bar{\nu}} = \frac{F - H_{f=1}}{1 - G_{f=1}},
\end{equation}
where
\begin{equation}
    G_{f=1} = \sum_{i \in \mathcal{I}_{f=1}} G_i, ~~ H_{f=1} = \sum_{i \in \mathcal{I}_{f=1}} H_i.
\end{equation}
Thus, we get
\begin{equation}
\label{equ:bar_f}
    \Bar{f}_i = 
    \begin{dcases}
    1 & (i \in \mathcal{I}_{f=1}) \\
    \frac{(F - H_{f=1})G_i}{(1 - G_{f=1})H_i} & (i \notin \mathcal{I}_{f=1}).
    \end{dcases}
\end{equation}
Substituting this into the objective function of the optimization problem (\cref{equ:prob_appendix}), we obtain
\begin{align}
  &~ \sum_{i=1}^{k} c|\mathcal{S}| G_i  \log_{2}\left(\frac{1}{\Bar{f}_i}\right) \\
  =&~ \sum_{i\notin\mathcal{I}_{f=1}} - c|\mathcal{S}| G_i  \log_{2}\left(
    \frac{F - \sum_{j\in\mathcal{I}_{f=1}} H_j}{1 - \sum_{j\in\mathcal{I}_{f=1}} G_j}  \frac{G_i}{H_i}
  \right) \\
  =&~ c|\mathcal{S}| (1-G_{f=1})\log_{2}\left(
    \frac{1-G_{f=1}}{F-H_{f=1}}
  \right) \nonumber\\
  &\quad -
  c|\mathcal{S}| \sum_{i\notin\mathcal{I}_{f=1}} G_i \log_{2}\left(
    \frac{G_i}{H_i}
  \right).
\end{align}

We can also obtain the conditions of $\mathcal{I}_{f=1}$.
First, from the assumption, 
\begin{equation}
\label{equ:H_i=0_and_I_f=1}
    H_i = 0 ~ \Rightarrow ~ i \in \mathcal{I}_{f=1}.
\end{equation}
Using $\Bar{\nu}>0, c>0,|\mathcal{S}|>0$, \cref{equ:f=1,equ:nu} derives
\begin{equation}
\label{equ:H_i>0_and_I_f=1}
    i\in\mathcal{I}_{f=1} ~ \land ~ H_i > 0 ~ \Rightarrow ~ \frac{G_i}{H_i} \geq \frac{\Bar{\nu}}{c|\mathcal{S}|} = \frac{1 - G_{f=1}}{F - H_{f=1}},
\end{equation}
and  \cref{equ:f<1,equ:nu} derives
\begin{equation}
\label{equ:not_I_f=1}
    i\notin\mathcal{I}_{f=1} ~ \Rightarrow ~ \frac{G_i}{H_i} < \frac{\Bar{\nu}}{c|\mathcal{S}|} = \frac{1 - G_{f=1}}{F - H_{f=1}}.
\end{equation}

\section{Algorithm Details}
\label{app:AlgorithmDetails}

\begin{algorithm}[t]
    \caption{PLBF~\cite{vaidya2021partitioned}}
    \label{alg: PLBF}
    \begin{algorithmic}
        \State \textbf{Input:}
        \State \hspace{1em} $\bm{g} \in \mathbb{R}^N$: key occurrence probability per segment
        \State \hspace{1em} $\bm{h} \in \mathbb{R}^N$: non-key occurrence probability per segment
        \State \hspace{1em} $F \in (0,1)$: target overall FPR
        \State \hspace{1em} $k \in \mathbb{N}$: number of regions
        \State \textbf{Output:}
        \State \hspace{1em} \Rone{$\bm{t} \in [0,1]^{k+1}$}: threshold boundaries of each region
        \State \hspace{1em} \Rone{$\bm{f} \in (0,1]^{k}$}: FPRs of each region
        \State \textbf{Algorithm:}
        \State \hspace{1em} \textsc{ThresMaxDivDP}$(\bm{g}, \bm{h}, j, k)$: Builds $\mathrm{DP}^j_{\mathrm{KL}}$ and recovers optimal thresholds by backtracking from $\mathrm{DP}_\mathrm{KL}^{j}[j-1][k-1]$
        \State \hspace{1em} \textsc{OptimalFPR}$(\bm{g}, \bm{h}, \bm{t}, F, k)$: Returns the optimal FPRs for each region under the given thresholds
        \State \hspace{1em} \textsc{SpaceUsed}$(\bm{g}, \bm{h}, \bm{t}, \bm{f})$: Returns the total size of the backup Bloom filters for the given thresholds and FPRs

        \State

        \State $\mathrm{MinSpaceUsed} \gets \infty$
        \State $\bm{t}_\mathrm{best} \gets \text{None}, ~ \bm{f}_\mathrm{best} \gets \text{None}$
        \For{$j = k, k+1, \dots, N$}
            \State $\bm{t} \gets \textsc{ThresMaxDivDP}(\bm{g}, \bm{h}, j, k)$
            \State $\bm{f} \gets \textsc{OptimalFPR}(\bm{g}, \bm{h}, \bm{t}, F, k)$
            \If{$\mathrm{MinSpaceUsed} > \textsc{SpaceUsed}(\bm{g}, \bm{h}, \bm{t}, \bm{f})$}
                \State $\mathrm{MinSpaceUsed} \gets \textsc{SpaceUsed}(\bm{g}, \bm{h}, \bm{t}, \bm{f})$
                \State $\bm{t}_\mathrm{best} \gets \bm{t}, ~ \bm{f}_\mathrm{best} \gets \bm{f}$
            \EndIf
        \EndFor
        \State \Return $\bm{t}_\mathrm{best}, \bm{f}_\mathrm{best}$
    \end{algorithmic}
\end{algorithm}

\begin{algorithm}[t]
    \caption{OptimalFPR (Algorithm 1 in \cite{vaidya2021partitioned})}
    \label{alg: OptimalFPR}
    \begin{algorithmic}
        \State \textbf{Input:}
        \State \hspace{1em} $\bm{G} \in \mathbb{R}^k$: key occurrence probability per region
        \State \hspace{1em} $\bm{H} \in \mathbb{R}^k$: non-key occurrence probability per region
        \State \hspace{1em} $F \in (0,1)$: target overall FPR
        \State \hspace{1em} $k \in \mathbb{N}$: number of regions
        \State \textbf{Output:}
        \State \hspace{1em} \Rone{$\bm{f} \in (0,1]^{k}$}: FPRs of each region

        \State

        \For{$i = 1, 2, \dots, k$}
            \State $f_i \gets G_i F / H_i$
        \EndFor
        \While{$\exists ~ i : f_i > 1$}
            \State $G_{f=1} \gets 0, H_{f=1} \gets 0$
            \For{$i = 1, 2, \dots, k$}
                \If{$f_i \geq 1$}
                    \State $f_i \gets 1$
                    \State $G_{f=1} \gets G_{f=1} + G_i, ~ H_{f=1} \gets H_{f=1} + H_i$
                \EndIf
            \EndFor
            \For{$i = 1, 2, \dots, k$}
                \If{$f_i < 1$}
                    \State $f_i \gets \frac{(F - H_{f=1}) G_i}{(1 - G_{f=1}) H_i}$
                \EndIf
            \EndFor
        \EndWhile
        \State \Return $\bm{f}$
    \end{algorithmic}
\end{algorithm}

\begin{algorithm}[t]
    \caption{Fast PLBF}
    \label{alg: FastPLBF}
    \begin{algorithmic}
        \State \textbf{Algorithm:}
        \State \hspace{1em} $\textsc{MaxDivDP}(\bm{g}, \bm{h}, N, k)$ : Builds $\mathrm{DP}_\mathrm{KL}^{N}$
        \State \hspace{1em} $\textsc{ThresMaxDiv}(\mathrm{DP}_\mathrm{KL}^{N}, j, k)$ : Traces the transitions backward and finds the optimal thresholds

        \State

        \State $\mathrm{MinSpaceUsed} \gets \infty$
        \State $\bm{t}_\mathrm{best} \gets \text{None}, ~ \bm{f}_\mathrm{best} \gets \text{None}$
        \State $\mathrm{DP}_\mathrm{KL}^{N} \gets \textsc{MaxDivDP}(\bm{g}, \bm{h}, N, k)$
        \For{$j = k, k+1, \dots, N$}
            \State $\bm{t} \gets \textsc{ThresMaxDiv}(\mathrm{DP}_\mathrm{KL}^{N}, j, k)$
            \State $\bm{f} \gets \textsc{OptimalFPR}(\bm{g},\bm{h},\bm{t},F,k)$
            \If{$\mathrm{MinSpaceUsed} > \textsc{SpaceUsed}(\bm{g},\bm{h},\bm{t},\bm{f})$}
                \State $\mathrm{MinSpaceUsed} \gets \textsc{SpaceUsed}(\bm{g},\bm{h},\bm{t},\bm{f})$
                \State $\bm{t}_\mathrm{best} \gets \bm{t}, ~ \bm{f}_\mathrm{best} \gets \bm{f}$
            \EndIf
        \EndFor
        \State \Return $\bm{t}_\mathrm{best}, \bm{f}_\mathrm{best}$
    \end{algorithmic}
\end{algorithm}

In this appendix, we explain the algorithms for PLBF and the proposed method using pseudo-code in detail.
This will help to clarify how each method differs from the others.

First, we show the PLBF algorithm in \cref{alg: PLBF}, and the \textsc{OptimalFPR} algorithm, which is used by PLBF, in \cref{alg: OptimalFPR}.
For details of \textsc{SpaceUsed}, please refer to Equation (2) in the original PLBF paper~\cite{vaidya2021partitioned}.
The worst case time complexity of \textsc{OptimalFPR} is $\mathcal{O}(k^2)$ because this algorithm repeatedly solves the relaxation problem up to $k$ times with a complexity of $\mathcal{O}(k)$ until there are no more $i$s such that $f_i > 1$.
The time complexity of \textsc{SpaceUsed} is $\mathcal{O}(k)$.
As the DP table is constructed with a time complexity of $\mathcal{O}(j^2 k)$ for each $j=k,k+1, \dots, N$, the overall complexity is $\mathcal{O}(N^3k)$.

Then, we show the fast PLBF algorithm in \cref{alg: FastPLBF}. 
The time complexity of building $\mathrm{DP}_\mathrm{KL}^{N}$ is $\mathcal{O}(N^2k)$, and the worst-case complexity of subsequent computations is $\mathcal{O}(Nk^2)$.
Because $N > k$, the total complexity is $\mathcal{O}(N^2k)$, which is faster than $\mathcal{O}(N^3k)$ for PLBF, although fast PLBF constructs the same data structure as PLBF.

\section{Modification of the PLBF Framework}
\label{app: modification of the PLBF framework}

Here, we describe the framework modifications we made to PLBF in our experiments.
In the original PLBF paper~\cite{vaidya2021partitioned}, the optimization problem was designed to minimize the amount of memory usage under a given target FPR (\cref{equ:prob_appendix}).
However, this framework makes it difficult to compare the results of different methods and hyperparameters.
Therefore, we designed the following optimization problem, which minimizes the expected FPR under a given memory usage condition:
\begin{equation}
\label{equ:modified_prob}
\begin{aligned}
& \underset{\bm{f}, \bm{t}} {\text{minimize}} && 
\sum_{i=1}^{k} H_i f_i \\
&\text{subject to} &&
\sum_{i=1}^{k} c|\mathcal{S}| G_i \log_{2}\left(\frac{1}{f_i}\right) \leq M \\
& && t_0 = 0 < t_1 < t_2 < \dots < t_k = 1 \\
& && f_i \leq 1 ~~~~~~ (i=1,2,\dots, k),
\end{aligned}
\end{equation}
where $M$ is a parameter that is set by the user to determine the upper bound of memory usage.

As in \cref{app:SolutionOptProb}, we introduce the index set $\mathcal{I}_{f=1}$.
Then, the optimal FPR vector $\bar{\bm{f}}$ is given by
\begin{equation}
\label{equ:fi_conditioned_by_M}
\Bar{f}_i = 
    \begin{dcases}
    1 & (i \in \mathcal{I}_{f=1}) \\
    2^{-\beta} \frac{G_i}{H_i} & (i \notin \mathcal{I}_{f=1})
    \end{dcases},
\end{equation}
where
\begin{equation}
\label{equ:beta}
    \beta = \frac{M}{c|\mathcal{S}|\left(1 - G_{f=1}\right)} + \frac{1}{1 - G_{f=1}} \sum_{i\notin\mathcal{I}_{f=1}} G_i \log_{2} \left(\frac{G_i}{H_i}\right).
\end{equation}
When $\bm{f} = \Bar{\bm{f}}$, the expected FPR, that is, the objective function of \cref{equ:modified_prob} is 
\begin{equation}
\label{equ:expected FPR}
    H_{f=1} + 2^{-\beta} (1 - G_{f=1}).
\end{equation}
Therefore, as in the original framework, we can find the best $\bm{t}$ and $\bm{f}$ by finding a way to cluster the $1$st to $(j-1)$-th segments into $k-1$ regions that maximizes $\sum_{i=1}^{k-1} G_i \log_{2}\left(G_i / H_i\right)$ for each $j=k,k+1,\dots,N$.
We can find it using the same DP algorithm and its acceleration methods as in the original framework.

\section{\texorpdfstring{\Rone{Theoretical Guidance for Choosing $N$ and $k$}}{Theoretical Guidance for Choosing N and k}}
\label{sec: theoretical guidance for choosing n and k}

\Rone{
Here, we provide theoretical guidance for choosing two important hyperparameters of PLBF and our proposed methods: the number of segments $N$ and the number of regions $k$.
From the viewpoint of the optimization problem in \cref{equ:prob}, larger values of $N$ and $k$ are preferable.
This is because increasing $N$ makes the threshold vector $\bm{t}$ more fine-grained, while increasing $k$ increases the number of regions and thus allows the optimization to search over a larger solution space.
Therefore, larger values of $N$ and $k$ tend to improve the optimal value of \cref{equ:prob}.

However, increasing $N$ and $k$ has two drawbacks.
The first is that the construction time of PLBF increases.
The second is that overfitting may occur, causing the test-time FPR to become significantly worse than the target FPR.
The first drawback is substantially mitigated by the methods proposed in \cref{sec: methods}, which allows us to construct PLBFs with much larger values of $N$ and $k$ than before within practical running time.
Consequently, the second drawback becomes more important: we need a theoretical guarantee that clarifies how large $N$ and $k$ can be while keeping overfitting unlikely.
To this end, we provide the following theorem.
\begin{theorem}
\label{thm:guide_n_and_k}
Assume that the set of non-keys $\mathcal{Q}$ used for constructing the PLBF is generated from the same distribution as the test-time non-keys.
Let $\bar{F}$ denote the expected test-time FPR of the PLBF.
Then, for any $\tau \in (0,1)$, the following inequality holds:
\begin{equation}
\label{equ:guide_n_and_k}
\Pr\left[\bar{F}>F+\tau\right]
\leq
\binom{N-1}{k-1}2^k\exp(-2|\mathcal{Q}|\tau^2).
\end{equation}
\end{theorem}
\begin{proof}
Let $\mathcal{T}_{N,k}$ denote the set of threshold vectors obtained by partitioning the $N$ segments into $k$ regions.
Then, $|\mathcal{T}_{N,k}|=\binom{N-1}{k-1}$.

For a fixed $\bm{t}\in\mathcal{T}_{N,k}$, let $\bar H_i^{\bm{t}}$ and $H_i^{\bm{t}}$ be the true and empirical (i.e., measured using $\mathcal{Q}$) non-key masses of the $i$-th region, respectively.
Also, let $f_i^{\bm{t}}\in(0,1]$ be the FPR assigned to the $i$-th region under $\bm{t}$.
From the constraints of \cref{equ:prob}, the empirical overall FPR satisfies $\sum_{i=1}^{k} H_i^{\bm{t}} f_i^{\bm{t}} \leq F$.
The true expected overall FPR under this fixed threshold vector $\bm{t}$ is $\bar F^{\bm{t}} \coloneq \sum_{i=1}^{k}\bar H_i^{\bm{t}} f_i^{\bm{t}}$.
Therefore,
\begin{equation}
\bar F^{\bm{t}} - F
\leq
\sum_{i=1}^{k}
\left(\bar H_i^{\bm{t}}-H_i^{\bm{t}}\right)f_i^{\bm{t}}
\leq
\mathrm{TV}
\left(\bar H^{\bm{t}},H^{\bm{t}}\right),
\end{equation}
where $\mathrm{TV}(\bar H^{\bm{t}},H^{\bm{t}})$ is the total variation distance between $\bar H^{\bm{t}}$ and $H^{\bm{t}}$.
For any fixed $\bm{t}$, $H^{\bm{t}}$ is the empirical distribution of a $k$-category multinomial distribution.
Thus,
\begin{equation}
\label{equ:total variation distance}
\Pr\left[
    \mathrm{TV}\left(\bar H^{\bm{t}},H^{\bm{t}}\right)>\tau
\right]
\leq
2^k\exp\left(-2|\mathcal{Q}|\tau^2\right).
\end{equation}

Now, let $\bm{t}^{*}$ be the threshold vector finally selected by PLBF construction.
From the definition, $\bar F = \bar F^{\bm{t}^{*}}$.
If $\bar F>F+\tau$, then there exists some $\bm{t}\in\mathcal{T}{N,k}$ such that $\bar F^{\bm{t}} > F+\tau$.
Therefore, by applying \cref{equ:total variation distance} and union bound, we obtain
\begin{equation}
\Pr\left[\bar F>F+\tau\right]
\leq
\binom{N-1}{k-1}
2^k\exp\left(-2|\mathcal{Q}|\tau^2\right)
\end{equation}
This completes the proof.
\end{proof}

\begin{figure}[t]
    \centering
    \includegraphics[width=\columnwidth]{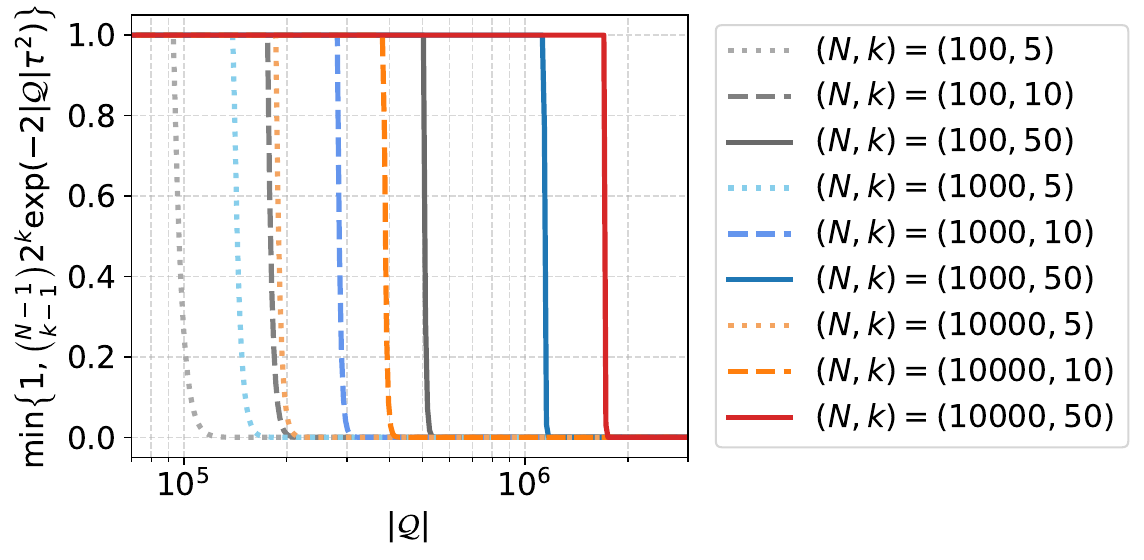}
    \caption{\Rone{Values of the right-hand side of \cref{equ:guide_n_and_k} as a function of $|\mathcal{Q}|$ for several choices of $N$ and $k$, with $\tau=0.01$.}}
    \label{fig: guide n and k}
\end{figure}

\cref{fig: guide n and k} illustrates how the right-hand side of \cref{equ:guide_n_and_k} changes as a function of $|\mathcal{Q}|$ for several choices of $N$ and $k$.
In this figure, we fix the tolerance parameter to $\tau=0.01$, which corresponds to allowing a test-time FPR increase of at most one percentage point over the target FPR $F$.
For example, for the setting $N=1000$ and $k=5$, which was used as a baseline in the original PLBF paper, the bound becomes less than $0.01$ when $|\mathcal{Q}| \geq 1.7\times 10^5$.
This condition is satisfied in our experimental setting (\cref{sec: experiments}).
The figure also shows that larger values of $N$ and $k$ require a larger non-key set $\mathcal{Q}$ to obtain the same level of statistical guarantee.
For instance, when $N=10000$ and $k=50$, the bound becomes less than $0.01$ around $|\mathcal{Q}| \geq 1.8\times 10^6$.
Thus, \cref{equ:guide_n_and_k} provides a simple way to choose $N$ and $k$: given the available size of $\mathcal{Q}$ and the desired tolerance $\tau$, one can select $N$ and $k$ as large as possible while keeping the upper bound sufficiently small.

We note that this guarantee should be interpreted as a conservative sufficient condition rather than a necessary condition.
In particular, even when the right-hand side of \cref{equ:guide_n_and_k} is not small, the test-time FPR may still remain close to the empirical FPR in practice.
Indeed, in our experiments, we observe that increasing $N$ and $k$ beyond the range suggested by this bound can sometimes further improve the test-time FPR.
This is expected because the bound is obtained by applying a union bound over all possible threshold vectors, which can be loose.
Therefore, the theorem is intended to provide a safe theoretical guideline for avoiding overfitting, rather than a tight characterization of the optimal choice of $N$ and $k$.
}

\section{Proof for \cref{thm: optimal solution}}
\label{app: Proof for Theorem 1}
Here, we provide the proof of \cref{thm: optimal solution}, which establishes the structural property of the optimal solution and offers key insights for improving the efficiency of the \textsc{OptimalFPR} algorithm.
\begin{proof}
    First, in the case of $\sum_{i \in \{j \mid G_j > 0\}} H_i \leq F$, one of the optimal solution $\bm{f}$ is
    \begin{equation}
        f_i = 
        \begin{dcases}
        0 & (G_i = 0) \\
        1 & (G_i > 0),
        \end{dcases}
    \end{equation}
    and it is clear that the claim holds.

    Otherwise, in the case of $\sum_{i \in \{j \mid G_j > 0\}} H_i > F$, from \cref{equ:H_i=0_and_I_f=1,equ:H_i>0_and_I_f=1,equ:not_I_f=1}, there exists an optimal solution $\bm{f}$ and a threshold $T \in \mathbb{R}$ that satisfy
    \begin{equation}
        H_i = 0 ~ \Rightarrow ~ f_i=1,
    \end{equation}
    \begin{equation}
        H_i > 0 ~ \land ~ \frac{G_i}{H_i} \geq T ~ \Rightarrow ~ f_i=1,
    \end{equation}
    \begin{equation}
        H_i > 0 ~ \land ~ \frac{G_i}{H_i} < T ~ \Rightarrow ~ f_i<1.
    \end{equation}
    Therefore, it is clear that the claim holds.
\end{proof}

\section{Proof for \cref{lem: dkl_monge}}
\label{app: Proof for Lemma dkl_monge}
Here, we provide the proof of \cref{lem: dkl_monge}, which implies the property of the dynamic programming in PLBF construction when the score distribution is \textit{ideal}, and which supports the performance guarantees for fast PLBF++ and fast PLBF\#.
\begin{proof}
We define $u_1$, $u_2$, $v_1$, $v_2$, $x$, and $y$ as follows:
\begin{equation}
\begin{split}
    u_1 \coloneqq \sum_{j=i'}^{p} g_j, ~~ 
    u_2 \coloneqq \sum_{j=i}^{p} g_j, ~~
    x \coloneqq \sum_{j=p+1}^{p'} g_j, \\
    v_1 \coloneqq \sum_{j=i'}^{p} h_j, ~~ 
    v_2 \coloneqq \sum_{j=i}^{p} h_j, ~~ 
    y \coloneqq \sum_{j=p+1}^{p'} h_j.
\end{split}
\end{equation}

Here, we define the function $D(x,y)$ as follows:
\begin{align}
    D(x,y) &\coloneqq d_\mathrm{KL}(i, p) + d_\mathrm{KL}(i', p') - d_\mathrm{KL}(i, p') - d_\mathrm{KL}(i', p) \\
&= u_2 \log_{2} \left(
        \frac{u_2}{v_2}
    \right) + 
    (u_1 + x) \log_{2} \left(
        \frac{u_1 + x}{v_1 + y}
    \right) \nonumber\\
    &\quad~~ - (u_2 + x) \log_{2} \left(
        \frac{u_2 + x}{v_2 + y}
    \right) - u_1 \log_{2} \left(
            \frac{u_1}{v_1}
        \right).
\end{align}

Since the score distribution is \textit{ideal}, we have $(u_2 + x)/(v_2 + y) \leq (u_1 + x)/(v_1 + y)$. Therefore,
\begin{align}
    \frac{\partial D}{\partial x} &= 
    \log_{2} \left(\frac{u_1 + x}{v_1 + y}\right) - 
    \log_{2} \left(\frac{u_2 + x}{v_2 + y}\right) \geq 0, \\
    \frac{\partial D}{\partial y} &= 
    \left(- \frac{u_1 + x}{v_1 + y} + \frac{u_2 + x}{v_2 + y}\right) \log_{2}(e) \leq 0.
\end{align}
Since the score distribution is \textit{ideal}, we also have $u_1/v_1 \leq x/y$. Thus,
\begin{align}
D(x,y) &\geq \inf_{0<x,0<y,\frac{u_1}{v_1} \leq \frac{x}{y}} D(x,y) \\
&\geq \inf_{0<x,0<y,\frac{u_1}{v_1} = \frac{x}{y}} D(x,y).
\end{align}
When $\frac{u_1}{v_1} =  \frac{x}{y}$, we can introduce a positive variable $z$ such that $x = z u_1$ and $y = z v_1$. Substituting this into the expression, we have 
\begin{align}
D(x,y) &\geq \inf_{z > 0} D(z u_1, z v_1) \\
&= \inf_{z > 0} \Biggl(
        zu_1 \log_{2} \left(\frac{zu_1}{zv_1}\right) + 
        u_2 \log_{2} \left(\frac{u_2}{v_2}\right) \nonumber\\
        &\quad\quad\quad\quad\quad - \left(zu_1 + u_2\right) \log_{2} \left(\frac{zu_1 + u_2}{zv_1 + v_2}\right)
    \Biggr).
\end{align}
By applying Jensen's inequality~\cite{jensen1906fonctions}, we obtain $D(x, y) \geq 0$.
\end{proof}

\section{Proof for \cref{thm: diff DPKL}}
\label{app: Proof for Theorem 4}
Here, we provide the proof of \cref{thm: diff DPKL}, which guarantees the performance of Fast PLBF++ even in \textit{non-ideal} situations.
\begin{proof}
$\mathrm{DP}_\mathrm{true}[p][q]$ is a DP table derived from the exact calculation of all transitions, whereas $\mathrm{DP}_\mathrm{mm}[p][q]$ is a DP table derived from calculations where some transitions are omitted.
Therefore, it is evident that $ 0 \leq \mathrm{DP}_\mathrm{true}[p][q] - \mathrm{DP}_\mathrm{mm}[p][q]$.
We will now prove the right-hand side of the inequality.

Let $J(p)$ be the smallest column index $j$ such that $A_{p,j}$ is the maximum value in the $p$-th row of $A$.
Additionally, let $\hat{J}(p)$ be the column index $j$ such that the monotone maximum algorithm considers $A_{p,j}$ to be the maximum value in the $p$-th row of $A$.
$J(p)$ and $\hat{J}(p)$ always satisfy $A_{p,J(p)} - A_{p,\hat{J}(p)} \geq 0$.

First, we prove recursively that $A_{p,J(p)} - A_{p,\hat{J}(p)} \leq \delta(p)$ for all $p=1,2,\dots,N-1$.
We introduce virtual $0$-th and $N$-th row to the matrix $A$, and we assume $J(0)=\hat{J}(0)=0$, $J(N)=\hat{J}(N)=N-1$, and $A_{0,J(0)} - A_{0,\hat{J}(0)} = A_{N,J(N)} - A_{N,\hat{J}(N)} = 0$.
\Rtwo{We prove the claim by induction on the row-processing order of the monotone maxima algorithm.
The induction hypothesis is that, for every row that has already been processed, $A_{q,J(q)} - A_{q,\hat{J}(q)} \leq \delta(q)$ holds.
This hypothesis is true for the virtual $0$-th and $N$-th rows by construction.}

\Rtwo{Let $p$ be the row currently being processed.
Let $l$ and $r$ be the largest processed row index smaller than $p$ and the smallest processed row index larger than $p$, respectively, as in the definition of $\delta(p)$.
Then, the range of indices to be searched by the monotone maxima algorithm in the $p$-th row is determined by the search results of the $l$-th and $r$-th rows.}
We denote the minimum and maximum values of this range of indices as $L(p)$ and $R(p)$, which can be expressed as follows:
\begin{equation}
    L(p) = \begin{dcases}
        1 & (l = 0) \\
        \hat{J}(l) & (\mathrm{else})
    \end{dcases}, ~~~ 
    R(p) = \begin{dcases}
        N-1 & (r = N) \\
        \hat{J}(r) & (\mathrm{else}).
    \end{dcases}
\end{equation}

First, from the definition of $\hat{J}(p)$,
\begin{equation}
\label{equ:A diif (L(p) to R(p))}
    \max_{L(p) \leq j \leq R(p)} A_{p, j} - A_{p, \hat{J}(p)} = A_{p, \hat{J}(p)} - A_{p, \hat{J}(p)} = 0.
\end{equation}

When $R(p) < p$,
\begin{align}
&~ \max_{R(p)< j \leq p} A_{p, j} - A_{p, \hat{J}(p)} \\
\leq&~ \max_{R(p)< j \leq p} A_{p, j} - A_{p, \hat{J}(r)} \\
\leq&~ \max_{R(p)< j \leq p} A_{r, j} - A_{r, \hat{J}(r)} + \Delta(p, r) \\
\leq&~ A_{r, J(r)} - A_{r, \hat{J}(r)} + \Delta(p, r) \\
\label{equ:A diif (R(p) to p)}
\leq&~ \delta(r) + \Delta(p, r).
\end{align}
Here, the first and third inequality is due to
\begin{equation}
    A_{p, \hat{J}(p)} = \max_{\hat{J}(l)\leq j \leq \hat{J}(r)} A_{p, j} \geq A_{p, \hat{J}(r)}
\end{equation}
and
\begin{equation}
    A_{r, J(r)} = \max_{1\leq j \leq r} A_{r, j} \geq \max_{R(p)< j \leq p} A_{r, j}.
\end{equation}
The fourth inequality is based on the induction hypothesis, since $r$ has already been processed before $p$.
The second inequality can be shown as follows, using the definition of $A$ and $\Delta$:
\begin{align}
&~ A_{p, j} - A_{p, \hat{J}(r)} \\
=&~ \left(\mathrm{DP}_\mathrm{KL}[j-1][q-1] + d_\mathrm{KL}(j, p) \right) \nonumber\\
&\quad - \left( \mathrm{DP}_\mathrm{KL}[\hat{J}(r)-1][q-1] + d_\mathrm{KL}(\hat{J}(r), p)\right) \\
\leq&~ \left(\mathrm{DP}_\mathrm{KL}[j-1][q-1] + d_\mathrm{KL}(j, r) \right) \nonumber\\
&\quad - \left( \mathrm{DP}_\mathrm{KL}[\hat{J}(r)-1][q-1] + d_\mathrm{KL}(\hat{J}(r), r)\right) + \Delta(p, r) \\
=&~ A_{r, j} - A_{r, \hat{J}(r)} + \Delta(p, r).
\end{align}
In the same way, we have
\begin{equation}
\label{equ:A diif (1 to L(p))}
    \max_{1\leq j < L(p)} A_{p, j} - A_{p, \hat{J}(p)} \leq \delta(l) + \Delta(l, p).
\end{equation}
Therefore, from \cref{equ:A diif (1 to L(p)),equ:A diif (L(p) to R(p)),equ:A diif (R(p) to p)},
\begin{align}
&~ A_{p, J(p)} - A_{p, \hat{J}(p)} \\
=&~ \max_{1\leq j \leq p}A_{p,j} -  A_{p, \hat{J}(p)} \\
\leq&~ \max\left(\delta(l) + \Delta(l, p), ~ 0, ~ \delta(r) + \Delta(p, r)\right) \\
=&~ \delta(p).
\end{align}
\Rtwo{Thus, we have $A_{p,J(p)} - A_{p,\hat{J}(p)} \leq \delta(p)$ for all $p \in \{1,2,\dots,N-1\}$.
This means that, in each DP layer, the maximum value returned by the monotone maxima algorithm is at most $\delta_\mathrm{max}$ smaller than the true row maximum.
Since this error can increase by at most $\delta_\mathrm{max}$ at each DP layer, the maximum error in $\mathrm{DP}_\mathrm{KL}[p][q]$ is at most $q\delta_\mathrm{max}$.}
\end{proof}

\begin{acknowledgements}
    This work was supported by JST AIP Acceleration Research JPMJCR23U2, Japan, and by JSPS KAKENHI (Grant Number 25KJ0754).
    We appreciate the valuable feedback from the anonymous reviewers of NeurIPS, which has improved the quality of our paper and led to the development of new methods and analyses.
\end{acknowledgements}

\bibliographystyle{spmpsci}
\bibliography{list}

\end{document}